\documentclass[10pt]{article}
\usepackage{amsmath}
\usepackage{graphicx,psfrag,epsf}
\usepackage{enumerate}
\usepackage{natbib}
\usepackage{caption}
\usepackage{fullpage}
\usepackage[sectionbib]{bibunits}
\usepackage{url} 

\usepackage[utf8]{inputenc}
\usepackage{amsfonts}
\usepackage{wrapfig}
\usepackage{graphicx, xcolor, subcaption, float, booktabs, array, setspace} 
\newcommand{\tabitem}{~~\llap{\textbullet}~~}

\usepackage{algorithm}
\usepackage{algpseudocode}
\usepackage{abstract}

\usepackage{amsmath,amssymb,mathtools,amsthm, bm, bbm}
\usepackage[shortlabels]{enumitem}
\usepackage{hyperref}

\hypersetup{
    colorlinks=true,
    linkcolor=blue, 
    citecolor=blue, 
}

\newtheorem{thm}{Theorem}[section]
\newtheorem{cor}{Corollary}[section]
\newtheorem{lemma}[thm]{Lemma}
\newtheorem{prop}[thm]{Proposition}

\newtheorem{assumption}{Assumption}

\newtheorem*{remark}{Remark}
\newcommand{\PP}{\mathbb{P}}
\newcommand{\Var}{\text{Var}}

\newcommand{\iid}{\overset{i.i.d.}{\sim}}
\DeclareMathOperator*{\argmin}{arg\,min}
\DeclareMathOperator*{\argmax}{arg\,max}

\definecolor{ps}{RGB}{0,0,200}

\newcommand{\mcl}{\mathcal}
\newcommand{\mcn}{\mathcal{N}}

\newcommand{\mbb}{\mathbb}

\newcommand{\eps}{\epsilon}

\newcommand{\df}{\textnormal{df}}
\newcommand{\op}{\textnormal{op}}
\newcommand{\rmL}{\textrm{L}}
\newcommand{\rmR}{\textrm{R}}

\newcommand{\ba}{\bm{a}}
\newcommand{\bb}{\bm{b}}

\newcommand{\bg}{\bm{g}}
\newcommand{\bh}{\bm{h}}

\newcommand{\bt}{\bm{t}}
\newcommand{\bu}{\bm{u}}
\newcommand{\bv}{\bm{v}}
\newcommand{\bw}{\bm{w}}
\newcommand{\bx}{\bm{x}}

\newcommand{\by}{\bm{y}}

\newcommand{\bz}{\bm{z}}

\newcommand{\bA}{\bm{A}}
\newcommand{\bB}{\bm{B}}

\newcommand{\bG}{\bm{G}}

\newcommand{\bI}{\bm{I}}

\newcommand{\bS}{\bm{S}}
\newcommand{\bT}{\bm{T}}
\newcommand{\bX}{\bm{X}}

\newcommand{\bZ}{\bm{Z}}

\newcommand{\bbeta}{\bm{\beta}}

\newcommand{\bep}{\bm{\epsilon}}

\newcommand{\bphi}{\bm{\phi}}
\newcommand{\bxi}{\bm{\xi}}

\newcommand{\bSigma}{\bm{\Sigma}}

\newcommand{\btheta}{\bm{\theta}}

\newcommand{\R}{\mbb{R}}

\newcommand{\E}{\mbb{E}}

\newcommand{\V}{V}

\newcommand{\convP}{\overset{P}{\rightarrow}}

\newcommand{\suplk}{\sup_{\lambda_k\in[\lambda_{\min},\lambda_{\max}]}}
\newcommand{\supl}{\sup_{\lambda\in[\lambda_{\min},\lambda_{\max}]}}
\newcommand{\supll}{\sup_{\lambda_{\textrm{L}}\in[\lambda_{\min},\lambda_{\max}]}}
\newcommand{\suplr}{\sup_{\lambda_{\textrm{R}}\in[\lambda_{\min},\lambda_{\max}]}}
\newcommand{\suplb}{\sup_{\lambda_{\textrm{L}},\lambda_{\textrm{R}}\in[\lambda_{\min},\lambda_{\max}]}}

\definecolor{ysr}{RGB}{0,200,200}

\begin{document}


\pagenumbering{gobble}
{
  \title{ HEDE: 
  Heritability estimation in high dimensions \\ by Ensembling Debiased Estimators}

  \author{Yanke Song, Xihong Lin, and Pragya Sur\thanks{Yanke Song is a graduate student ({\em  ysong@g.harvard.edu}) in the Department of Statistics at Harvard University. Xihong Lin  ({\em xlin@hsph.harvard.edu}) is Professor of Biostatistics at Harvard T.H. Chan School of Public Health and Professor of Statistics at Harvard University. Pragya Sur is Assistant Professor ({\em  
   pragya@fas.harvard.edu}) in the Department of Statistics at Harvard University.  Lin's work was supported by the National Institutes of Health grants R35-CA197449, R01-HL163560, U01-HG012064, U19-CA203654,  and P30-ES000002. Song and Sur's work was supported by NSF DMS-2113426 and a Dean's Competitive Fund for Promising Scholarship. Corresponding authors are Song and Sur.}\hspace{.2cm}\\
    Harvard University}
  \maketitle
} 

\bigskip
\begin{center}
 {\bf Abstract}
 \end{center}
 \vspace{-0.6in}
\begin{abstract}
Estimating heritability remains a significant challenge in statistical genetics. Diverse approaches have emerged over the years that are broadly categorized as either random effects or fixed effects heritability methods. In this work, we focus on the latter. We propose HEDE, an ensemble approach to estimate heritability or the signal-to-noise ratio in high-dimensional linear models where the sample size and the dimension grow proportionally. Our method ensembles post-processed versions of the debiased lasso and debiased ridge estimators, and incorporates a data-driven strategy for hyperparameter selection that significantly boosts estimation performance. We establish rigorous consistency guarantees that hold despite adaptive tuning. Extensive simulations demonstrate our method's superiority over existing state-of-the-art methods across various signal structures and genetic architectures, ranging from sparse to relatively dense and from evenly to unevenly distributed signals. Furthermore, we discuss the advantages of fixed effects heritability estimation compared to random effects estimation. Our theoretical guarantees hold for realistic genotype distributions observed in genetic studies, where genotypes typically take on discrete values and are often well-modeled by sub-Gaussian distributed random variables. We establish our theoretical results by deriving uniform bounds, built upon the convex Gaussian min-max theorem, and leveraging universality results. Finally, we showcase the efficacy of our approach in estimating height and BMI heritability using the UK Biobank.
\end{abstract}

\noindent%


\pagenumbering{arabic}
\section{Introduction}\label{sec:introduction}

Accurate heritability estimation presents a significant challenge in statistical genetics. Distinguishing the contributions of genetic versus environmental factors is crucial for understanding how our genetic makeup influences the development of complex diseases.
Importantly, genetic research over the past decade has highlighted a significant discrepancy between heritability estimates from twin studies and those from Genome-Wide Association Studies (GWAS) \citep{manolio2009finding, yang2015genetic}
This discrepancy, known as the ``missing heritability problem", underscores the importance of developing reliable heritability estimation methods   to understand the extent of missing heritability.
In this paper, we propose an ensemble method for heritability estimation,   using high-dimensional Single Nucleotide Polymorphisms (SNPs) commonly collected in GWAS,  that exhibits superior performance across a wide variety of settings and significantly outperforms existing methods in practical scenarios of interest. We focus on the narrow-sense heritability \cite{yang2011gcta}, which refers to the proportion of phenotypic variance explained by additive genotypic effects. For conciseness, we henceforth refer to this simply as the heritability.

GWAS data typically contains more SNPs than individuals. This presents a
challenge when estimating heritability. Addressing this challenge necessitates modeling assumptions. Previous work in this area can be broadly categorized into  two classes based on these assumptions: random effects models and fixed effects models.

  Random effects models treat the underlying regression coefficients as random variables and have been widely used for heritability estimation. Arguably, the most popular approach in this vein is GCTA, introduced by \cite{yang2011gcta}. GCTA employs the classical restricted maximum likelihood method (REML) (also called genomic REML/GREML in their work) within a linear mixed effects model.   It assumes that the regression coefficients are 
i.i.d. draws from a  Gaussian distribution with zero mean and a constant variance. 
 Since the  genetic effect sizes/regression coefficients might depend on the design matrix (the genotypes), subsequent work provides various improvements that relax this assumption, including GREML-MS \citep{lee2013estimation}, GREML-LDMS \citep{yang2015genetic}, BOLT-LMM \citep{loh2015efficient}, among others.  Specifically, GREML-MS allows the signal variance to depend on minor allele frequencies (MAFs), while GREML-LDMS allows it to depende on linkage disequilibrium (LD) levels/quantiles.
Prior statistical literature has studied the theoretical properties of several random effects methods 
\citep{jiang2016high, bonnet2015heritability, dicker2016maximum, hu2022misspecification}. 

Despite this rich literature, a major limitation of random effects model based heritability estimation method is that its consistency often 
relies on the correct parametric assumptions of regression coefficients. Specifically, these methods assume that regression coefficients are either i.i.d, or depend on the design matrix in  specific ways.
When these assumptions are violated, random effects methods can produce significantly biased heritablity estimates. The issue is magnified when SNPs are in linkage disequilibrium (LD) or when there is additional heterogeneity across the  LD  blocks not captured by the random effects assumptions. We elaborate this issue further in Figure \ref{fig:attack}, where we demonstrate that  GREML, GREML-MS and GREML-LDMS all incur significant biases when their respective random effects assumptions are violated. 

These challenges associated with random effects methods have spurred a distinct strand of methods by positing regression coefficients to be fixed while treating the design matrix to be random. Referred to as fixed effects methods, these approaches often exhibit greater robustness to diverse realizations of the underlying signal without the need to make parametric assumptions on regression coefficients (Figure \ref{fig:attack}). Consequently, our focus lies on estimating the heritability within this fixed effects model framework. 

Several fixed effects methods have emerged for heritability estimation. However, a substantial body of the proposed work requires the underlying signal to be sparse \citep{fan2012variance, sun2012scaled, tony2020semisupervised, guo2019optimal}, and they are subject to loss of efficiency  when these sparsity assumptions are violated \citep{dicker2014variance,janson2017eigenprism,bayati2013estimating}, see Sections 5.2, 3.3, 2, respectively. 
The methods proposed by several authors \citep{chen2022statistical,dicker2014variance,janson2017eigenprism} operate under less stringent assumptions regarding the signal, but they depend on ridge-regression-analogous concepts. Therefore they are expected to have good efficiency when the signal is dense, and are subject to loss of efficiency when signals are sparse. We illustrate these issues further in Section \ref{section:simulations_fixed}.
In practice, the underlying genetic architecture for a given trait is unknown, i.e., the signal could be sparse or dense.  
Hence, it is desirable to develop a robust method  across a broad spectrum of signal regimes from sparse to dense. In addition, fixed effects methods often rely on assuming the covariate distribution to be Gaussian  \citep{bayati2013estimating, dicker2014variance, janson2017eigenprism, verzelen2018adaptive}. This assumption fails to capture genetic data where genotypes are discrete, taking values in ${0,1,2}$.

To address these issues,  
we propose a robust ensemble heritability estimation method HEDE (Heritability Estimation via Debiased Ensemble), which ensembles the debiased Lasso and ridge estimators, with degrees-of-freedom corrections \citep{bellec2019biasing, javanmard2014hypothesis}.
Debiased estimators, originally proposed to correct the bias for Lasso/ridge estimators \citep{zhang2014confidence, van2014asymptotically, javanmard2014confidence}, have been employed for tasks such as confidence interval construction and hypothesis testing, and have nice properties such as optimal testing properties \cite{javanmard2014confidence}.  Unlike these methods, HEDE 
utilizes debiasing for heritability estimation.
To develop HEDE, we first demonstrate that a reliable estimator for heritability can be formulated using any linear combination of these debiased estimators, in which we devise a consistent bias correction strategy. Subsequently, we refine the linear combination based on this bias correction strategy. This methodology enables the consistent estimation of heritability using any linear combination of the debiased Lasso and the debiased Ridge. HEDE then uses a data-driven approach to selecting an optimal ensemble from the array of possible linear combinations. Finally, we introduce an adaptive method for determining the underlying regularization parameters. We develop both adaptive strategies—ensemble selection and tuning parameter selection for the Lasso and the Ridge—by minimizing a certain mean square error. This results in enhanced statistical performance for heritability estimation.

Underpinning our method is a rigorous statistical theory that ensures consistency under minimal assumptions. Specifically, 
we operate under sub-Gaussian assumptions, which are consistent with a broad class of designs observed in GWAS that motivates this paper.
We show that the consistency of heritability estimation using HEDE guarantees remain intact despite employing adaptive strategies for selecting the tuning parameters and the debiased Lasso and ridge ensemble, all while accommodating sub-Gaussian designs.

We perform extensive simulation studies to evaluate the finite sample performance of HEDE compared with the existing methods in  a range of settings mimicking GWAS data.
We demonstrate  that HEDE overcomes the limitations in random effects methods, remaining unbiased across a wide variety of signal structures, specifically signals with different amplitudes across various combinations of MAF and LD levels, as well as LD blocks. 
our simulation also illustrates that HEDE consistently achieves the lowest mean square error when compared to popular fixed effects methods across a wide spectrum of sparse and dense signals, as well as low and high heritability values.  
We compare HEDE's performance with both random and fixed effects methods on the problems of estimating height and BMI heritability using the UK Biobank data.

The rest of this paper is organized as follows. Section \ref{section:setup} introduces the  problem setup. Section \ref{section:methodology} presents our HEDE while Section \ref{section:mathematical} discusses its theoretical properties. Section \ref{section:simulations} presents extensive simulation studies while Section \ref{section:real_data} complements these numerical investigations with a real-data application on the UK Biobank data \citep{sudlow2015uk}. Finally, Section \ref{section:discussion} concludes with discussions. 

\section{Problem Setup}\label{section:setup}
 In this section, we formally introduce our problem setup. We consider a high-dimensional regime where the covariate dimension grows with the sample size. To elaborate, we consider a sequence of problem instances $\{\by(n),\bX(n),\bbeta(n),\bep(n)\}_{n \geq 1}$ 
that satisfies a linear model
\begin{equation}\label{model:linear}
 \by(n) = \bX(n) \bbeta(n) + \bep(n),
\end{equation}
where $\by(n) =(y_1,\hdots,y_n)^{\top}\in \R^n$ represents observations of a real-valued phenotype of interest, such as height, BMI, blood pressure, for $n$ unrelated individuals from a population. The unknown regression coefficient vector, $\bbeta(n)=(\beta_1,\hdots,\beta_{p(n)})^{\top} \in \R^{p(n)}$, captures the effects of $p(n)$ SNPs and is assumed to be fixed parameters. The environmental noise, $\bep(n)=(\epsilon_1,\hdots,\epsilon_n) \in \R^n$, is independent of $\bX(n)$, with each entry also independent. Lastly, $\bX(n) \in \R^{n \times p(n)}$ signifies the matrix of observed genotypes with i.i.d.~rows $\bx_{i\bullet}$, which are normalized as follows
\begin{equation}
\label{eqn:X_normalization}
X_{ij} = \frac{G_{ij} - 2\bar{G}_j}{\sqrt{2{\bar{G}_j(1-\bar{G}_j)}}},
\end{equation}
where $G_{ij} = 0,1,2$ is the genotype containing the minor allele count (the number of copies of the less frequent allele) at SNP $j$ for individual $i$ and $\bar{G}_j=\sum_{i=1}^nG_{ij}/n$
is the minor allele frequency of SNP $j$ in the sample. 

We assume the covariates have sub-Gaussian distributions. Since $\bX$ is random and $\bbeta$ is deterministic, we have placed ourselves in the fixed effects framework. Moving forward, we suppress the dependence on $n$ when the meaning is clear from the context. Our parameter of interest is the heritability, which is the proportion of variance in $y_i$ explained by features in $\bm{x}_{i\bullet}$ defined as
\begin{equation}\label{def:heritability}
h^2 = 
\lim_{n,p \rightarrow\infty, \,\, p/n \rightarrow \delta}\frac{\Var(\bx_{i\bullet}^{\top}\bbeta)}{\Var(y_i)}.
\end{equation}

We emphasize the following remark regarding the aforementioned definition. We assume that $n$ and $p$ diverge with $p/n\rightarrow \delta > 0$. This setting--also known as the proportional asymptotics regime--has gained significant recent attention in high-dimensional statistics. A major advantage of this regime is that theoretical results derived under such asymptotics demonstrate attractive finite sample performance \citep{sur2019likelihood, sur2019modern,candes2020phase,zhao2022asymptotic,liang2022precise,jiang2022new}. 
Furthermore, prior research in heritability estimation \citep{jiang2016high, dicker2016maximum,janson2017eigenprism} utilized this framewrok to characterize theoretical properties of GREML and its variants. 
In practical situations, we observe a single value for $n$ and $p$, so we may substitute the corresponding ratio $p/n$ into our theory in place of $\delta$ to execute data analyses. Formal assumptions on the covariate and noise distributions, and the signal are deferred to Section \ref{section:mathematical}. 

\section{Methodology of HEDE: Heritability estimation Ensembling Debiased Estimators }\label{section:methodology}
  We begin by describing the overarching theme underlying our method. 
For simplicity, we discuss the case where $\bX$ has independent columns, deferring the extension to correlated columns to Section \ref{sec:cov}.
The denominator of $h^2$ in (\ref{def:heritability}) can be consistently estimated using the sample variance of $\by$. The numerator of $h^2$ in (\ref{def:heritability})   simplifies to $\|\bbeta\|_2^2$ in the setting of independent columns. Consequently, we seek to develop an accurate estimator for this norm. 

To this end, we turn to the debiased Lasso estimator and the debiased ridge estimator. Denote the Lasso and the ridge by $\hat\bbeta_{\rmL}$ and $\hat\bbeta_{\rmR}$ respectively. 
These penalized estimators are known to shrink estimated coefficients to $0$, therefore inducing bias. The existing literature provides debiased versions of these estimators, denoted as $\hat\bbeta_{\rmL}^d$ and $\hat\bbeta_{\rmR}^d$. These debiased estimators roughly satisfy the following property: $\hat\bbeta_{\rmL}^d \approx \bbeta + \sigma_{\rmL} Z_1$, $\hat\bbeta_{\rmR}^d \approx \bbeta +  \sigma_{\rmR} Z_2$ meaning that each entry of $\hat\bbeta_{\rmL}^d, \hat\bbeta_{\rmR}^d$ is approximately centered around the corresponding true signal coefficient with Gaussian fluctuations and variances given by $\sigma_L,\sigma_R$ respectively.  By taking norms on both sides, we obtain that $\|\hat\bbeta_k^d\|_2^2 - p\cdot \sigma_k^2 \approx \|\bbeta\|_2^2$, where $k=\rmL \,\, \text{or} \,\, \rmR$ depending on whether we consider the Lasso or the ridge. Thus, precise estimates for the variances $\sigma_k^2$'s produce accurate estimates  for $h^2$ starting from each debiased estimator.

This intuition extends to any linear combination of the debiased Lasso and ridge in the form of $\tilde\bbeta_C^d:=\alpha_{\textrm{L}} \hat \bbeta_{\textrm{L}}^d + (1 - \alpha_{\textrm{L}})\hat\bbeta_{\textrm{R}}^d$. A similar relationship holds for the norm: $\|\tilde\bbeta_C^d\|_2^2 -  p\cdot\tilde\tau_C^2 \approx   \|\bbeta \|_2^2$. We establish that for any $\alpha_L \in [0,1]$, we can derive a consistent estimator for the ensemble variance $\tilde\tau_C^2$. Among all possible ensemble choices with different ensembling parameter $\alpha_{\rmL}$ and different tuning parameters $\lambda_{\rmL},\lambda_{\rmR}$ for calculating the Lasso and the ridge, we adaptively choose the combination that minimizes the mean square error of the ensemble estimator.
This method, dubbed HEDE, harnesses the strengths of $L_1$ and $L_2$ regression and demonstrates strong practical performance, since we design our adaptive tuning process to particularly enhance statistical accuracy. We provide further elaboration on HEDE  in the subsequent subsections: Section \ref{debiasingfirstsection}, discusses the specific forms of the debiased estimators utilized in our framework; Section \ref{subsec:ensembling} describes our ensembling approach, and Section \ref{sec:ourmethod} discusses the final HEDE method with our adaptive tuning strategy.

\subsection{{Debiasing Regularized Estimators}}\label{debiasingfirstsection}
Define the Lasso and ridge estimators as follows
\begin{equation}\label{method:estimators_original}
\begin{aligned}
    \hat{\bbeta}_{\textrm{L}} := \argmin_{\bb} \frac{1}{2n} \|\by - \bX \bb\|_2^2 + \frac{\lambda_{\textrm{L}}}{\sqrt{n}} \|\bb\|_1, \quad 
    \hat{\bbeta}_{\textrm{R}} := \argmin_{\bb} \frac{1}{2n} \|\by - \bX \bb\|_2^2 + \frac{\lambda_{\textrm{R}}}{2} \|\bb\|_2^2,
\end{aligned}
\end{equation}
where $\lambda_{\textrm{L}}$ and $\lambda_{\textrm{R}}$ denote the respective tuning parameters. We remark that the different scaling in the regularization terms ensure comparable scales for the solutions in our high-dimensional regime (this is easy to see by noting that if each entry of $\bb$ satisfies  $b_j \sim 1/\sqrt{n}$, then $\|\bb\|_1 \sim \sqrt{n}$ whereas $\|\bb\|_2^2\sim 1$). These estimators incur a regularization bias. To tackle this issue, several authors proposed debiased versions of these estimators that remain asymptotically unbiased with Gaussian fluctuations \citep{zhang2014confidence,van2014asymptotically,javanmard2014confidence,bellec2019biasing,bellec2023debiasing}.  

Specifically, we work with the following   debiased estimators:
\begin{equation}\label{method:estimators_debiased}
\begin{aligned}
    \hat{\bbeta}_{\textrm{L}}^d = \hat{\bbeta}_{\textrm{L}} + \frac{\bX^\top (\by - \bX \hat{\bbeta}_{\textrm{L}})}{n - \|\hat{\bbeta}_{\textrm{L}}\|_0}, \, \quad \,
    \hat{\bbeta}_{\textrm{R}}^d = \hat{\bbeta}_{\textrm{R}} + \frac{\bX^\top (\by - \bX \hat{\bbeta}_{\textrm{R}})}{n - \text{Tr}((\frac{1}{n} \bX^\top \bX + \lambda_{\textrm{R}} \bm{I})^{-1}\frac{1}{n} \bX^\top \bX)}.
\end{aligned}
\end{equation}
Detailed in \cite{bellec2019biasing}, these versions are known as debiasing with degrees-of-freedom (DOF) correction, since the second term in the denominator is precisely the degrees-of-freedom of the regularized estimator. Prior debiasing theory \citep{bellec2019second,celentano2021cad,bellec2022observable} roughly established that under suitable conditions, for each coordinate $j$, the debiased estimators jointly satisfy (asymptotically)  

\begin{equation}\label{eqn:approx_debiased_true}
\begin{pmatrix}
\hat\beta_{\rmL,j}^d-\beta_j\\
\hat\beta_{\rmR,j}-\beta_j
\end{pmatrix} \approx \mcn\left(
\begin{pmatrix}
0\\
0
\end{pmatrix},
\begin{pmatrix}
\tau_{\rmL}^2 & \tau_{\rmL\rmR}\\
\tau_{\rmL\rmR} & \tau_{\rmR}^2\\
\end{pmatrix}\right),
\end{equation}
for appropriate variance and covariance parameters $\tau_\rmL^2,\tau_\rmR^2,\tau_{\rmL\rmR}$. We remark that this is a non-rigorous statement (hence the ``$\approx$" sign, and a precise statement is presented in Theorem \ref{main:thm:debiased_uniform}.
Additionally, these parameters can be consistently estimated using the following (\cite{celentano2021cad}):
\begin{equation}\label{main:eqn:tau_estimation}
\begin{aligned}
    \hat{\tau}_{\textrm{L}}^2 = \frac{ \|\by - \bX \hat{\bbeta}_{\textrm{L}}\|_2^2}{(n - \|\hat{\bbeta}_{\textrm{L}}\|_0)^2}, & \qquad
    \hat{\tau}_{\textrm{R}}^2 = \frac{ \|\by - \bX \hat{\bbeta}_{\textrm{R}}\|_2^2}{(n - \operatorname{Tr}((\frac{1}{n} \bX^\top \bX + \lambda_{\textrm{R}} \bm{I})^{-1}\frac{1}{n} \bX^\top \bX))^2},\\
    \hat{\tau}_{\rmL\rmR} &= \frac{ <\by - \bX \hat{\bbeta}_{\textrm{L}},\by - \bX \hat{\bbeta}_{\textrm{R}}>}{(n - \|\hat{\bbeta}_{\textrm{L}}\|_0)(n-\operatorname{Tr}((\frac{1}{n} \bX^\top \bX + \lambda_{\textrm{R}} \bm{I})^{-1}\frac{1}{n} \bX^\top \bX)}.
\end{aligned}
\end{equation}
This serves as the basis of our HEDE heritability estimator. 

\subsection{{ Ensembling and a Primitive Heritability Estimator}}\label{subsec:ensembling}
We seek to combine the strengths of $L_1$ and $L_2$ regression estimators. To achieve this, we consider linear combinations of the debiased Lasso and Ridge, taking the form  
\begin{equation}\label{eq:comb}
    \tilde\bbeta_C^d = \alpha_{\textrm{L}} \hat \bbeta_{\textrm{L}}^d + (1 - \alpha_{\textrm{L}})\hat\bbeta_{\textrm{R}}^d,
\end{equation}
where $\alpha_{\rmL} \in [0,1]$ is held fixed for the time being.  We remark that the well-known ElasticNet estimator \citep{zou2005regularization} also combines the strengths of the Lasso and the ridge. Though our theoretical framework allows potential extensions to the ElasticNet, we do not pursue this direction due to computational concerns, discussed further in the end of Section \ref{sec:cov}.

Combining \eqref{eqn:approx_debiased_true} and \eqref{main:eqn:tau_estimation}, we obtain that each coordinate of $\tilde\bbeta_C^d $ satisfies, roughly, 
\begin{equation}\label{eq:tildetauCdef}
\frac{\tilde\beta_{C,j}^d -\beta_j}{\sqrt{\tilde\tau_C^2}} \approx \mathcal{N}(0,1) \quad \text{with} \quad \tilde\tau_C^2 := \alpha_{\rmL}^2 \hat\tau_{\rmL}^2  + (1-\alpha_{\rmL})^2\hat\tau_{\rmR}^2   + 2 \alpha_{\rmL}(1-\alpha_{\rmL}) \hat\tau_{\rmL\rmR}.
\end{equation} 
Furthermore, the Gaussianity property \eqref{eq:tildetauCdef} holds on an ``average" sense as well, where the result can be effectively summed across all coordinates. We describe the precise sense next. Specifically,  for any Lipschitz function $\psi$, the difference  $\sum_{j=1}^p\psi(\tilde\beta_{C,j}^d)-\sum_{j=1}^p\psi(\beta_j)$ behaves roughly as $pE\{\psi(\sqrt{\tilde{\tau}_C^2} Z)\}$, where $Z \sim \mathcal{N}(0,1)$.   On setting $\psi(x) = x^2$ with $x$ being bounded, we obtain
\begin{equation}\label{eq:tauCdef}
\|\tilde\bbeta_C^d\|_2^2 - \|\bbeta \|_2^2 - p\cdot\tilde\tau_C^2  \approx 0,
\end{equation}
where from \eqref{eq:comb},
$\| \tilde\bbeta_C^d\|_2^2 = \alpha_{\rmL}^2\|\hat \bbeta_{\textrm{L}}^d \|_2^2 + (1-\alpha_{\rmL})^2\|\hat\bbeta_{\textrm{R}}^d\|_2^2+2 \alpha_{\rmL}(1-\alpha_{\rmL}) \langle \hat \bbeta_{\textrm{L}}^d , \hat\bbeta_{\textrm{R}}^d\rangle.$
Synthesizing these arguments, we observe that
\begin{equation}\label{eq:hgeneral}
\hat{h}^2_{\alpha_{\rmL},\lambda_{\rmL},\lambda_{\rmR}}
 = \min\left [1,\max\left\{0,\frac{\|\tilde\bbeta_C^d\|_2^2 - p\cdot\tilde{\tau}_C^2}{\widehat{\Var}(\by)}\right\}\right ] 
 \end{equation}
yields a consistent estimator of $h^2$. The clipping is crucial for ensuring that $h^2$ falls within the range [0,1], which is necessary considering $h^2$ represents a proportion. 
Given a fixed $\alpha\in[0,1]$ and $\lambda_{\rmL},\lambda_{\rmR}$ within a suitable range $[\lambda_\text{min},\lambda_\text{max}]$, we can establish that   $\hat{h}^2_{\alpha_{\rmL},\lambda_{\rmL},\lambda_{\rmR}}$ is consistent for $h^2$. 

For our final algorithm, we set a range $[t_{\min},t_{\max}]$ on the the degrees-of-freedom (the denominators in \eqref{method:estimators_debiased}), and filter out all $(\lambda_{\rmL},\lambda_{\rmR})$ values whose resulting degrees-of-freedom falls outside this range. Details on how we set the range is described in  Supplementary Materials \ref{sec:thresholding}.
This results in a class of consistent estimators for the heritability, each corresponding to a different choice of $\alpha_L,\lambda_L,\lambda_R$. Among these, we propose to select the estimator that minimizes the mean square error in estimating the signal $\bbeta$.  This leads to our adaptive tuning strategy, which we describe in the next subsection. 

\subsection{{The method HEDE}}\label{sec:ourmethod}
In order to develop our method, HEDE, we take advantage of the flexibility  provided by $\alpha_{\rmL},\lambda_{\rmL},\lambda_{\rmR}$ as adjustable hyperparameters. 
Initially, we fix $\lambda_{\rmL},\lambda_{\rmR}$ and adaptively select $\alpha_{\rmL}$. To this end, we minimize the estimated variability in the ensemble estimator, given by $\tilde{\tau}_C$ \eqref{eq:tildetauCdef}, for all possible choices of $\alpha_L$.
This results in a choice that consistently estimates each coordinate of $\bbeta$ with the least possible variability among all ensemble estimators of the form \eqref{eq:comb}.
As this holds true for every coordinate and $\tilde{\beta}_{C,j}$ is asymptotically unbiased for $\beta_j$, minimizing this variance concurrently minimizes the mean squared error. Formally,
minimizing \eqref{eq:tildetauCdef} as a function of $\alpha_{\rmL}$ leads to the following data-dependent choice:
\begin{equation}\label{main:thm_eqn:alpha_optimal}
    \hat\alpha_{\textrm{L}} = \max\left\{0,\min\left\{1,\frac{\hat\tau_{\textrm{R}}^2 - \hat\tau_{\rmL\rmR}}{\hat\tau_{\textrm{R}}^2 - 2\hat\tau_{\rmL\rmR}+\hat\tau_{\textrm{L}}^2}\right\}\right\}.
\end{equation}

Once again, the maximum and minimum operations ensure that $\hat\alpha_{\rmL} \in [0,1]$. Subsequently, by plugging in $\hat\alpha_{\rmL}$ into the ensembling formula \eqref{eq:comb}, we obtain the estimator
\begin{equation}\label{method:estimator_ensemble}
    \hat{\bbeta}_C^d = \hat\alpha_{\textrm{L}}\hat{\bbeta}_{\textrm{L}}^d + (1-\hat\alpha_{\textrm{L}})\hat{\bbeta}_{\textrm{R}}^d.
\end{equation}

For any fixed tuning parameter pair $(\lambda_{\rmL},\lambda_{\rmR})$, $ \hat{\bbeta}_C^d$ achieves the lowest mean squared error (MSE) for estimating each coordinate, among all possible ensembled estimators of the form \eqref{eq:comb}. Further, the mean square error of this estimator is given by 
\begin{equation}\label{main:thm_eqn:tau_alpha_optimal}
    \hat\tau_C^2 = \hat\alpha_{\textrm{L}}^2 \hat\tau_{\textrm{L}}^2 + 2\hat\alpha_{\textrm{L}}(1-\hat\alpha_{\textrm{L}})\hat\tau_{\rmL\rmR} + (1-\hat\alpha_{\textrm{L}})^2 \hat\tau_{\textrm{R}}^2.
\end{equation}

Note that $\hat\tau_C^2$ is $\tilde\tau_C^2$ evaluated at $\alpha_{\rmL}=\hat\alpha_{\rmL}$. Observe that both $\hat{\bbeta}_C^d $ and $\hat{\tau}^2_C$ are functions of the initially selected tuning parameter pair $(\lambda_{\rmL},\lambda_{\rmR})$. Therefore, to construct a competitive heritability estimator, we minimize $\hat{\tau}_C^2$ over a range of tuning parameter values, compute the corresponding $\hat{\bbeta}_C^d $, and construct a heritability estimator according to \eqref{eq:hgeneral} using this $\hat{\bbeta}_C^d $  and the minimum $\hat{\tau}_C^2$ (in place of $\tilde\bbeta_C^d$ and $\tilde{\tau}_C^2$, respectively). This yields our final estimator, HEDE.

\subsection{Correlated Designs}\label{sec:cov}
In the previous subsections, we discussed our method in the setting of independent covariates. In practical situations, covariates are typically correlated, i.e., SNPs are in LD in GWAS. Thus it is more realistic to assume that each row $\bx_{i\bullet}$ of our design matrix takes the form ${\bz}_{i\bullet}\bSigma^{1/2}$, where ${\bz}_{i\bullet}$ has independent entries and $\bSigma$ denotes the covariance of $\bx_{i\bullet}$.
Hence, if we can develop an accurate estimate $\hat{\bSigma
}$ for the covariance $\bSigma$, $\bx_{i\bullet}\hat{\bSigma}^{-1/2}$ should have roughly independent entries. Alternately, the whitened design matrix $\bX \hat{\bSigma}^{-1/2}$ should have roughly identity covariance and our methodology in the prior section would apply directly. 

Estimation of high-dimensional covariance matrices  poses challenges, and has been a rich area of research. Here, we consider a  special class of covariance matrices.  Our motivation stems from GWAS data, where chromosomes constitute approximately tens of thousands of independent LD blocks \citep{berisa2016approximately}, where the number of SNPs in each LD block is much smaller in size compared to the total number of SNPs. To reflect this, we assume  the population covariance matrix is block-diagonal. If the size of each block is negligible compared to the total sample size, the block-wise sample covariance matrix estimates the corresponding block of the population covariance matrix well.   We establish in Supplementary Materials, Proposition \ref{examples_covariance_estimation} that stitching together these block-wise sample covariances provides an accurate estimate of the population covariance matrix in the operator norm, i.e., $\|\hat\bSigma-\bSigma\|_{\op}\rightarrow 0$. The block diagonal approximation suffices for our purposes (see Section \ref{section:real_data}). 

Other scenarios exist where accurate estimates for the population covariance matrix is available, but we will not discuss these further. 
Our final procedure, therefore, first uses the observed data to calculate $\hat{\bSigma}$, then uses $\hat{\bSigma}$ to whiten the covariates, and finally applies HEDE to the whitened data.

\begin{table}
\caption{Heritability Estimation via Debiased Ensembles (HEDE)}\label{algo:algo}
\begin{tabular*}{\textwidth}{@{\extracolsep{\fill}}l}
\toprule
\multicolumn{1}{c}{\textbf{Input}: $\by,\bX,t$}\\
\multicolumn{1}{c}{\textbf{Output}: $\hat{h}^2$}\\
\midrule
0. If necessary, estimate $\hat\bSigma$ as block-diagonal and whiten the data.\\
1.~Fix a range $[t_{\min},t_{\max}]\subset (0,1)$. For all $\lambda_{\rmL},\lambda_{\rmR}$:\\
\tabitem Calculate $\hat\bbeta_{\rmL},\hat\bbeta_{\rmR}$ in \eqref{method:estimators_original}.\\
\tabitem If $\frac{\|\hat\bbeta_{\rmL}(\lambda_{\rmL})\|_0}{n} \notin [t_{\min},t_{\max}]$ or $\frac{1}{n}\operatorname{Tr}((\frac{1}{n} \bX^\top \bX + \lambda_{\textrm{R}} \bm{I})^{-1}\frac{1}{n} \bX^\top \bX) \notin [t_{\min},t_{\max}]$, drop $(\lambda_{\rmL},\lambda_{\rmR})$.\\
\tabitem Calculate $\hat\tau_{\rmL}^2,\hat\tau_{\rmR}^2,\hat\tau_{\rmL\rmR}$ in \eqref{main:eqn:tau_estimation}.\\
\tabitem Calculate $\hat\alpha_{\rmL}$ in \eqref{main:thm_eqn:alpha_optimal}, and $\hat\tau_C^2$ in \eqref{main:thm_eqn:tau_alpha_optimal}.\\
2.~Select $(\hat\lambda_{\rmL},\hat\lambda_{\rmR}):=\argmin_{\lambda_{\rmL},\lambda_{\rmR}}\hat\tau_C^2(\lambda_{\rmL},\lambda_{\rmR})$. Denote $\hat\tau_{C,\min}^2:=\hat\tau_C^2(\hat\lambda_{\rmL},\hat\lambda_{\rmR})$.\\
3.~Calculate $\hat\bbeta_{\rmL}^d,\hat\bbeta_{\rmR}^d$ from \eqref{method:estimators_debiased}, and $\hat\bbeta_{C}^d$ from \eqref{method:estimator_ensemble} corresponding to $\hat \lambda_L,\hat\lambda_R$.\\
4.~Calculate the heritability estimator using the formula\\
\vbox{\begin{equation}\label{method:estimator_FEH}
    \hat{h}^2_{\hat\alpha_{\rmL},\hat{\lambda}_{\rmL},\hat{\lambda}_{\rmR}} = \min\left\{1,\max\left\{0,\frac{\|\hat\bbeta_C^d\|_2^2 - p\cdot\hat{\tau}_{C,\min}^2}{\widehat{\Var}(y_i)}\right\}\right\},
\end{equation}}\\
where $\widehat{\Var}(y_i)$ denotes the sample variance of the outcome. \\
\bottomrule
\end{tabular*}
\end{table}

We algorithmically summarize HEDE in Table \ref{algo:algo}. 
As we can see for step $1$, assuming our hyperparameter grids have $m_{\rmL},m_{\rmR}$ different values of $\lambda_{\rmL},\lambda_{\rmR}$, respectively, we need $m_{\rmL}+m_{\rmR}$ \textit{glmnet} calls. As for the well known ElasticNet estimator, the same grids would require $m_{\rmL}m_{\rmR}$ \textit{glmnet} calls, making it much more expensive than the current ensembling method. We also remark that we did not make full effort into the computational/algorithmic optimization. With more advanced techniques such as \textit{snpnet} from \citep{li2022fast}, HEDE could potentially be scaled to very large dimensions.

\section{Theoretical Properties}\label{section:mathematical}
We turn to discuss the key theoretical properties of HEDE. We first focus on the case of independent covariates. Recall that throughout, we assume an asymptotic framework where $p,n \rightarrow \infty$ with  $\lim_{n\rightarrow \infty} p/n \rightarrow \delta \in (0,\infty)$. We work with the following assumptions on the covariate, signal, and noise. 
\begin{assumption}\label{assumptions}
\ 
\begin{enumerate}
    \item[(i)] Each entry $X_{ij}$ of $\bX$ is independent (but not necessarily identically distributed), has mean $0$ and variance $1$, and are uniformly sub-Gaussian.
    \item[(ii)] The norm of the signal $\bbeta$ satisfies: $0 \leq \|\bbeta\|_2^2 \convP \sigma_\beta^2 \leq \sigma_{\max}^2 < \infty$.
    \item[(iii)] The noise $\bep$ has independent entries (not necessarily identically distributed) with mean $0$ and variance $\sigma^2$ that is bounded ($\sigma^2 \leq \sigma_{\max}^2$), and are uniformly sub-Gaussian. Further, $\bep$ and $\bX$ are independent.
\end{enumerate}
\end{assumption}

It is worth noting that the uniform sub-Gaussianity in Assumption \ref{assumptions}(i) implies that the $\bar{G}_j$'s from \eqref{eqn:X_normalization} are bounded.  This is satisfied in GWAS, as $\bar{G}_j$ is minor allele frequency (MAF) \citep{psychiatric2009genomewide}.
The remaining Assumptions \ref{assumptions}(ii) and (iii) are all weak and natural to impose.
With Assumption \ref{assumptions} in mind, our primary goal is to establish that HEDE consistently estimates the heritability. The following theorem establishes this result under 
mild additional conditions on the minimum singular values of submatrices of $\bX$ (Assumption \ref{assumptionL}, Supplementary Materials \ref{additional_assumptions}). The proof is provided in Supplementary Materials \ref{subsec:proof:consistency}.

\begin{thm}\label{method:thm:consistency}
Recall $\hat{h}_{\alpha_{\rmL},\lambda_{\rmL},\lambda_{\rmR}}^2$ from \eqref{eq:hgeneral} and $h_{\hat\alpha_{\rmL},\hat\lambda_{\rmL},\hat\lambda_{\rmR}}^2$ from \eqref{method:estimator_FEH}. Under Assumptions \ref{assumptions} and \ref{assumptionL}, the following conclusions hold:
\begin{itemize}
\item[(i)] $\sup_{\lambda_{\textrm{L}},\lambda_{\textrm{R}} \in [\lambda_{\min},\lambda_{\max}],\alpha_{\rmL}\in[0,1]} \left|\hat{h}_{\alpha_{\rmL},\lambda_{\rmL},\lambda_{\rmR}}^2 - h^2\right|  \convP 0$
\item[(ii)] $\left|\hat h_{\hat\alpha_{\rmL},\hat\lambda_{\rmL},\hat\lambda_{\rmR}}^2 - h^2\right| \convP 0$.
\end{itemize}
\end{thm}

Theorem \ref{method:thm:consistency} states that HEDE is consistent when $\lambda_L,\lambda_R$ and $\alpha_{\rmL}$ are arbitrarily chosen in $[\lambda_{\min}, \lambda_{\max}]$ and $[0,1]$, respectively. It follows that HEDE is consistent using our proposed choices $\hat\alpha_{\rmL},\hat\lambda_{\rmL},\hat\lambda_{\rmR}$. Here $[\lambda_{\min}, \lambda_{\max}]$ denotes the range of $\lambda_L,\lambda_R$ that we use. See details in Supplementary Materials \ref{sec:thresholding}. 
 
Proving the consistency for given fixed choices of these hyperparameters is relatively straightforward, However, establishing the uniform gurantee is non-trivial. It requires analyzing the debiased Lasso and ridge estimators jointly, and furthermore, pinning down this joint behavior uniformly over all hyperparameter choices. We achieve both these goals under relatively mild assumptions on the covariate and noise distributions. Once we establish the uniform control, Part (ii) follows trivially since the uniform result provides consistency for the particular choices $\hat{\alpha}_L,\hat{\lambda}_L,\hat{\lambda}_R$ produced by our algorithm. Below, we present the key theorem necessary for proving Theorem \ref{method:thm:consistency}. 

\begin{prop}\label{main:cor:linear_combination}
Recall $\tilde \bbeta_C^d$ from \eqref{eq:comb} and $\tilde\tau_C^2 $ from \eqref{eq:tildetauCdef}. Under Assumptions \ref{assumptions} and \ref{assumptionL},
\begin{equation*}
    \sup_{\lambda_{\textrm{L}},\lambda_{\textrm{R}} \in [\lambda_{\min},\lambda_{\max}]} \sup_{\alpha_{L} \in [0,1]} \left|\|\tilde \bbeta_C^d\|_2^2 - \|\bbeta\|_2^2 -  p\cdot \tilde\tau_C^2\right|\convP 0,
\end{equation*}
\end{prop}

 We now discuss how Proposition \ref{main:cor:linear_combination} leads to our previous Theorem \ref{method:thm:consistency}. 
Proposition \ref{main:cor:linear_combination} equivalently proposes a consistent estimator for $\|\bbeta\|_2^2$, which is directly related to $h^2$ by a factor of the population variance of $y_i$. Since the population variance of $y_i$ can be consistently estimated by the sample variance of the observed outcomes, we can safely replace the former by the latter in the denominator. Theorem \ref{method:thm:consistency}(i) then follows on using the final fact that clipping does not affect consistency. 

Proposition is \ref{main:cor:linear_combination} proved in Supplementary materials \ref{subsec:proof:linear_combination}. Here we briefly outline our main idea. Recall from \eqref{eq:comb} that $\tilde \bbeta_C^d$ is a linear combination of $\hat\bbeta_{\rmL}^d$ and $\hat\bbeta_{\rmR}^d$. To develop uniform guarantees for $\tilde \bbeta_C^d$, it therefore suffices to establish the corresponding uniform guarantees for the debiased Lasso and ridge jointly, which we establish in the following theorem. 

\begin{thm}\label{main:thm:debiased_uniform}
Recall the estimators \eqref{method:estimators_debiased}. If Assumptions \ref{assumptions} and \ref{assumptionL} hold, then for any $1$-Lipschitz function $\phi_\beta:(\R^p)^3\rightarrow \R$, we have
$$\sup_{\lambda_{\textrm{L}},\lambda_{\textrm{R}} \in [\lambda_{\min},\lambda_{\max}]}|\phi_\beta(\hat\bbeta_{\textrm{L}}^d,\hat\bbeta_{\textrm{R}}^d,\bbeta) - \E[\phi_\beta(\hat\bbeta_{\textrm{L}}^{f,d},\hat\bbeta_{\textrm{R}}^{f,d},\bbeta)]| \convP 0,$$
where $(\hat\bbeta_{\textrm{L}}^{f,d}, \hat\bbeta_{\textrm{R}}^{f,d}) = (\bbeta + \bg_{\textrm{L}}^f, \bbeta + \bg_{\textrm{R}}^f)$ with $(\bg_{\textrm{L}}^f,\bg_{\textrm{R}}^f) \sim \mathcal{N}(0,\bS \otimes \bm{I}_p)$, where $\bS = \begin{pmatrix}\tau_{\textrm{L}}^2 & \tau_{LR} \\ \tau_{LR} & \tau_{\textrm{R}}^2\end{pmatrix}$ is formally defined via the system of equations \eqref{def:fixed_point}.
\end{thm}

Theorem \ref{main:thm:debiased_uniform} is proved in Supplementary materials \ref{subsec:proof:debiased_uniform}. It states that the joint ``Gaussianity" described in \eqref{eqn:approx_debiased_true} occurs at the population level in terms of the value of Lipschitz functions of the estimators, in addition to a per-coordinate basis. While results of this nature have appeared in the recent literature on exact high-dimensional asymptotics, 
majority of these works focus on a single regularized estimator (or its debiased version) with a fixed tuning parameter value. Theorem \ref{main:thm:debiased_uniform} provides the first joint characterization of multiple debiased estimators that is uniform over a range of tuning parameter values. 

Among related works, \cite{celentano2021cad} established joint characterization of two estimators and their debiased versions for given fixed tuning parameter choices. Meanwhile,  \cite{miolane2021distribution} characterized certain properties of the Lasso and the debiased Lasso uniformly across tuning parameter values.  However, developing a heritability estimator that performs well in practice across sparse to dense signals requires the stronger result of the form in Theorem \ref{main:thm:debiased_uniform}. Dealing with the debiased lasso and ridge jointly, while requiring uniform characterization over tuning parameters as well as allowing sub-Gaussian covariates, poses unique challenges. We overcome these challenges in our work, building upon 
\cite{miolane2021distribution}, \cite{ celentano2020lasso}, \cite{celentano2021cad} and universality techniques proposed in \cite{han2022universality}, which in turn rely on the Convex Gaussian Minmax Theorem (GCMT) \citep{thrampoulidis2015regularized}. We detail the connections and differences in the Supplementary Materials.

Note that Theorem \ref{main:thm:debiased_uniform} involves the variance-covariance parameters $\tau_L^2,\tau_{LR},\tau_R^2$ (formal definitions deferred to later \eqref{def:bS_g} \eqref{def:fixed_point} in the interest of space.
We establish in Supplementary Materials, Theorem \ref{main:thm:tau_estimation} that $\hat\tau_{\rmL}^2,\hat\tau_{\rmR}^2,\hat\tau_{\rmL\rmR}$, defined in \eqref{main:eqn:tau_estimation}, consistently estimate these parameters. Furthermore, we show that this consistency holds uniformly over the required range of tuning parameter values. This result, coupled with Theorem \ref{main:thm:debiased_uniform}, yields our key Proposition \ref{main:cor:linear_combination}. 

Finally we turn to the case where the observed features are correlated. Recall from Section \ref{sec:cov} that here we consider the design matrix to take the form $\bX = \bZ \bSigma^{1/2}$, where $\bZ$ now
has independent entries and satisfies the conditions in Assumption \ref{assumptions}(i). Thus, our previous theorems apply to $\bZ$ (which is unobserved), but not to the observed $\bX$. 
Applying whitening with estimated $\hat\bSigma$ results in final features of the form $\bX\hat{\bSigma}^{-1/2}=\bZ\bSigma^{1/2}\hat{\bSigma}^{-1/2}$. When $\hat{\bSigma}^{1/2}$ approximates $\bSigma$ accurately, $\bSigma^{1/2}\hat{\bSigma}^{-1/2}$ behaves as the identity matrix $\bI$. In this case, $\bX \hat{\bSigma}^{-1/2}$ is close to $\bZ$ which satisfies our previous Assumption \ref{assumptions}. HEDE therefore continues to enjoy uniform consistency guarantees, as formalized below.

\begin{thm}\label{method:thm:consistency_estimated_cov}
Consider the setting of Theorem \ref{method:thm:consistency} except that the observed features now take the form $\bX = \bZ \bSigma^{1/2}$ where $\bZ$ satisfies Assumptions \ref{assumptions}(i) and \ref{assumptionL}(i). 
Let $\bSigma$ be a sequence of deterministic symmetric matrices whose eigenvalues are bounded in $[1/M,M]$ for some $M>1$. Let $\hat\bSigma$ be a sequence of random matrices such that $\|\hat\bSigma - \bSigma\|_{\op} \convP 0$. If $\bA = \bSigma^{1/2}\hat\bSigma^{-1/2}$ satisfies Assumption \ref{assumptionA} 
and $\hat{h}^2_{\alpha_L,\lambda_L,\lambda_R}$ denotes the estimator \ref{eq:hgeneral} calculated using the whitened data $(\by, \bX\hat{\bSigma}^{-1/2}) $, then
\begin{equation}\label{method:thm:equation:consistency_estimated_cov}
    \sup_{\lambda_{\textrm{L}},\lambda_{\textrm{R}} \in [\lambda_{\min},\lambda_{\max}],\alpha_{\rmL}\in[0,1]} \left|\hat{h}_{\alpha_{\rmL},\lambda_{\rmL},\lambda_{\rmR}}^2 - h^2\right| \convP 0,
\end{equation}
Thus, the corresponding HEDE estimator $\hat{h}^2_{\hat{\alpha}_L,\hat{\lambda}_L,\hat{\lambda}_R}$ is also consistent for $h^2$. 
\end{thm}
Although we present this theorem for a general sequence of covariance matrices $\bSigma$ that allow accurate estimators in operator norm, we use it only for block diagonal covariances relevant for our application, where the block sizes are relatively small compared to the sample size. In Supplementary Materials, Proposition \ref{examples_covariance_estimation}, we establish that for such covariance matrices, the block-wise sample covariance matrix provides such an acccurate estimator. Theorem \ref{method:thm:consistency_estimated_cov} is proved in Supplementary Materials \ref{sec:robustness_proof}. The key idea lies in establishing that the estimation error incurred while replacing $\bSigma$ by $\hat{\bSigma}$ does not affect the uniform consistency properties of HEDE. 

\section{Simulation Results}\label{section:simulations}
We evaluate the performance of HEDE through simulation studies. In Section \ref{section:simulations_random}, we compare HEDE to some representative random effects methods for estimating the heritability under different signal structures. In Section \ref{section:simulations_fixed}, we compare HEDE against several fixed effect based methods for estimating the  heritability.

\subsection{Comparison with random effects methods}\label{section:simulations_random}
Given the widespread adoption of the random effects assumption in genetics, it is important to benchmark HEDE against various random-effects-based methods. A diverse array of such methods exist, each based on a different random effects assumption. We show that a random-effects-based method fails to provide unbiased heritability estimates when corresponding assumptions are violated (i.e. signals are not generated as assumed). In contrast, HEDE constantly offers unbiased estimates (albeit with a slightly higher variance). 

For conciseness, we focus on three representative random effects methods: GREML-SC \citep{yang2011gcta}: effectively GCTA, which assumes no stratification; GREML-MS \citep{lee2013estimation}, which is a modified version of GREML that accounts for minor allele frequency (MAF) stratification; GREML-LDMS \citep{yang2015genetic}, which is a modified version of GREML that accounts for both linkage disequilibrium (LD) level and MAF stratification. 
Their corresponding random effects assumptions state that within respective stratifications, the signal entries are i.i.d.~draws from a zero-inflated normal distribution (equivalently, the non-zero entries are randomly/evenly located and normally distributed).

For our simulations, we use genotype data from unrelated white British individuals in the UK Biobank dataset \citep{sudlow2015uk}. We select common variants with a minor allele frequency (MAF) greater than $1 \%$ from the UK Biobank Axiom array using PLINK. Additionally, we exclude SNPs with genotype missingness over $1 \%$ and Hardy-Weinberg disequilibrium p-values greater than $10^{-6}$. We impute the remaining small number of missing values as 0. This results in a design matrix $\tilde{\bX}$ with $n=332,430$ individuals and $p = 533,169$ variants across 22 chromosomes, which is then normalized so each column has a mean of 0 and variance of 1. To expedite the process, our simulation analysis focuses on chromosome $22$. We use $100$ randomly sampled disjoint subsets of $3000$ individuals each to simulate $100$ independent random draws of $\bX$. We then generate a random $\bbeta$ based on different assumptions and generate $\by$ following the linear model in \eqref{model:linear}. This results in $100$ heritability estimates, which are estimated using GREML methods and HEDE. We then compare the average of these estimates with the true population heritability value.

While our formal mathematical guarantees are established under the fixed effects assumption, it is desirable to stress-test the performance of our method in the random effects simulation settings, especially since we wish to benchmark our method against popular methods from statistical genetics. We further discuss the connection between the random effect assumption and the fixed effect assumption in terms of the true population parameter (Supplementary Materials \ref{subsec:random_simulation_heritability_definition}), with the discussion aimed at explaining why the empirical comparison we perform here is  reasonable, given the methods were developed under different assumptions. 

For the GREML methods, we use $4$ LD bins generated from the quantiles of individual LD values \citep{evans2018comparison}, and $3$ MAF bins: $[0.01,0.05],[0.05,0.1]$, and $[0.1,0.5]$. 
For HEDE, we perform the whitening step by estimating the population covariance $\bSigma$ as a block diagonal with blocks specified in \citep{berisa2016approximately}. We estimate the block-wise population covariances using sample covariances. Given that $n=332,430$ and $p<1000$ for each LD block, this approach results in minimal high-dimensional error.

\begin{figure}[ht]
    \centering
    \includegraphics[width=\linewidth]{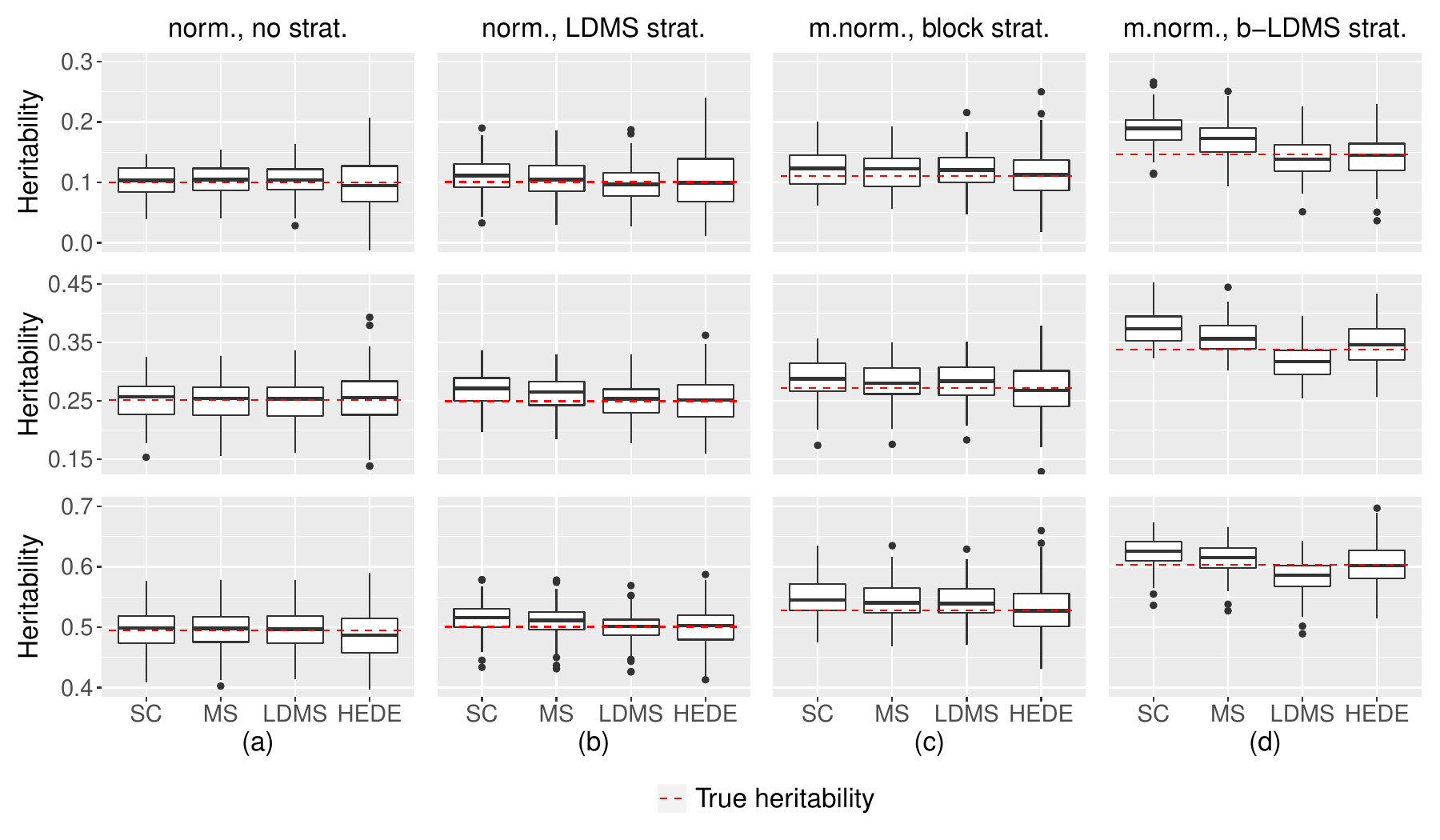}
    \caption{Box plots comparing three random effects methods, namely GREML with different stratifications, with our fixed effect method HEDE. Meaning of method abbreviations: "SC": single-component \citep{yang2011gcta}, "MS": MAF stratified, "LDMS": LD and MAF stratified \citep{yang2015genetic} Different rows use different true $\|\bbeta\|_2$ values. Each column indicates a different distribution from which the true signal is drawn: (a) zero-inflated normal, where the nonzero entries are uniformly distributed (no stratification). (b) zero-inflated normal, where the nonzero entries have different densities in each LDMS stratification. (c) zero-inflated mixture of two normals, where the nonzero entries have different densities in each (consecutive) LD block. (d) zero-inflated mixture of two normals, where the nonzero entries have different densities in each (consecutive) LD block and in each LDMS stratification. Red dotted horizontal line represents true heritability values (which varies depending on the signal distribution). Boxes represents distributions of different methods' estimates from $100$ independent draws of $\bX$ (from UKB data chromosome $22$, 3000 individuals each) and $\bbeta$ (from the aforementioned distribtions). GREML methods show biases under model misspecification while HEDE remains unbiased. See details in Section \ref{section:simulations_random}. }
    \label{fig:attack}
\end{figure}

To assess the robustness of different methods against the misspecification of typical random effects assumptions on the regression coefficients, we created $4$ types of random signals. These signals vary in the locations and distributions of non-zero entries (\textit{causal variants}) as shown in Figure \ref{fig:attack}. In Figure \ref{fig:attack}a, the non-zero locations of $\bbeta$ are uniformly distributed, with non-zero components normally distributed. Here the assumptions for all three GREML methods are satisfied. In Figure \ref{fig:attack}b, the non-zero locations of $\bbeta$ vary in concentration across  LDMS strata, with non-zero components normally distributed. Here only the assumption for GREML-LDMS is satisfied. In Figure \ref{fig:attack}c, the non-zero locations of $\bbeta$ are more concentrated in the last LD block, with non-zero components distributed as a mixture of two normals (centered symmetrically around zero). Here none of the assumptions of the three GREML methods is satisfied. In Figure \ref{fig:attack}d, the non-zero locations of $\bbeta$ vary in concentrations across both LD blocks and LDMS strata, with non-zero components distributed as a mixture of two normals. Here none of the assumptions of the three GREML methods are satisfied, with varying degrees of violations. The full details of the signal generation process are described in Supplementary Materials \ref{subsec:random_simulation_details}.

Figure \ref{fig:attack} clearly demonstrates that the GREML methods are susceptible to signal distribution misspecifications, exhibiting increased bias as the extent of assumption violation rises.  In contrast, HEDE is robust and provides unbiased estimates in all cases, despite exhibiting higher variance. This pattern remains robust across various levels of true population heritability ($0.1,0.25$ and $0.5$ shown).

We believe these findings are likely to apply to other random-effects-based methods. A random effects method can only guarantee unbiased estimates when its underlying random effects assumptions are satisfied---a condition that is often difficult to verify with real data. In contrast, HEDE is not susceptible to such biases, demonstrating that the fixed effect based HEDE method 
applies under much diverse signal structures. 

\subsection{Comparison of HEDE with other fixed effects methods}\label{section:simulations_fixed}
In this section, we compare the performance of HEDE with other fixed effects heritability estimation methods.  Recall we showed in Section \ref{sec:cov} and Theorem \ref{method:thm:consistency_estimated_cov} that HEDE remains consistent under accurate covariance estimation strategies. Also, fixed effects heritability estimation methods all assume $\bSigma$ is known or can be accurately estimated. In light of these, we assume without loss of generality that $\bSigma=\bI$, in order to directly compare the relative performance of fixed effect methods.
To mimic discrete covariates in genetic scenarios, we generate the design matrix $\bX \in \R^{n \times p}$ via \eqref{eqn:X_normalization}, where $G_{ij} \sim \text{Binomial}(2,\pi_j), 
\forall i,j$ and $\pi_j \sim\text{Unif}[0.01,0.5],\forall j$, representing common variants with a minor allele frequency of at least $1\%$. We generate the true signal $\bbeta \in \R^p$ as i.i.d.~zero-inflated normals with non-zero co-ordinate weight $\kappa$ and heritability $h^2$: $\beta_j \iid \kappa \cdot \mcn(0,h^2/ (p\kappa(1-h^2))) + (1-\kappa) \cdot \delta_0$. 
We generate the noise $\bep \in \R^n$ with i.i.d.~$\mcn(0,1)$ entries. We vary $n,p,\kappa,h^2$ across our simulation settings.

For each simulation setting, we compare HEDE with a Lasso-based method, AMP \citep{bayati2013estimating}\footnote{We label this method as AMP, since it was created based on approximate message passing (AMP) algorithms.}, and three Ridge-based methods: Dicker's moment-based method \citep{dicker2014variance}, Eigenprism \citep{janson2017eigenprism}, and Estimating Equation \citep{chen2022statistical}. For AMP, \citep{bayati2013estimating} does not provide a strategy for tuning the regularization parameter. Therefore,  we choose the tuning parameter based on cross validation (specifically the  Approximate Leave one Out (ALO) \cite{rad2020scalable} approach) 
performed over $10$ smaller datasets
with the same values of $\delta,h^2,\kappa$. All aforementioned methods operate under similar fixed effects assumption as HEDE, except that Eigenprism only handles $n\leq p$ in its current form. We first compared these methods in terms of bias and found that they all remain unbiased across a wide variety of settings. In Supplementary Matrials \ref{sec:bias_illustration}, we demonstrate their unbiasedness for a specific choice of $n,p,h^2,\kappa$. In light of this, here we present comparisons between these estimators in terms of their mean square errors (MSEs). 

\begin{figure}[ht]
    \centering
\includegraphics[width=\linewidth]{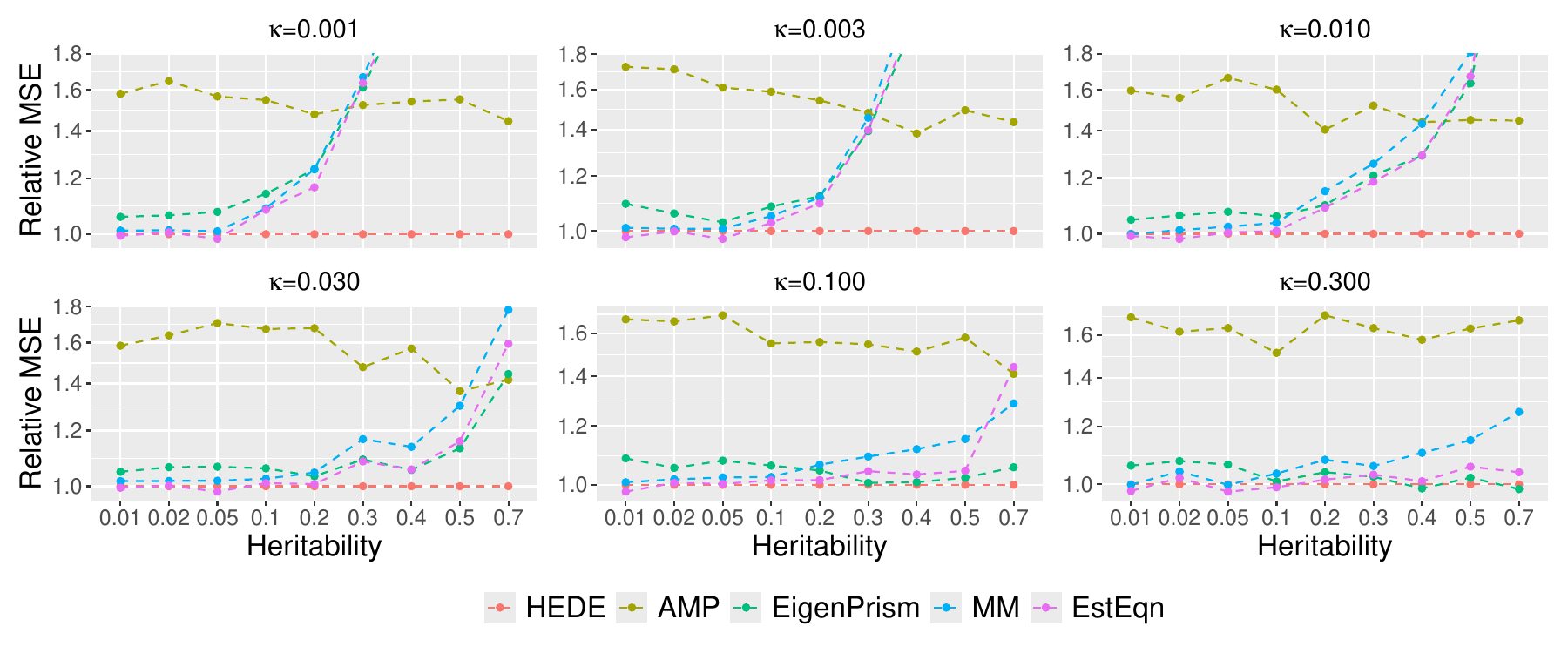}
    \caption{Ratio of the MSEs (called relative MSE) of different fixed effects methods---HEDE (ours), AMP \citep{bayati2013estimating}, EigenPrism \citep{janson2017eigenprism}, MM \citep{dicker2014variance}, EstEqn \citep{chen2022statistical}---with respect to HEDE. Thus the baseline is 1, represented by the orange line.  The design matrices satisfy \eqref{eqn:X_normalization}, where $G_{ij} \sim \text{Binomial}(2,\pi_j), 
    \forall i,j$ and $\pi_j \sim\text{Unif}[0.01,0.5],\forall j$,
     with dimensions $n,p$. Recall that signals are drawn from i.i.d.~zero-inflated normals with non-zero co-ordinate weight $\kappa$ and true heritability $h^2$: $\beta_j \iid \kappa \cdot \mcn(0,h^2/ (p\kappa(1-h^2))) + (1-\kappa) \cdot \delta_0$.   We generate $10$ draws of signals with $100$ draws of design matrices and display average MSEs. Parameter values: $n=1000,p=10000$, $\kappa$ varies across panels as indicated in the subtitles, and $h^2$ varies across the $x$-axis as indicated by the label. 
     See full details in Section \ref{section:simulations_fixed}.
     }
    \label{fig:change_eps_h_10000_com}
\end{figure}

In Figure \ref{fig:change_eps_h_10000_com}, we display the relative MSEs of heritability estimates in relation to $\kappa$ and  $h^2$, while fixing sample size $n=1000$ and number of variants $p=10000$. We vary sparsity from $0.001$ to $0.3$ on a roughly evenly spaced log scale, and vary heritability from $0.01$ to $0.7$.  This range mimics the potential true heritability values for traits such as autoimmune disorders (lower heritability) and height (higher heritability) \citep{hou2019accurate}. Although the methods under comparison operate under fixed effects assumptions, we still simulate $10$ signals to mitigate the influence of any particular signal choice. For each signal, we generate $100$ random samples and calculate the MSE for estimating heritability. We then average the error over the 10 signals.  
\begin{figure}[ht]
    \centering
    \includegraphics[width=\linewidth]{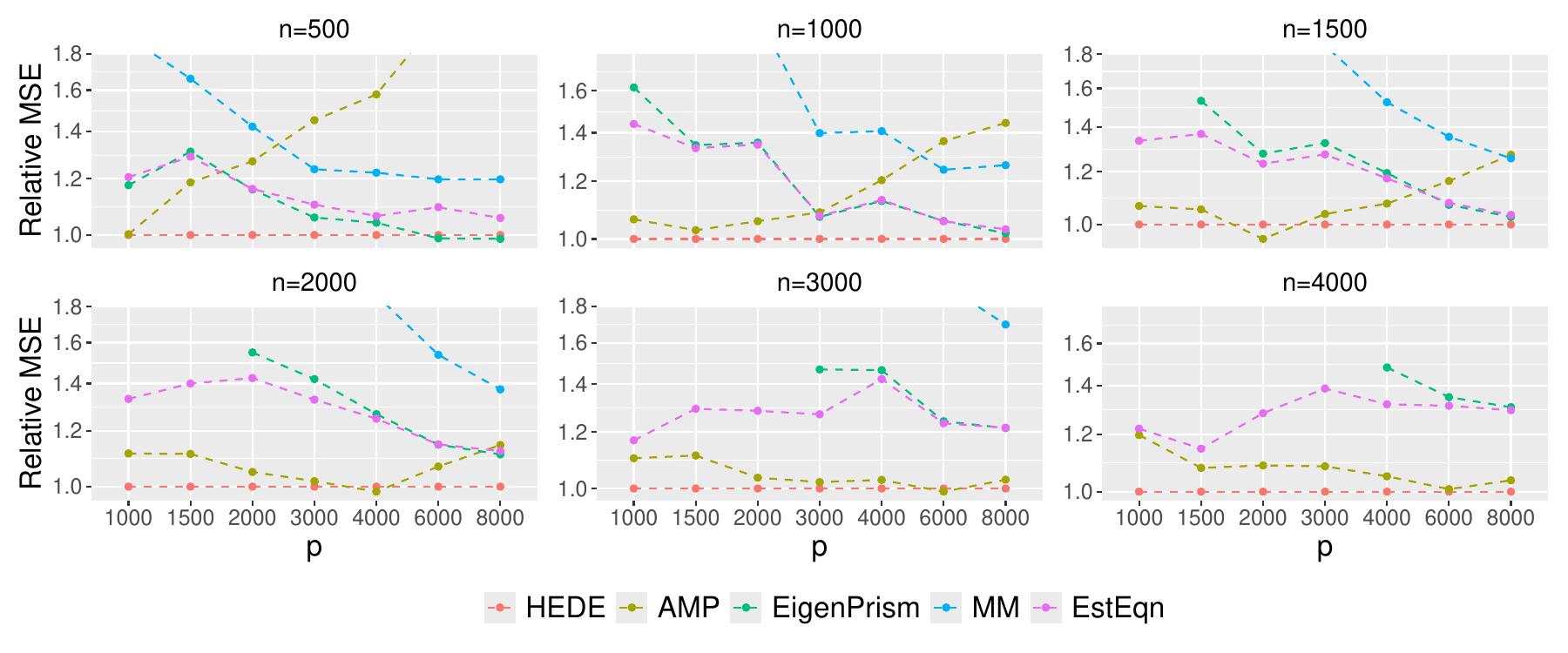}
    \caption{Settings same as Figure \ref{fig:change_eps_h_10000_com}. Parameter values: $h^2=0.5,\kappa=0.1$, $n$ varies across panels as indicated in the subtitles, and $p$ varies across the $x$-axis as indicated by the label.
     }
    \label{fig:change_n_p_0100_05_com}
\end{figure}

We observe that HEDE outperforms ridge-type methods particularly when the true signal is sparse and relatively strong, while performing comparably in other scenarios. Since HEDE utilizes the strengths of both Lasso and ridge, the performance gain  over pure ridge-based methods, such as EigenPrism, MM and EstEqn, is natural. The pure Lasso-based method (AMP, green) performs sub-optimal compared to all the others, and its higher MSE values are dominating the y-axis range. For a closer view, we present a zoomed in version of the plot in Figure \ref{fig:zoomedin} where we exclude AMP (we defer the figure to Supplementary Materials \ref{sec:bias_illustration} in light of space).   Figures \ref{fig:change_eps_h_10000_com} and \ref{fig:zoomedin} clearly demonstrate the efficacy of our ensembling strategy. It leads to a superior heritability estimator across signals ranging from less sparse to highly sparse.

To better investigate our relative performance, in Figure \ref{fig:change_n_p_0100_05_com}, we consider a specific setting---$\kappa=0.1,h^2=0.5$---where HEDE has a moderate advantage in Figure \ref{fig:change_eps_h_10000_com}, We vary $n$ and $p$ so the ratio $\delta=p/n$ ranges from $0.25$ to $20$.
We clearly observe that HEDE's advantage across a broad spectrum of configurations. Note that Eigenprism's results are missing for $n > p$  as it currently only supports $n \leq p$. 
Also observe that AMP's relative performance degrades sharply with larger $\delta$ values, as shown in the first panel of Figure \ref{fig:change_n_p_0100_05_com}. This explains its sub-optimal performance in Figure \ref{fig:change_eps_h_10000_com} when $\delta=10$. Figure \ref{fig:change_n_p_0100_05_com} demonstrates that HEDE leads to superior performance over a wide range of sample sizes, dimensions and dimension-to-sample ratios. It consistently outperforms the other methods for the setting examined.

To determine if HEDE's superior performance extends to  broader contexts, we conducted additional investigations. We examined settings where the designs contain larger sub-Gaussian norms, and in addition, settings where the non-zero entries of the signal were drawm  from non-normal distributions and/or are distributed with stratifications (see Section \ref{section:simulations_random}). Results for the former investigation are presented in Supplementary Materials. \ref{sec:more_simulations}, whereas we omit those for the latter in light of the performances looking similar. Across the board, we observed that HEDE maintained maintains its competitive advantage over other methods. 
We also investigated a Lasso-based method CHIVE \citep{tony2020semisupervised} in our experiments. However, CHIVE exhibited considerable bias in many settings, particularly when its underlying sparsity assumptions were violated. For this reason, we have omitted CHIVE from our result displays.

\section{Real data application}\label{section:real_data}
We next describe our real data experiments. We apply HEDE on the unrelated white British individuals from the UK biobank GWAS data which contain common variants  \citep{sudlow2015uk} to estimate the heritability for two commonly studied phenotypes: height and BMI. 
For the design matrix $\bX$, we follow the preprocessing steps outlined in Section \ref{section:simulations_random}, but include all $22$ chromosomes, unlike Section \ref{section:simulations_random} which considers only chromosome $22$. For each phenotype, we exclude individuals with missing values. We center the outcomes and denote the resultant vector by $\tilde \by$. We account for non-genetic factors by including covariates for age, age$^2$, sex as well as $40$ ancestral principal components. We derive the final response vector $\by$ by regressing out the influence of these non-genetic covariates from $\tilde \by$ using Ordinary Least Squares (OLS). The large sample size of approximately $330,000$ individuals compared to only $43$ non-genetic covariate, ensures that OLS can be confidently applied.

We apply HEDE to preprocessed $\bX$ and $\by$. For comparison, we also apply GREML-SC and GREML-LDMS (see Section \ref{section:simulations_random} for details) from the random effects methods family, as well as MM and AMP from the fixed effects methods family. We did not include Eigenprism and EstEqn since their runtimes are both $O(n^3)$ (assuming $p/n \rightarrow \delta$), so they could not be completed within a reasonable time frame.   We calculate standard errors using the standard deviation of point estimates from $15$ disjoint random subsets, each containing $20,000$ individuals.\footnote{In our high-dimensional setting, it is known that traditional resampling techniques for estimating standard errors such as the bootstrap, subsampling, etc. fail \citep{el2018can, clarte2024analysis, bellec2024asymptotics}. Corrections have been proposed to address issues with using these for estimating standard errors while inferring individual regression vector coefficients and closely related problems. However, such strategies for heritability estimation are yet to be developed. To circumvent this issue, we use disjoint subsets of observed data to mimic independent training data draws from the population}.  Additionally, due to memory constraints, many methods (including HEDE) could not process all $22$ chromosomes simultaneously. Therefore, for all methods, we estimate the total heritability by summing the heritabilities calculated for each chromosome separately. This approach is expected to introduce minimal error, given the widely accepted assumption of independence across chromosomes. 

\begin{table}
\centering
\begin{tabular}{c c c c c c c}
\toprule
\textbf{Trait} & \textbf{HEDE} & \textbf{MM} & \textbf{AMP} & \textbf{GREML-SC}  & \textbf{GREML-LDMS}\\
\midrule
\text{Height} & 0.681  & 0.813  & 0.436  & 0.741 & 0.605\\
\text{(SE)} &(7.47E-02) & (9.30E-02) &  (3.01E-01) &  (2.87E-02) & (3.19E-02)\\
\hline
\text{BMI} & 0.279 & 0.369  & 0.510  & 0.316  & 0.265\\
\text{(SE)} & (5.20E-02) &  (6.66E-02) & (3.52E-01) &  (2.85E-02) & (3.22E-02)\\
\bottomrule
\end{tabular}
\caption{Estimates of heritability from HEDE (ours), MM \citep{dicker2014variance}, AMP \citep{bayati2013estimating}, GREML-SC \citep{yang2011gcta} and GREML-LDMS \citep{yang2015genetic} for certain traits in the UKB dataset. $533,169$ SNPs are present in $22$ chromosomes. For fixed effect methods, we assume independence across LD blocks from \cite{berisa2016approximately} for covariance estimation. Standard errors are computed from $15$ disjoint random subsets of $20,000$ individuals each. See full details in Section \ref{section:real_data}.  } 
\label{table:real_data}
\end{table}

We summarize the results in Table \ref{table:real_data}.
First, we observe that MM and AMP estimates differ more from HEDE compared to other methods. This is consistent with Figure \ref{fig:change_eps_h_10000_com}, where we observe that AMP often has larger MSE than other fixed effects methods. Thus, we expect AMP estimates to differ significantly from HEDE, and MM. For MM, the situation is more nuanced. Prior work suggests the true height heritability should be on the higher end of the range shown in Figure \ref{fig:change_eps_h_10000_com}. In this region, MM has higher MSE than HEDE across all sparsity levels. Therefore, MM estimates of height heritability may be less accurate than HEDE. For BMI, prior work provides heritability estimates within the range $0.285$ to $0.436$ using various methods. Referring to Figure \ref{fig:change_eps_h_10000_com}, we see that MM's relative MSE compared to HEDE depends on the sparsity level in this BMI heritability range. Without knowing the exact sparsity, it is difficult to draw conclusions. However, we still observe the general trend of MM producing larger estimates than other methods for BMI, similar to height.

Comparing HEDE to the two random effect methods, we observe that HEDE point estimates fall between GREML-SC and GREML-LDMS ones. For both height and BMI, HEDE estimates lie between GREML-SC and GREML-LDMS,, with the latter always undershooting.  This is consistent with the trend observed in the last column of Figure \ref{fig:attack}. This difference may indicate a heterogeneous underlying genetic signal, especially in a way that differentiates the SC and LDMS approaches within the GREML framework.

\section{Discussion}\label{section:discussion}
In this paper, we introduce HEDE, a new method for heritability estimation that harnesses the strengths of $\ell_1$ and $\ell_2$ regression in high-dimensional settings through a sophisticated ensemble approach. We develop data-driven techniques for selecting an optimal ensemble and fine-tuning key hyperparameters, aiming to enhance statistical performance in heritability estimation. As a result, we present a competitive estimator that surpasses existing fixed effects heritability estimators across diverse signal structures and heritability values. Notably, our approach circumvents bias issues inherent in random effects methods. In summary, our contribution offers a dependable heritability estimator tailored for high-dimensional genetic problems, effectively accommodating underlying signal variations. Importantly, our method maintains consistency guarantees even with adaptive tuning and under minimal assumptions on the covariate distribution. We validate the efficacy of our approach through comprehensive simulations and real-data analysis. 

HEDE's computational complexity is primarily driven by the Lasso/ridge solutions, making it more efficient than many existing heritability estimators (e.g. Eigenprism \citep{janson2017eigenprism}, where an SVD is required, and EstEqn \citep{chen2022statistical}, where chained matrix inversions are involved). HEDE is potentially scalable to much larger dimensions via \textit{snpnet} \citep{rad2020scalable}. Furthermore, with proper Lasso/ridge solving techniques, HEDE only needs knowledge of summary statistics $\bX^\top\bX$ and $\bX^\top\by$. This makes it an attractive option in situations where individual data sharing is difficult due to privacy concerns.  However, it is important to note that these discussions assume the availiability of an accurate estimate of the LD blocks and strength. This is a common limitation that other fixed effects methods also face. Any potential errors in covariance estimation could naturally lead to additional errors in heritability estimation. Therefore, improving and scaling covariance estimation in high dimensions could independently benefit heritability estimation. We manage to circumvent this issue by leveraging the hypothesis of independence across chromosomes, which significantly aids us computationally. 

Our theoretical framework relied on two pivotal assumptions: the independence of observed samples and a sub-Gaussian tailed distribution for covariates. Natural settings occur where these assumptions are violated. Within the genetics context, familial relationships and repeated measures in longitudinal studies among observed samples can introduce dependence. For the broader problem of signal-to-noise ratio estimation in various scientific fields, both assumptions might be violated.
For instance, applications such as wireless communications \citep{eldar2003asymptotic} and magnetic resonance imaging \citep{benjamini2013shuffle} may introduce challenges such as temporal dependence or heavier-tailed covariate distributions. Exploring analogues of HEDE under such conditions presents an intriguing avenue for future study. Recent advancements in high-dimensional regression, which extend debiasing methodology and proportional asymptotics theory for regularized estimators to these complex scenarios \citep{li2023spectrum,lahiry2023universality}, offer valuable insights for further exploration. We defer these investigations to future research.

We focus on continuous traits in this paper. Discrete disease traits occur commonly in genetics. Such situations are typically modeled using logistic, probit, or binomial regression models, depending on the number of possible categorical values of the response. Adapting our current methodology to these scenarios is a crucial avenue for future research. Debiasing methodologies for generalized linear models exist in the literature. A compelling question would be to understand how various debiasing techniques can be ensembled to optimize their benefits for heritability estimation, similar to our approach in this work for the linear model. 

\bibliographystyle{apalike} 
\bibliography{main.bib}

\appendix

\section{Additional Methodology Details}
\subsection{Additional Technical Assumptions}\label{additional_assumptions}
We present additional technical assumptions needed for our mathematical guarantees in Section \ref{section:mathematical}. The following assumption is needed for independent covariates:

\begin{assumption}\label{assumptionL}
\ 
\begin{enumerate}
    \item The design matrix $\bX$ satisfies for any $\gamma < 1$, there exists a constant $\kappa_{\min}>0$ such that $\PP(\frac{1}{\sqrt{n}}\kappa_-(\bX,\min\{p,\gamma n\}) \leq \kappa_{\min}) \rightarrow 0$ as $n \rightarrow \infty$, where $\kappa_-(\bX,s)$ denotes the minimum singular value of $\bX_S$ over all subsets $S$ of the columns with $|S|\leq s$.
    \item The tuning parameters are bounded, that is, there exists $\lambda_{\min},\lambda_{\max}$ such that $0 < \lambda_{\min} \leq \lambda_{\textrm{L}}, \lambda_{\textrm{R}} \leq \lambda_{\max} < \infty$.
\end{enumerate}
\end{assumption}

Assumption \ref{assumptionL}(1) states that the minimum singular value of $\bX$, across all subsets of a certain size, is lower bounded by some positive constant with high probability. This assumption is required for our universality proof. In fact, we conjecture that this assumption is true given Assumption \ref{assumptions}(1) and Assumption \ref{assumptions}(2) (c.f. Section B.5.4 in \cite{celentano2020lasso} for a proof in Gaussian case). We leave the proof of this assumption to future work.

Assumption \ref{assumptionL}(2) restricts the range of $\lambda_\rmL, \lambda_\rmR$ to a predetermined range. We require a specific lower bound $\lambda_{\min} = \lambda_{\min}(\sigma_{\max}^2,\delta)$. The specific functional form is complicated, but plays a crucial role in our proof in Section \ref{sec:universality} (which builds upon results from \cite{han2022universality}). On the other hand, $\lambda_{\max}$ is not restricted and can take any proper value. See Section \ref{sec:thresholding} for a heuristic method that we use to determine $(\lambda_{\min},\lambda_{\max})$ in practice.

The following assumption is needed for estimated covariance under correlated case, effectively perturbing HEDE by a matrix $\bA \approx \bI$:

\begin{assumption}\label{assumptionA}
Let $\bA \in \R^{p\times p}$ be a sequence of random matrices. Denote $\mcl{L}_{\bA}(\bb) := \frac{1}{2n} \|\by-\bX\bA\bb\|_2^2 + \frac{\lambda_{\rmL}}{\sqrt{n}}\|\bb\|_1$ to be the perturbed Lasso cost function when we replace $\bX$ by $\bX\bA$, and denote
$\hat\bbeta(\bA) := \argmin_{\bb} \mcl{L}_{\bA}(\bb)$ to be the perturbed Lasso solution. We say $\bA$ satisfies Assumption \ref{assumptionA} if

\begin{enumerate}
    \item $\supll\|\hat\bbeta(\bA) \leq M$ for some constant $M$ with probability converging to $1$ as $n \rightarrow \infty$
    \item $\forall \lambda_{\rmL},\lambda_{\rmL}' \in [\lambda_{\min},\lambda_{\max}], \mcl{L}_{\lambda_{\rmL}',\bA}(\hat\bbeta_{\lambda_{\rmL}}(\bA)) \leq \mcl{L}_{\lambda_{\rmL}',\bA}(\hat\bbeta_{\lambda_{\rmL}'}(\bA)) + K|\lambda_{\rmL}-\lambda_{\rmL}'|$ with probability converging to $1$ as $n \rightarrow \infty$.
\end{enumerate}
\end{assumption}

Assumption \ref{assumptionA} is required as technical steps in our proof under sub-Gaussian design. They are automatically satisfied when the design $\bX$ is Gaussian (see proofs in Section \ref{sec:robustness_proof}).

\subsection{Hyperparameter Selection}\label{sec:thresholding}
Our algorithm necessitates a tuning parameter range $[\lambda_{\min},\lambda_{\max}]$. Assumption \ref{assumptionL}(2) defines $\lambda_{\min}(\sigma_{\max}^2,\delta)$ as a function of the unknown $\sigma_{\max}^2$, as well as an unconstrained $\lambda_{\max}$, for completely technical reasons, both of which give little practical guidance. Here we propose to determine the range by the empirical degrees-of-freedom from Lasso and Ridge, defined in \ref{def:hat_df}.

For $\lambda_{\max}$, there is no point increasing it infinitely since $\hat\bbeta_{\rmL},\hat\bbeta_{\rmR}$ will approach $0$, yielding very similar $\hat\bbeta_{\rmL}^d,\hat\bbeta_{\rmR}^d$ values. Thus, we choose $t_{\min}=0.01$, therefore filtering out excessively large $\lambda_{\rmL},\lambda_{\rmR}$ values.

For $\lambda_{\min}$, we will similarly pick an upper bound for $\frac{\|\hat\bbeta_{\rmL}(\lambda_{\rmL})\|_0}{n}$ and $\frac{1}{n}\operatorname{Tr}((\frac{1}{n} \bX^\top \bX + \lambda_{\textrm{R}} \bm{I})^{-1}\frac{1}{n} \bX^\top \bX)$. Due to numerical precision/stability issue, \emph{glmnet} (and possibly other numerical solvers) yields incorrect solutions for tiny values of $\lambda$. Also, the lower bound $\lambda_{\min}(\sigma_{\max}^2,\delta)$-although unknown-prohibits tiny values of $\lambda$. Therefore, simply choosing an upper bound such as $t_{\max}=0.99$ does not suffice. Heuristically, we find any value from $0.3$ to $0.7$ acceptable, and we take $t_{\max}=0.5$ for simplicity.

Once the range for $\lambda_{\rmL},\lambda_{\rmR}$ are determined, we discretize it on the log scale with grid width $0.1$. This yields the final discrete collection of $\lambda_{\rmL},\lambda_{\rmR}$, from which we minimize $\hat\tau_c^2$.

\subsection{Covariance Estimation}
In Section \ref{sec:cov}, we discussed estimating the blocked population covariance matrix by block-wise sample covariance matrices. The following proposition supports its consistency under mind conditions:

\begin{prop}\label{examples_covariance_estimation}
Let $\bSigma \in \R^{p \times p}$ be a real symmetric matrix whose spectral distribution has bounded support $\mu_{\bSigma} \subset [\kappa_{\min},\kappa_{\max}]$ with $0 < \kappa_{\min} \leq \kappa_{\max} < \infty$. Further suppose that $\bSigma$ is block-diagonal with blocks $\bSigma_1,...,\bSigma_k$ such that the size of each block $\bSigma_i$ is bounded by some constant $m$. Let $\bX \in \R^{n \times p}$ with each row $\bx_{i\bullet}$ satisfying $\bx_{i\bullet} = \bz_{i\bullet}\bSigma^{1/2}$ where $\bz_{i\bullet}$ has independent sub-Gaussian entries with bounded sub-Gaussian norm. If we estimate $\hat\bSigma$ as also block-diagonal with blocks $\hat\bSigma_1,...,\hat\bSigma_k$ where $\hat\bSigma_i=\frac{1}{n}\bX_i^\top\bX_i$ are corresponding sample covariances, then as $p/n\rightarrow \delta \in (0,\infty)$, we have $\|\hat\bSigma-\bSigma\|_{\op} \convP 0$.
\end{prop}
\begin{proof}
The proof is straightforward by applying \cite[Theorem 6.5]{wainwright2019high} on each block component, in conjunction with a union bound. Thus we omit it here.  
\end{proof}

\section{Additional Simulation Details}
\subsection{Random effects comparison: Heritability definition}\label{subsec:random_simulation_heritability_definition}
The common setting in Section \ref{section:simulations_random} assumes both $\bX$ and $\bbeta$ are zero-mean random and that $\bX$ has population covariance $\bSigma$. Under random effects setting, a random $\bX$ is allowed, and the denominator in the heritability definition \eqref{def:heritability} reads
$$\Var(\bx_{i\bullet}^{\top}\bbeta)
= \E[\Var[\bx_{i\bullet}^{\top}\bbeta|\bbeta]] + \Var[\E[\bx_{i\bullet}^{\top}\bbeta|\bbeta]]
=\E[\bbeta^\top \bSigma\bbeta],$$
which is a well-defined population quantity. Yet, $\E[\bbeta^\top \bSigma\bbeta]$ does not have a closed form formula if $\bbeta$ is generated with stratifications as in Section \ref{section:simulations_random}, so we approximated it with $100$ iterations.

Under fixed effect setting, conditioning on a given $\bbeta$, the same definition reads
$$\Var(\bx_{i\bullet}^{\top}\bbeta)
= \bbeta^\top \bSigma\bbeta,$$
which is an empirical quantity that fluctuates around the population quantity $\E[\bbeta^\top \bSigma\bbeta]$. Now with $100$ random $\bbeta$, estimating the corresponding $100$ individual $\bbeta^\top \bSigma\bbeta$ approximately constitutes as estimating its expectation. This, therefore, facilitates fair comparisons between random and fixed effects methods. 

\subsection{Random effects comparison: Signal generation details}\label{subsec:random_simulation_details}
To generate $\bbeta$ with varying non-zero entry locations and distributions, we employed the following generation process. 1) We fixed two levels of concentration: $c_{l} = 0.05$ and $c_h = 0.5$. 2) We determine $k_s$, the number of stratifications needed. For uniformly distributed non-zero entries $k_s=1$. For LDMS stratification, $k_s=12$: the cross product of $4$ LD stratifications and $3$ MAF stratifications mentioned in Section \ref{section:simulations_random}. For LDMS plus block stratification, $k_s=24$: the cross product of the previous $12$ LDMS stratifications and $2$ blocks: the last LD block vs the rest, specified in \cite{berisa2016approximately}. 3) For each of the $k_s$ stratifications determined, we alternatively assign $c_l$ and $c_h$ as the concentration, and generate non-zero entries uniformly randomly with the assigned concentration. In the special case $k_s=1$, $c_l$ is assigned. 4) After selecting non-zero entry locations, count the number of non-zero entries $K$, and then calculate the entrywmeshyise variance $\sigma_+^2 = h^2/K$ where $h^2$ is the desired heritability value. 5) Lastly, pick the non-zero entry distribution. For random normal entries, the distribution is $\mcn(0,\sigma_+^2)$. For mixture-of-normal entries, the distribution is $\mcn(\pm\frac{\sigma_+}{\sqrt{10}},\frac{9\sigma_+^2}{10})$ (an equal mixture of two symmetric normals).

\subsection{Fixed effects comparisons: Unbiasedness and a closer look}\label{sec:bias_illustration}
In this section, we present two plots. The first investigates the bias properties of different fixed effects methods. For a representative example, we chose a setting with $n=1000, p=10000, h^2=0.1$, and $\kappa=0.003$, though this selection is not particularly unique. Using a specific $\bbeta$, we generated 100 random instances of $\bX$. The box plots in Figure \ref{fig:bias_illustration} depict the heritability estimates from all the methods under consideration. We observe that every method is (approximately) unbiased, which was confirmed by in additional settings. We also note that truncating negative estimates to $0$ would yield lower MSEs but more bias, so we  opted to keep the negative estimates.

\begin{figure}[ht]
    \centering
    \includegraphics[width=0.6\linewidth]{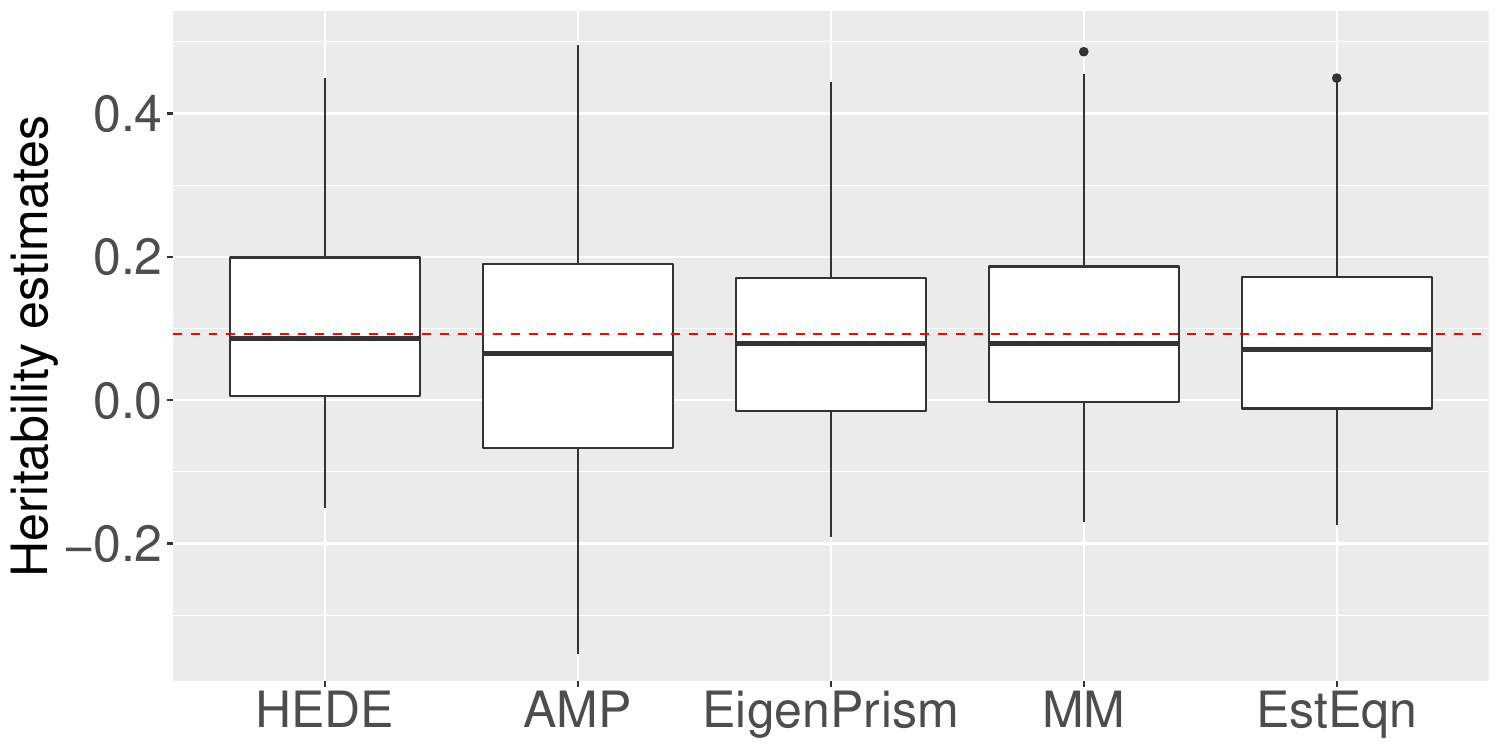}
    \caption{Same setting as Figure \ref{fig:change_eps_h_10000_com}, showing box plots of estimates for a certain realized signal with $100$ draws of design matrices. All methods are unbiased.}
    \label{fig:bias_illustration}
\end{figure}
As a second line of investigation, we recall Figure \ref{fig:change_eps_h_10000_com} from the main manuscript, which showed that the AMP MSEs are often much higher than the remaining. In this light, we zoom into Figure \ref{fig:change_eps_h_10000_com} to obtain a clearer picture for the performance of our method relative to others. Figure \ref{fig:zoomedin} shows this zoomed in version. Note that our relative superiority is mostly maintained across diverse settings. 
\begin{figure}[ht]
    \centering
    \includegraphics[width=\linewidth]{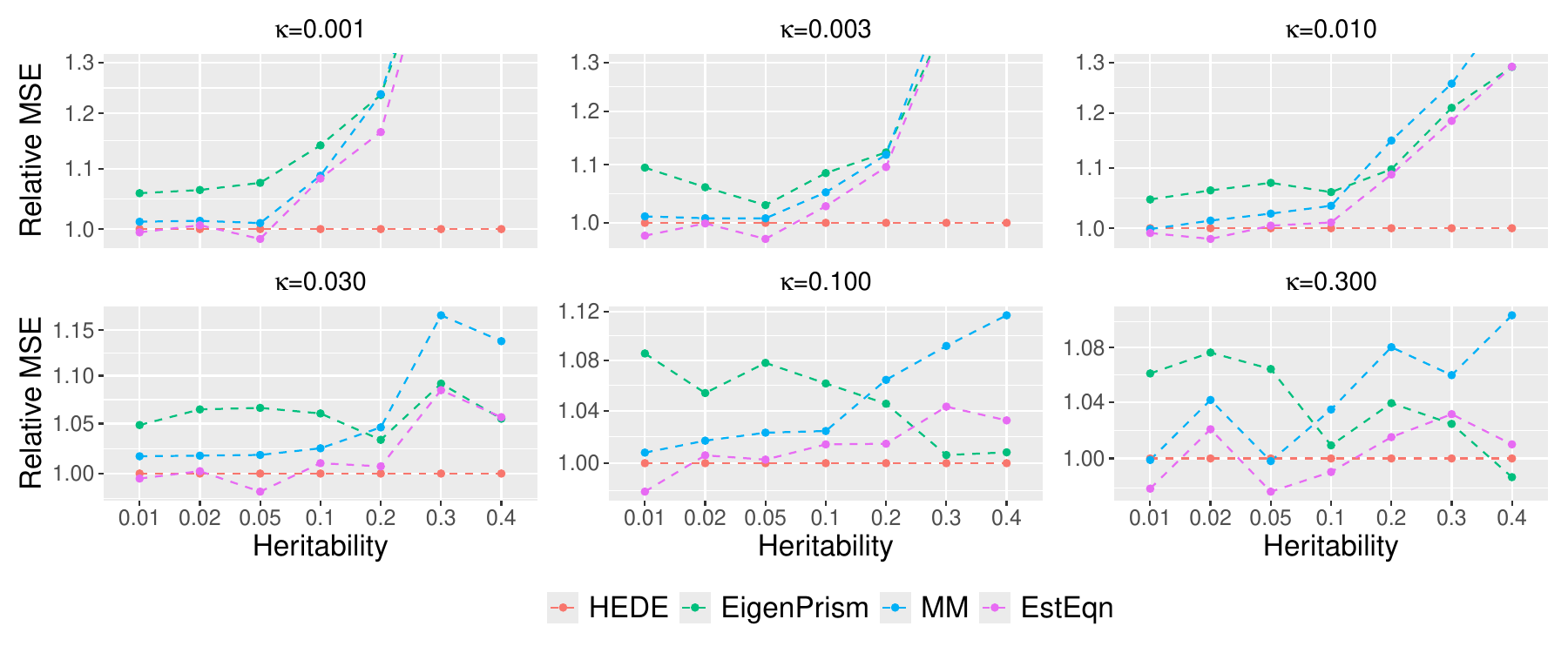}
    \caption{Zoomed in of Figure \ref{fig:change_eps_h_10000_com}.}
    \label{fig:zoomedin}
\end{figure}
\subsection{Fixed effects comparison: Designs with larger sub-Gaussian norms}\label{sec:more_simulations}
To generate $\bX$ with larger sub-Gaussian norms, we followed the exact same setup as in Section \ref{section:simulations_fixed}, except that $\bar{G}_j \sim \text{Unif}[0.005,0.01]$ instead of $[0.01,0.5]$ in \eqref{eqn:X_normalization}. This mimicks lower-frequency variants, leading to larger sub-Gaussian norms in the normalized $\bX$. Comparing Figures \ref{fig:change_eps_h_10000_low},  \ref{fig:change_n_p_0100_05_low} with Figures \ref{fig:change_eps_h_10000_com}, \ref{fig:change_n_p_0100_05_com}, we observe similar trends across the board, with HEDE having either similar or dominating performance.

\begin{figure}[ht]
    \centering
    \includegraphics[width=\linewidth]{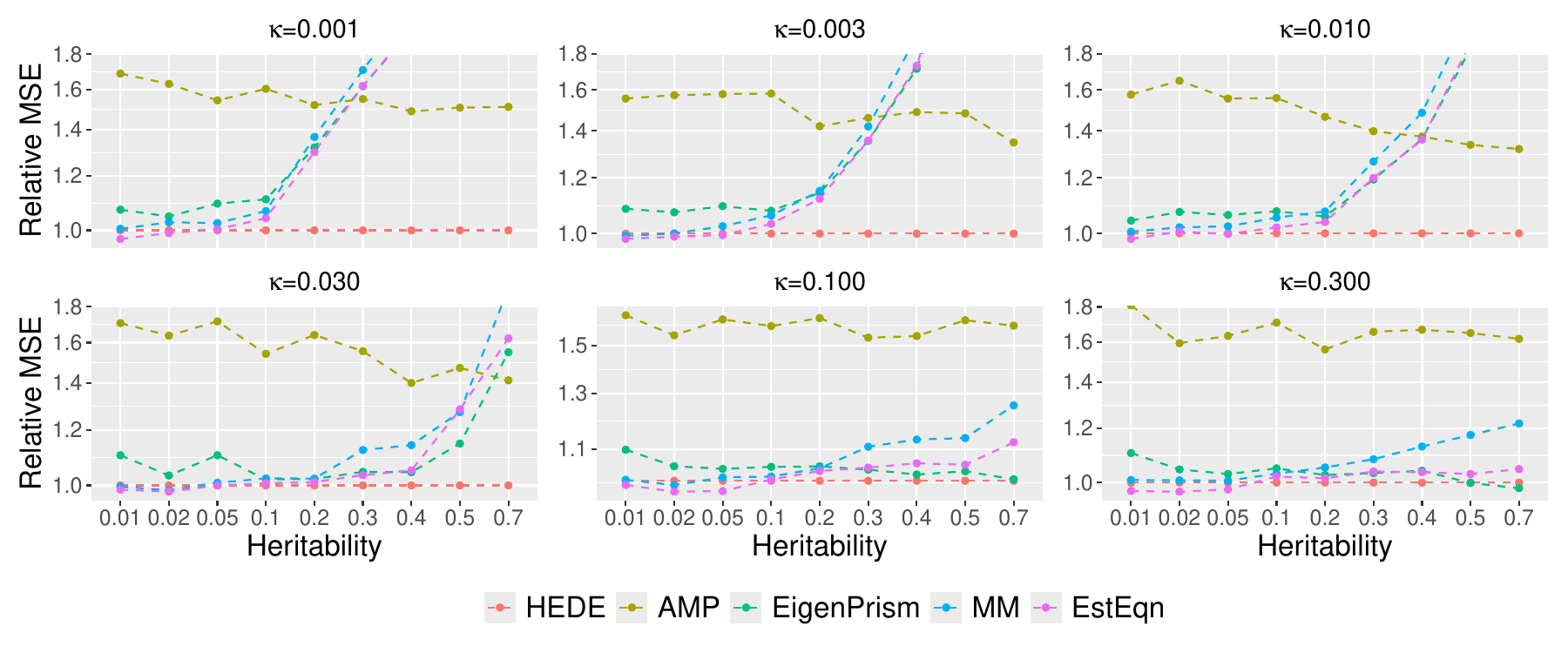}
    \caption{Same setting as Figure \ref{fig:change_eps_h_10000_com}, except that $\pi_j \sim \text{Unif}[0.005,0.01]$.}
    \label{fig:change_eps_h_10000_low}
\end{figure}

\begin{figure}[ht]
    \centering
    \includegraphics[width=\linewidth]{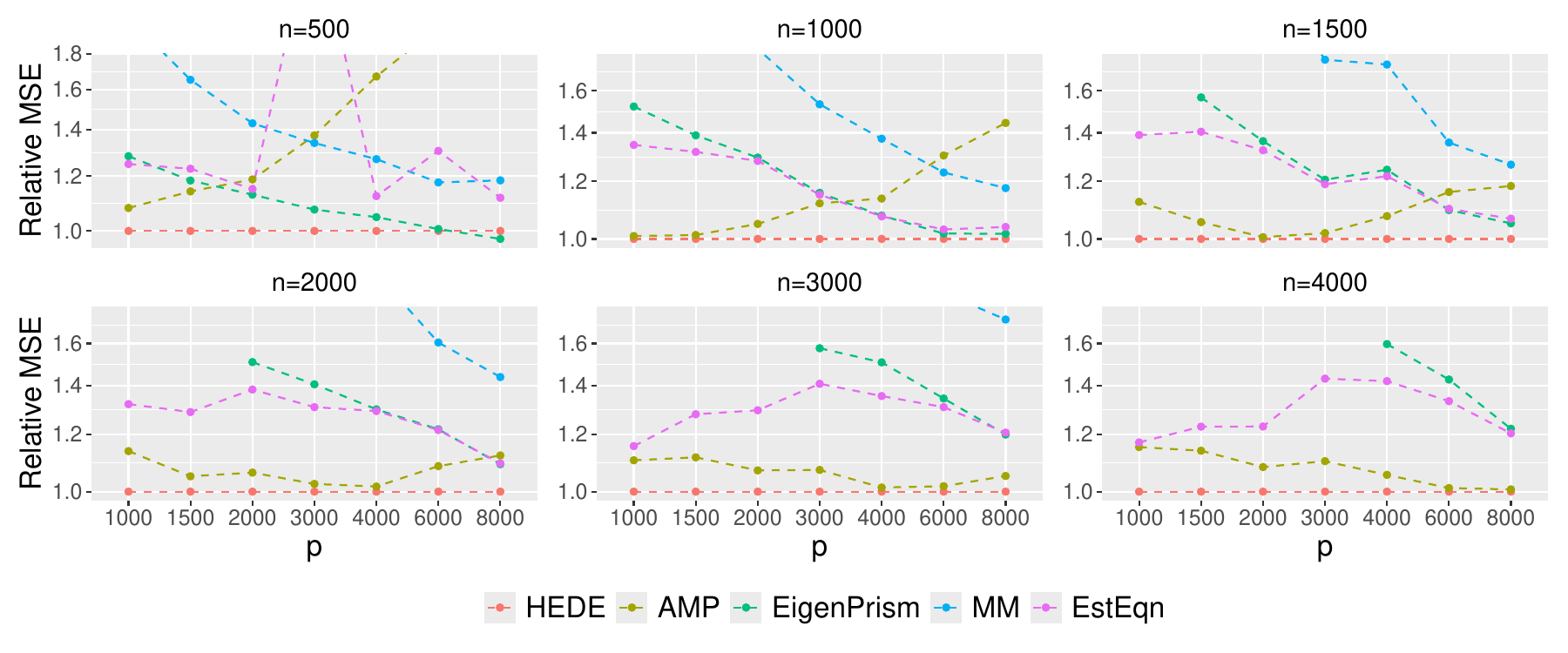}
    \caption{Same setting as Figure \ref{fig:change_n_p_0100_05_com}, except that $\pi_j \sim \text{Unif}[0.005,0.01]$.}
    \label{fig:change_n_p_0100_05_low}
\end{figure}

\section{Proof Notations and Conventions}
This section introduces notations that we will use through the rest of the proof. We consider both the Lasso and ridge estimators in the context of linear models. Several important quantities in our calculations appear in two versions---one computed based on the ridge and the other based on the Lasso. We use subscripts $\textrm{R}, \textrm{L}$ respectively to denote the versions of these quantities corresponding to the ridge and the Lasso. Recall that the setting we study may be expressed via the linear model $\by = \bX \bbeta + \sigma\bz$ where $\by \in \R^n$ denote the responses, $\bX \in \R^{n\times p}$ the design matrix, $\bbeta \in \R^p$ the unknown regression coefficients, and $\bz$ the noise. Then for $k= \textrm{L}, \textrm{R}$, the corresponding Lasso and Ridge estimators $\hat\bbeta_k \in \R^p$ are defined as
\begin{equation}\label{def:hat_beta}
    \hat \bbeta_k = \argmin_{\bb} \left \{ \frac{1}{2n} \|\by - \bX \bb\|_2^2 + \Omega_k(\bb) \right\},
\end{equation}
where $\Omega_{\textrm{L}}(\bb) = \frac{\lambda_{\textrm{L}}}{\sqrt{n}}  \|\bb\|_1, \Omega_{\textrm{R}}(\bb) = \frac{\lambda_{\textrm{R}}}{2} \|\bb\|_2^2$.

Our methodology relies on debiased versions of these estimators, defined as  
\begin{equation}\label{def:hat_beta_d}
        \hat\bbeta_k^d = \hat \bbeta_k + \frac{\bX^\top (\by - \bX \hat \bbeta_k)}{n - \hat\df_k}, \,\, \text{for} \,\, k= \textrm{L}, \textrm{R},
\end{equation}
where  $\hat\df_k \in \R$ are terms that may be interpreted as degrees-of-freedom, and are defined in each of the aforementioned cases as follows,
\begin{equation}\label{def:hat_df}
    \begin{aligned}
        \hat{\text{df}}_{\textrm{L}} &= \|\hat \bbeta_{\rmL}\|_0\\
        \hat{\text{df}}_{\textrm{R}} &= \text{Tr}((\frac{1}{n} \bX^\top \bX + \lambda_{\textrm{R}} \bm{I})^{-1}\frac{1}{n} \bX^\top \bX).
    \end{aligned}
\end{equation}
Our proofs use some intermediate quantities that we introduce next. First, we will require intermediate quantities that replace $\hat{\textrm{df}}_k$ by a set of parameters $\textrm{df}_k$ that rely on underlying problem parameters. For $k=\textrm{L},\textrm{R}$, we define these to be \begin{equation}\label{def:tilde_beta_d}
    \tilde\bbeta_k^d = \hat \bbeta_k + \frac{\bX^\top (\by - \bX \hat \bbeta_k)}{n - \df_k}.
\end{equation}
The exact definition of $\df_k$ is involved, so we defer its presentation to \ref{def:fixed_point}.
Second, we need another set of intermediate quantities $\hat\bbeta_k^f$ that form solutions to optimization problems of the form \eqref{def:hat_beta}, but where the observed data $\{\bm{y},\bm{X}\}$ is replaced by random variables $ \hat\bbeta_k^{f,d}$ that can be expressed as gaussian perturbations of the true signal, with appropriate adjustments to the penalty function. For $k=\textrm{L},\textrm{R}$, these are defined as follows:
\begin{equation}\label{def:hat_beta_f}
    \hat\bbeta_k^f := \eta_k(\hat\bbeta_k^{f,d};\zeta_k) := \argmin_{\bb} \left\{\frac{1}{2}\|\hat\bbeta_k^{f,d} - \bb\|_2^2 + \frac{1}{\zeta_k} \Omega_k(\bb) \right\}, 
\end{equation}
\begin{equation}\label{def:hat_beta_fd}
\text{where} \quad     \hat\bbeta_k^{f,d} = \bbeta + \bg_k^f, \quad (\bg_{\textrm{L}}^f,\bg_{\textrm{R}}^f) \sim \mcn(\bm{0},\bS \otimes \bm{I}_p),
\end{equation}
with $\bS$ of the form
    
\begin{equation}\label{def:bS_g}
    \bS := \begin{pmatrix}\tau_{\textrm{L}}^2 & \rho\tau_{\textrm{L}}\tau_{\textrm{R}} \\ \rho\tau_{\textrm{L}}\tau_{\textrm{R}} & \tau_{\textrm{R}}^2\end{pmatrix},
\end{equation}
for suitable choices of $\tau_\textrm{L},\tau_{\textrm{R}}, \rho,\zeta_k$ that we define later in \eqref{def:fixed_point}.  
Our exposition so far refrains from providing additional insights regarding the necessity of these intermediate quantities--however, the role of these quantities will unravel in due course through the proof. As an aside, note that  the randomness in $\hat\bbeta_k^{f,d},\hat\bbeta_k^{f}$ comes from $\bg_{\textrm{k}}^f$, which is independent of the observed data. We use the superscript `f' (standing for fixed) to denote that these do not depend on our observed data.

Our aforementioned notations are complete once we define the parameters $\textrm{df}_k$ from \eqref{def:tilde_beta_d}, $\zeta_k$ from \eqref{def:hat_beta_f}, and $\bS$ from \eqref{def:bS_g}. We define these as solutions to the following system of equations in the variables $\{ \check{\bS},\check\zeta_{k},\check\df_k, k=\textrm{L},\textrm{R}\}$. 
\begin{equation}\label{def:fixed_point}
    \begin{aligned}
    \check{\bS} &= \frac{1}{n} (\sigma^2 \bI + \E[(\eta_{\rmL}(\hat\bbeta_{\textrm{L}}^{f,d};\check\zeta_{\rmL})-\bbeta,\eta_{\rmR}(\hat\bbeta_{\textrm{R}}^{f,d};\check\zeta_{\rmR})-\bbeta)^\top(\eta_{\rmL}(\hat\bbeta_{\textrm{L}}^{f,d};\check\zeta_{\rmL})-\bbeta,\eta_{\rmR}(\hat\bbeta_{\textrm{R}}^{f,d};\check\zeta_{\rmR})-\bbeta)])\\
    \check\df_k &= \E[\text{div}\eta_k(\hat\bbeta_k^{f,d};\check\zeta_k)],k=\rmL,\rmR\\
    \check\zeta_k &= 1 - \frac{\check\df_k}{n},k=\rmL,\rmR,\\
    \end{aligned}
\end{equation}
where $\text{div}\eta_k(\hat\bbeta_k^{f,d};\check\zeta_k) := \sum_{j=1}^p \frac{\partial}{\partial \hat\bbeta_{k,j}^{f,d}}\eta_k(\hat\bbeta_k^{f,d};\check\zeta_k)$ is defined as the divergence of $\eta_k(\hat\bbeta_k^{f,d};\check\zeta_k)$ with respect to its first argument and recall that we defined $\eta_k(\cdot,\cdot)$ in \eqref{def:hat_beta_f}. 

Extracting the first and the last entries of the first equation in \eqref{def:fixed_point}, we observe that $\tau_{\rmL},\zeta_\rmL,\df_\rmL$ depends on $\lambda_\rmL$ and not on $\lambda_\rmR$, and similarly for the corresponding Ridge parameters. We also note that $\tau_{\rmL\rmR}$ depends on both $\lambda_{\rmL}$ and $\lambda_{\rmR}$. We will use this observation multiple times in our proofs later. 

The system of equations \ref{def:fixed_point} first arose in the context of the problem studied in \cite{celentano2021cad} and Lemma C.1 from the aforementioned paper establishes the uniqueness of the solutions. Furthermore, it follows that 
there exist positive constants $ (\tau_{\min},\tau_{\max},\zeta_{\min},\rho_{\max})$ such that
\begin{equation}\label{lemma:existence_fixed_point}
\tau_{\min}^2+\sigma^2 < n\tau_{k}^2 < \tau_{\max}^2+\sigma^2, \zeta_{\min} < \zeta_k \leq 1, |\rho| < \rho_{\max} < 1,
\end{equation}
for $k=L,R$ where $(\tau_L,\tau_R,\zeta_L,\zeta_R, \rho)$ denotes the unique solution to equation \eqref{def:fixed_point}.

Whenever we mention constants in our proof below, we mean values that only depend on the model parameters $\{\kappa_{\min},\kappa_{\max},\delta,\sigma_{\max},\lambda_{\min},\lambda_{\max}\}$ laid out in Assumptions \ref{assumptions} and \ref{assumptionL}, and do not depend on any other variable (especially $n,p$, and realization of any random quantities). 

Finally in Sections \ref{subsec:proof:debiased_uniform}--\ref{sec:dependencies}, we prove our results under the stylized setting where the design matrix entries are i.i.d. $\mcn(0,1)$ and the noise vector entries are i.i.d. $\mcn(0,1)$. In Section \ref{sec:universality} we establish universality results that show that the same conclusions hold under the setting of Assumptions \ref{assumptions} and \ref{assumptionL}, where the covariate and error distributions are more general. To avoid confusion, we define below a stylized version of Assumptions \ref{assumptions} and \ref{assumptionL}, where everything remains the same except the design, error distributions are taken to be i.i.d. Gaussian.
So, Sections \ref{subsec:proof:debiased_uniform}-\ref{sec:dependencies} work under Assumption \ref{assumptions2}, while Section \ref{sec:universality} works under  Assumptions \ref{assumptions} and \ref{assumptionL}.
\begin{assumption}\label{assumptions2}
    \ 
\begin{enumerate}
    \item Same as Assumption \ref{assumptions}(1).
    \item Each entry of $\bX$ satisfies $X_{ij}\iid \mcn(0,1)$. 
    \item Same as Assumption \ref{assumptions}(3).
    \item The noise $\bep$ satisfies $\epsilon_i \iid \mcn(0,\sigma^2)$, with $\sigma^2 \leq \sigma_{\max}^2$.
    \item Same as Assumption \ref{assumptionL}(2).
\end{enumerate}
\end{assumption}

Note that we don't need Assumption \ref{assumptionL}(1) since it is true for Gaussians (see B.5.4 in \cite{celentano2020lasso}).

\section{Proof of Theorem \ref{main:thm:debiased_uniform}}\label{subsec:proof:debiased_uniform}

Theorem \ref{main:thm:debiased_uniform} forms the backbone of our results so we begin by presenting its proof here. 
Towards this goal, we first prove the following slightly different version.

\begin{thm}\label{main:thm:debiased_uniform_tilde}
Suppose that Assumption \ref{assumptions2} holds. Then for any $1$-Lipschitz function $\phi_\beta:(\R^p)^3\rightarrow \R$, we have 
$$\sup_{\lambda_{\textrm{L}},\lambda_{\textrm{R}} \in [\lambda_{\min},\lambda_{\max}]}|\phi_\beta(\tilde\bbeta_{\textrm{L}}^d,\tilde\bbeta_{\textrm{R}}^d,\bbeta) - \E[\phi_\beta(\hat\bbeta_{\textrm{L}}^{f,d},\hat\bbeta_{\textrm{R}}^{f,d},\bbeta)]| \convP 0,$$
where $\tilde\bbeta_{k}^d,\hat\bbeta_{k}^{f,d}$ are as defined in \eqref{def:tilde_beta_d}, \eqref{def:hat_beta_fd} for $k = \textrm{L},\textrm{R}$ respectively. 
\end{thm}

Theorem \ref{main:thm:debiased_uniform_tilde} differs from Theorem \ref{main:thm:debiased_uniform} in that $\hat\bbeta_k^d$ is now replaced by $\tilde{\bbeta}_k^d$. The difference between these lies in the fact that $\hat{\textrm{df}}_k$ is replaced by $\textrm{df}_k$ from the former to the latter (c.f.~Eqns \ref{def:hat_beta_d} and \ref{def:tilde_beta_d}).

Given Theorem \ref{main:thm:debiased_uniform_tilde}, the proof of Theorem \ref{main:thm:debiased_uniform} follows a two step procedure: we first establish that the empirical quantities $\hat{\textrm{df}}_k$ and the parameters 
$\textrm{df}_k$ 
are uniformly close asymptotically. This allows us to establish that $\tilde{\bbeta}_k^d$ and $\hat{\bbeta}_k^d$ are asymptotically close. Theorem \ref{main:thm:debiased_uniform} thereby follows from here, when applied in conjunction with Theorem \ref{main:thm:debiased_uniform_tilde} and the fact that $\phi_\beta$ is Lipschitz. We formalize these arguments below. 

\begin{lemma}\label{lemma:df_diff}
Recall the definitions of $\hat\df_k$ and $\df_k$ from \eqref{def:hat_df}, \eqref{def:fixed_point}.
Under Assumption \ref{assumptions2}, for $k=\rmL,\rmR$,
\begin{equation*}
\suplk \left|\frac{\hat\df_k}{p} - \frac{\df_k}{p}\right| \convP 0.
\end{equation*}
\end{lemma}

\begin{proof}
\cite[Theorem F.1.]
{miolane2021distribution} established the aformentioned for $k=\textrm{L}$.

For $k=R$, \cite[Lemma H.1]{celentano2021cad} (which further cites \cite[Theorem 3.7]{knowles2017anisotropic}) showed that that $\hat\df_\rmR/p$ converges to $\df_\rmR/p$ for any fixed $\lambda_\rmR$. So for our purposes it suffices to extend this proof to all $\lambda_\rmR \in [\lambda_{\min},\lambda_{\max}]$ simultaneously. We achieve this below.

Recall from Eqn. \ref{def:fixed_point} that $\hat\df_{\textrm{R}}/p = 1-\lambda_\rmR\text{Tr}(\frac{1}{p}(\frac{1}{n} \bX^\top \bX + \lambda_{\textrm{R}} \bm{I})^{-1})$, where the trace on the right hand side is the negative resolvent 
of $\frac{1}{n} \bX^\top \bX$ evaluated at $-\lambda_\rmR$ and normalized by $1/p$. \cite[Theorem 3.7]{knowles2017anisotropic} (with some notation-transforming algebra) states that this negative normalized resolvent converges in probability to $(1-\df_\rmR/p)/\lambda_\rmR$ with fluctuation of level $O(N^{-1/2}\kappa^{-1/4})$, where $\kappa = \lambda_{\rmR} + 1$ is the distance from $-\lambda_{\rmR}$ to the spectrum of $\bm{I}$. Therefore, $\hat\df_{\rmR}/p$ concentrates around $\df_{\rmR}/p$ with fluctuation of level $O(N^{-1/2}\lambda_{\rmR}(1+\lambda_{\rmR})^{-1/4}) = O(N^{-1/2}(1+\lambda_{\max})^{3/4})$, for all $\lambda_{\rmR} \in [\lambda_{\min},\lambda_{\max}]$. 
\end{proof}

As a direct consequence, we have

\begin{cor}\label{cor:df_bounded}
Under Assumption \ref{assumptions2}, we have for $k=\rmL,\rmR$,
\begin{equation*}
\begin{aligned}
        \suplk \left|1-\frac{\df_k}{n}\right| &=\Theta_p(1),\\
        \suplk \left|1-\frac{\hat{\df}_k}{n}\right| &=\Theta_p(1).
\end{aligned}
\end{equation*}
\end{cor}

\begin{proof}
The first line follows directly from \eqref{def:fixed_point} and the fact that $\zeta_k$ is bounded by \ref{lemma:existence_fixed_point}.
For the second line, we have
$$\suplk \left|1-\frac{\hat{\df}_k}{n}\right| \leq \suplk \left|1-\frac{{\df}_k}{n}\right| + \suplk \left|\frac{\hat\df_k}{n} - \frac{\df_k}{n} \right|.$$
where the first term is $\Theta_p(1)$ and the second term is $o_P(1)$ by Lemma \ref{lemma:df_diff}.
\end{proof}

From here, one can derive the following Lemma using the definitions of $\hat\bbeta_k^{d}$, $\tilde\bbeta_k^d$. One should compare the following to \cite[Lemma H.1.ii]{celentano2021cad} that proved a pointwise version of this result without any supremum over the tuning parameters. 

\begin{lemma}\label{lemma:beta_emp_diff}
Under Assumption \ref{assumptions2}, for $k=\rmL,\rmR$, 
\begin{equation*}
    \suplk\|\hat\bbeta_k^{d} - \tilde\bbeta_k^d\|_2 \convP 0.
\end{equation*}
\end{lemma}

\begin{proof}
By definition,
\begin{align*}
    &\suplk\|\hat\bbeta_k^{d} - \tilde\bbeta_k^d\|_2 \\
    =&\suplk \left\|\frac{\bX^\top(\by-\bX\hat\bbeta_k)}{n-\hat\df_k} - \frac{\bX^\top(\by-\bX\hat\bbeta_k)}{n-\df_k}\right\|_2\\
    \leq & \frac{\|\bX\|_{\op}}{\sqrt{n}} \suplk\frac{\|\by-\bX\hat\bbeta_k\|_2}{\sqrt{n}} \suplk\left|\frac{1}{{1-\hat\df_k/n}}-\frac{1}{{1-\df_k/n}}\right|.
\end{align*}
By Corollary \ref{largest_singular_value}, $\|\bX\|_{\op}/\sqrt{n}$ is bounded with probability $1-o(1)$. 
By Lemma \ref{lemma:uvw_convergence}, $\suplk\|\by-\bX\hat\bbeta_k\|/\sqrt{n}$ is bounded with probability $1-o(1)$. Lastly,
$$\suplk\left|\frac{1}{{1-\hat\df_k/n}}-\frac{1}{{1-\df_k/n}}\right|
    =\suplk \left|\frac{\df_k/n-\hat\df_k/n}{(1-\hat\df_k/n)(1-\df_k/n)}\right| \convP 0,$$
where the numerator is $o_P(1)$ by Lemma \ref{lemma:df_diff} and the denominator is $\Theta_P(1)$ by Corollary \ref{cor:df_bounded}.
\end{proof}

We next turn to prove Theorem \ref{main:thm:debiased_uniform}.

\begin{proof}[Proof of Theorem \ref{main:thm:debiased_uniform}]
By triangle inequality,
\begin{align*}
     &\sup_{\lambda_{\textrm{L}},\lambda_{\textrm{R}} \in [\lambda_{\min},\lambda_{\max}]}|\phi_\beta(\hat\bbeta_{\textrm{L}}^d,\hat\bbeta_{\textrm{R}}^d,\bbeta) - \E[\phi_\beta(\hat\bbeta_{\textrm{L}}^{f,d},\hat\bbeta_{\textrm{R}}^{f,d},\bbeta)]| \\
     \leq &\sup_{\lambda_{\textrm{L}},\lambda_{\textrm{R}} \in [\lambda_{\min},\lambda_{\max}]}|\phi_\beta(\tilde\bbeta_{\textrm{L}}^d,\tilde\bbeta_{\textrm{R}}^d,\bbeta) - \E[\phi_\beta(\hat\bbeta_{\textrm{L}}^{f,d},\hat\bbeta_{\textrm{R}}^{f,d},\bbeta)]|  + \sup_{\lambda_{\textrm{L}},\lambda_{\textrm{R}} \in [\lambda_{\min},\lambda_{\max}]}|\phi_\beta(\hat\bbeta_{\textrm{L}}^d,\hat\bbeta_{\textrm{R}}^d,\bbeta) - \phi_\beta(\tilde\bbeta_{\textrm{L}}^d,\tilde\bbeta_{\textrm{R}}^d,\bbeta)| \\
     \leq & o_P(1) + \suplb \|\hat\bbeta_{\textrm{L}}^d - \tilde\bbeta_{\textrm{L}}^d\|_2 + \suplb \|\hat\bbeta_{\textrm{R}}^d - \tilde\bbeta_{\textrm{R}}^d\|_2\\
     \convP & \  0,
\end{align*}
where we used Theorem \ref{main:thm:debiased_uniform_tilde} and Lemma \ref{lemma:beta_emp_diff}, as well as the fact that $\phi_\beta$ is $1$-Lipschitz.
\end{proof}

Thus, our proof of Theorem \ref{main:thm:debiased_uniform} is complete if we  prove Theorem
\ref{main:thm:debiased_uniform_tilde}. We present this in the next sub-section.

\subsection{Proof of Theorem \ref{main:thm:debiased_uniform_tilde}}
The overarching structure of our proof of Theorem \ref{main:thm:debiased_uniform_tilde} is inspired by \cite[Section D]{celentano2021cad}. However, unlike in our setting below, the results in the aforementioned paper do not require uniform convergence over a range of values of the tuning parameter. This leads to novel technical challenges in our setting that we handle as we proceed. 

To prove Theorem \ref{main:thm:debiased_uniform_tilde}, we introduce an intermediate quantity $\E[\phi_\beta (\hat\bbeta_{\rmL}^{f,d},\hat\bbeta_{\textrm{R}}^{f,d},\bbeta)|\bg_{\textrm{L}}^f = \hat\bg_{\textrm{L}}] $ that conditions on $\bg_{\textrm{L}}^f$ (recall the definition from \eqref{def:hat_beta_fd}) taking on a specific value $\hat\bg_{\textrm{L}}$ defined as follows:
\begin{equation}\label{def:hat_gl}
\begin{gathered}
    \hat\bg_{\textrm{L}} = \frac{n\tau_{\textrm{L}}\zeta_{\textrm{L}}\sqrt{n\tau_{\textrm{L}}^2-\sigma^2}}{\|\by - \bX\hat\bbeta_{\textrm{L}}\|_2\|\hat\bbeta_{\textrm{L}}-\bbeta\|_2}(\hat\bbeta_{\textrm{L}}-\bbeta) + \frac{\tau_{\textrm{L}}}{\|\by - \bX\hat\bbeta_{\textrm{L}}\|_2}\bX^\top(\by - \bX\hat\bbeta_{\textrm{L}}),
\end{gathered}
\end{equation}
where $\tau_\textrm{L},\zeta_{\textrm{L}}$ are defined as in \eqref{def:fixed_point}. Recall from \eqref{lemma:existence_fixed_point} that 
 $n \tau_{\textrm{L}}^2 \geq \sigma^2$ so that the square root above is well-defined. 
 
\begin{remark}\label{remark_gl}
The definition of $\hat\bg_{\textrm{L}}$ is non-trivial.
However, the takeaway is that the realization $\bg_{\rmL}^f=\hat\bg_{\rmL}$ should be understood as the coupling of $\bX$ and $\bg_{\rmL}^f$ such that $\hat\bbeta_{\rmL}^f,\hat\bbeta_{\rmL}^{f,d}$ equals $\hat\bbeta_{\rmL},\tilde\bbeta_{\rmL}^d$, respectively. We refer readers to \cite[Section F.1 and Section L]{celentano2021cad} for the underlying intuition as well as the proof of this equivalence.
\end{remark}

Since $\hat\bbeta_{\rmL}^{f,d}$ becomes $\tilde\bbeta_{\rmL}^d$ conditioning on $\bg_{\rmL}^f = \hat\bg_{\rmL}$, we have $\E[\phi_\beta (\hat\bbeta_{\rmL}^{f,d},\hat\bbeta_{\textrm{R}}^{f,d},\bbeta)|\bg_{\textrm{L}}^f = \hat\bg_{\textrm{L}}] = \E[\phi_\beta (\tilde\bbeta_{\textrm{L}}^d,\hat\bbeta_{\textrm{R}}^{f,d},\bbeta)|\bg_{\textrm{L}}^f = \hat\bg_{\textrm{L}}]$, where the randomness inside the expectation comes from $\hat\bbeta_{\rmR}^{f,d}$, and $\tilde\bbeta_{\rmL}^d$ is fixed. In fact, analogous to $\bg_{\rmL}^f=\hat\bbeta_{\rmL}^{f,d}-\bbeta$, $\hat\bg_{\textrm{L}}$ approximates $\tilde\bbeta_{\textrm{L}}^d-\bbeta$, as formalized in the result below. 

\begin{lemma}\label{lemma:convergence_gl}
Under Assumption \ref{assumptions2},
\begin{align*}
\sup_{\lambda_{\textrm{L}}\in[\lambda_{\min},\lambda_{\max}]}\|\hat\bg_{\textrm{L}} - (\tilde\bbeta_{\textrm{L}}^d - \bbeta) \|_2 \convP 0.
\end{align*}
\end{lemma}

With these definitions in hand, the proof is complete on combining the following two lemmas. 

\begin{lemma}\label{lemma:bridge_1}
Recall the definitions of $\tilde\bbeta_k^d$ from Eqn. \ref{def:tilde_beta_d} and $\hat\bbeta_k^{f,d}$ from Eqn. \ref{def:hat_beta_fd}. Under Assumption \ref{assumptions2}, for any $1$-Lipschitz function $\phi_\beta:(\R^p)^3\rightarrow \R$,
\begin{equation}\label{eq:interim1}
\begin{gathered}
    \suplb|\E[\phi_\beta (\hat\bbeta_{\textrm{L}}^{f,d},\hat\bbeta_{\textrm{R}}^{f,d},\bbeta)|\bg_{\textrm{L}}^f = \hat\bg_{\textrm{L}}] - \E[\phi_\beta(\hat\bbeta_{\textrm{L}}^{f,d},\hat\bbeta_{\textrm{R}}^{f,d},\bbeta)]| \convP 0.
\end{gathered}
\end{equation}
\end{lemma}

\begin{lemma}\label{lemma:bridge_2}
 Under Assumption \ref{assumptions2}, for any $1$-Lipschitz function $\phi_\beta:(\R^p)^3\rightarrow \R$,
\begin{equation*}
    \suplb|\E[\phi_\beta (\hat\bbeta_{\textrm{L}}^{f,d},\hat\bbeta_{\textrm{R}}^{f,d},\bbeta)|\bg_{\textrm{L}}^f = \hat\bg_{\textrm{L}}] - \phi_\beta(\tilde\bbeta_{\textrm{L}}^{d},\tilde\bbeta_{\textrm{R}}^{d},\bbeta)| \convP 0.
\end{equation*}
\end{lemma}
We present the proofs of these supporting lemmas in the following subsection.

\subsection{Proof of Lemmas \ref{lemma:convergence_gl}--\ref{lemma:bridge_2}}
\label{sec:proof:lemma:bridge_1}
We start with the proof of Lemma \ref{lemma:convergence_gl} since it plays a crucial role in the proofs of the others. To this end, perhaps the most important step is the following lemma that establishes the  boundedness and limiting behavior of certain special entities. 

\begin{lemma}\label{lemma:uvw_convergence}
For $k = \rmL,\rmR$, define $\hat\bw_k,\hat\bu_k,\hat\bv_k$ as
\begin{equation}\label{uvw}
    \hat\bw_k = \hat\bbeta_k - \bbeta, \hat\bu_k = \by - \bX\hat\bbeta_k,\hat\bv_k = \bX^\top(\by-\bX\hat\bbeta_k).
\end{equation}
Under Assumption \ref{assumptions2},
\begin{itemize}
    \item Convergence: 
    \begin{align*}
        &\suplk \left| \|\hat\bw_k\|_2 - \sqrt{n\tau_k^2-\sigma^2} \right| \convP 0.\\
        &\suplk \left| \|\hat\bu_k\|_2/\sqrt{n} - \sqrt{n}\tau_k\zeta_k \right| \convP 0.
    \end{align*}
    \item Boundedness: there exists positive constants $c,C$ such that with probability $1-o(1)$,
    $$c\leq\suplk\|\hat\bw_k\|_2, \suplk\|\hat\bu_k\|_2/\sqrt{n}, \suplk\|\hat\bv_k\|_2/n\leq C.$$
\end{itemize}
\end{lemma}
\begin{proof} 

Lemma \ref{lemma:uvw_convergence} is proved in Section \ref{subsec:proof:lemma:uvw_convergence} 
\end{proof}

We will see below that Lemma \ref{lemma:uvw_convergence} forms crux of the proof of Lemma \ref{lemma:convergence_gl}. 
\subsubsection{Proof of Lemma \ref{lemma:convergence_gl}}\label{subsec:proof:convergence_gl}
\begin{proof}[Proof of Lemma \ref{lemma:convergence_gl}]
\cite[Section J.1]{celentano2021cad} proved a pointwise version of this lemma, without any supremum over the tuning parameters. Thus, with 
Lemma \ref{lemma:uvw_convergence} at our disposal, the proof of Lemma \ref{lemma:convergence_gl} follows by a suitable combination with the strategy in \cite{celentano2021cad}. 
Recall from Eqns. \ref{def:hat_gl} and \ref{uvw}, we have
$$\hat\bg_{\textrm{L}} = \frac{n\tau_{\textrm{L}}\zeta_{\textrm{L}}\sqrt{n\tau_{\textrm{L}}^2-\sigma^2}}{\|\hat\bu_{\textrm{L}}\|_2\|\hat\bw_{\textrm{L}}\|_2}\hat\bw_{\textrm{L}} + \frac{\tau_{\textrm{L}}}{\|\hat\bu_{\textrm{L}}\|_2}\hat\bv_{\textrm{L}}.$$
Next recall the definition of $\tilde\bbeta_{\textrm{L}}^d$ from \eqref{def:tilde_beta_d} and that $n - \df_{\textrm{L}} = n\zeta_{\textrm{L}}$ from our system of equations \eqref{def:fixed_point}. On applying  triangle inequality and Lemma \ref{prop:sup_inequalities}, we obtain that
\begin{align}\label{eq:glhatbound}
    \supll \|\hat\bg_{\textrm{L}} - (\tilde\bbeta_{\textrm{L}}^d - \bbeta_{\textrm{L}})\|_2 & \leq \supll \left| \frac{n\zeta_{\textrm{L}}\tau_{\textrm{L}}\sqrt{n\tau_{\textrm{L}}^2-\sigma^2}}{\|\hat\bu_{\textrm{L}}\|_2\|\hat\bw_{\textrm{L}}\|_2}-1\right|\cdot \supll \|\hat\bw_{\textrm{L}}\|_2 \\
    & + \supll \left| \frac{n\tau_{\textrm{L}}}{\|\hat\bu_{\textrm{L}}\|_2} - \frac{1}{\zeta_{\textrm{L}}}\right| \cdot \supll \frac{\|\hat\bv_{\textrm{L}}\|_2}{n}. \nonumber 
\end{align}
For the first summand, note that by Lemma \ref{lemma:uvw_convergence}, $\supll\|\hat\bw_{\textrm{L}}\|_2 = O_P(1)$. Thus it suffices to show that the first term in the first summand in \eqref{eq:glhatbound} is $o_P(1)$. To see this, observe that we can simplify this as follows
\begin{align*}
    &\supll\left| \frac{n\tau_{\textrm{L}}\zeta_{\textrm{L}}\sqrt{n\tau_{\textrm{L}}^2-\sigma^2}}{\|\hat\bu_{\textrm{L}}\|_2\|\hat\bw_{\textrm{L}}\|_2}-1\right| \\
    =&\supll \left| \frac{(\sqrt{n}\tau_{\textrm{L}}\zeta_{\textrm{L}})\sqrt{n\tau_{\textrm{L}}^2-\sigma^2}}{(\|\hat\bu_{\textrm{L}}\|_2/\sqrt{n})\|\hat\bw_{\textrm{L}}\|_2}-1\right|\\
    =&\supll \left| \frac{(\sqrt{n}\tau_{\textrm{L}}\zeta_{\textrm{L}}) (\sqrt{n\tau_{\textrm{L}}^2-\sigma^2}-\|\hat\bw_{\textrm{L}}\|_2) + \|\hat\bw_{\textrm{L}}\|_2(\sqrt{n}\tau_{\textrm{L}}\zeta_{\textrm{L}} - \|\hat\bu_{\textrm{L}}\|_2/\sqrt{n})}{(\|\hat\bu_{\textrm{L}}\|_2/\sqrt{n})\|\hat\bw_{\textrm{L}}\|_2}\right|\\
    \leq &  \supll \left|\frac{\sqrt{n}\tau_{\textrm{L}}\zeta_{\textrm{L}}}{(\|\hat\bu_{\textrm{L}}\|_2/\sqrt{n})\|\hat\bw_{\textrm{L}}\|_2}\right|\supll \left|\sqrt{n\tau_{\textrm{L}}^2-\sigma^2}-\|\hat\bw_{\textrm{L}}\|_2\right|\\
    &+ \supll \left|\frac{1}{\|\hat\bu_{\textrm{L}}\|_2/\sqrt{n}}\right|\supll \left|(\sqrt{n}\tau_{\textrm{L}}\zeta_{\textrm{L}} - \|\hat\bu_{\textrm{L}}\|_2/\sqrt{n})\right|\convP 0 \,\, \text{by Lemma \ref{lemma:uvw_convergence} }.
\end{align*}
We turn to the second summand in \eqref{eq:glhatbound}. It suffices to show that the first term is $o_P(1)$ since the second  is $O_P(1)$ by Lemma \ref{lemma:uvw_convergence}. But this follows from the same result on observing that
$$\supll \left| \frac{n\tau_{\textrm{L}}}{\|\hat\bu_{\textrm{L}}\|_2} - \frac{1}{\zeta_{\textrm{L}}}\right|
=\supll \left|\frac{\sqrt{n}\tau_{\textrm{L}}\zeta_{\textrm{L}} - \|\hat\bu_{\textrm{L}}\|_2/\sqrt{n}}{\|\hat\bu_{\textrm{L}}\|_2/\sqrt{n} \cdot \zeta_{\textrm{L}}}\right|$$
and that
 $\zeta_L$ is bounded below by a positive constant by \eqref{lemma:existence_fixed_point}. This completes the proof.
\end{proof}

\subsubsection{Proof of Lemma \ref{lemma:bridge_1}}\label{subsec:proof:lemma:bridge_1}
For the purpose of clarity, we sometimes make the dependence of estimators on $\lambda_{\rmL}$ and/or $\lambda_\rmR$ explicit (by writing $\hat\bbeta_{\rmL}(\lambda_\rmL)$, for example).
We first introduce a supporting Lemma that characterizes a convergence result individually for the Lasso and Ridge.

\begin{lemma}\label{lemma:marginal_both}
Under Assumption \ref{assumptions2}, for $k=\rmL,\rmR,$ and any $1$-Lipschitz function $\phi_\beta:(\R^p)^3\rightarrow \R$,
\begin{gather*}  \sup_{\lambda_k\in[\lambda_{\min},\lambda_{\max}]} \left|\phi_\beta(\hat\bbeta_k,\tilde\bbeta_k^d,\bbeta) - \E[\phi_\beta(\hat\bbeta_k^f,\hat\bbeta_k^{f,d},\bbeta)]\right| \convP 0.
\end{gather*}
\end{lemma}
\begin{proof}

First we consider $k=\rmR$. From Lemma \ref{lemma:marginal_single}, we know 
$$\suplr |\phi_\beta(\hat\bbeta_{\rmR}) - \E[\phi_\beta(\hat\bbeta_{\rmR}^{f})]| \convP 0$$ 
for any $1$-Lipschitz function $\phi_\beta$. Since $\tilde\bbeta_{\rmR}^d$ is a Lipschitz function of $\hat\bbeta_{\rmR}$ by Lemma \ref{lemma:debiased_lipschitz_hat}, the proof follows by the definition of these from \eqref{def:hat_beta_fd}.

Now for $k=\rmL$, however, $\tilde\bbeta_{\rmL}^d$ is not necessarily a Lipschitz function of $\hat\bbeta_{\rmL}$. Thus, we turn to the $\alpha$-smoothed Lasso ($\alpha > 0$) instead \cite{celentano2020lasso}, defined as
$$\hat\bbeta_\alpha = \argmin_{\bb} \frac{1}{2n}\|\by - \bX\bb\|_2^2 
 + \frac{\lambda_L}{\sqrt{n}}\inf_{\btheta}\left\{\frac{\sqrt{n}}{2\alpha}\|\bb-\btheta\|_2^2 + \|\btheta\|_1 \right\}.$$
Based on this, 
$\tilde\bbeta_\alpha^d,\hat\bbeta_\alpha^f,\hat\bbeta_\alpha^{f,d}$ can be similarly defined.
We omit the full details for simplicity since they are similar to \eqref{def:hat_beta_d}-\eqref{def:fixed_point}, but we note that they satisfy the following relation by KKT conditions:

\begin{equation*}
\begin{aligned}
    \tilde\bbeta_\alpha^d &= \hat\bbeta_\alpha + \frac{\lambda_{\rmL} \nabla M_\alpha(\hat\bbeta_\alpha)}{\sqrt{n}\xi_\alpha}\\
    \hat\bbeta_\alpha^{f,d} &= \hat\bbeta_\alpha^f + \frac{\lambda_{\rmL} \nabla M_\alpha(\hat\bbeta_\alpha^f)}{\sqrt{n}\xi_\alpha},
\end{aligned}
\end{equation*}
where $M_\alpha(\bb) := \inf_{\btheta}\left\{\frac{\sqrt{n}}{2\alpha}\|\bb-\btheta\|_2^2 + \|\btheta\|_1 \right\}$ and $\bxi_\alpha$ is defined similarly as \eqref{def:fixed_point}.
\cite[Theorem B.1]{celentano2020lasso} establishes a pointwise version of Lemma \ref{lemma:marginal_single} for the $\alpha$-smoothed Lasso, i.e. without supremum over the tuning parameter range, but uniform control results can be easily obtained following our techniques for Lemma \ref{lemma:marginal_single}. Furthermore, \cite[Section B.5.2]{celentano2020lasso} shows that $\tilde\bbeta_\alpha^d$ is an $C/\alpha$-Lipschitz function of $\hat\bbeta_\alpha$ for some positive constant $C$ (by checking the Lipschitzness of $\nabla M_\alpha$).  
Combining these, we obtain that for a fixed $\alpha > 0$ 
\begin{equation}\label{eqn:marginal_both_alpha}
\sup_{\lambda_{\rmL}\in[\lambda_{\min},\lambda_{\max}]} \left|\phi_\beta(\hat\bbeta_\alpha,\tilde\bbeta_\alpha^d,\bbeta) - \E[\phi_\beta(\hat\bbeta_\alpha^f,\hat\bbeta_\alpha^{f,d},\bbeta)]\right| \convP 0. 
\end{equation}
In addition, \cite[Lemma B.8 and Section B.5.1]{celentano2020lasso} 
establishes the closeness of (debiased) Lasso and (debiased) $\alpha$-smoothed Lasso as follows: there exists a constant $\alpha_{\max}>0$ such that
\begin{equation}\label{closeness_alpha_smoothed}
\begin{gathered}
    \PP\left(\|\hat\bbeta_\alpha - \hat\bbeta_{\rmL}\|_2 \leq C_1\sqrt{\alpha},\forall \alpha \leq \alpha_{\max}\right) = 1 - o(1)\\
    \PP\left(\|\hat\bbeta_\alpha^d - \hat\bbeta_{\rmL}^d\|_2 \leq C_1\sqrt{\alpha},\forall \alpha \leq \alpha_{\max}\right) = 1 - o(1),
\end{gathered}
\end{equation}

which can be easily extended to the uniform version since $C_1,\alpha_{\max}$ and constants hiding in $o(1)$ do not depend on $\lambda$, and estimators for all $\lambda$ share the same source of randomness.

Finally, \cite[Lemma A.5]{celentano2020lasso} shows there exists constants $C_2$ and $\alpha_{\max}$ such that $\forall \alpha\leq \alpha_{\max}$:
\begin{equation}\label{eqn:tau_xi_alpha}
    |\tau_{\rmL}-\tau_\alpha| \leq C_2 \sqrt{\alpha},\ \ |\xi_{\rmL}-\xi_\alpha| \leq C_2 \sqrt{\alpha},
\end{equation}

which can be extended to uniform version similarly.

Once uniform control results are extended, the rest of the proof stays the same as in \cite[Section B.5]{celentano2020lasso}, on combining \eqref{eqn:marginal_both_alpha}-\eqref{eqn:tau_xi_alpha}.
\end{proof}

As a consequence of 
Lemmas \ref{lemma:convergence_gl} and \ref{lemma:marginal_both}, we obtain the following corollary. An analogous result was established in \cite[Corollary J.2]{celentano2021cad} but unlike their case,  our guarantee here is uniform over the tuning parameter space. 
Thus the result below relies crucially on our earlier results that provide such uniform guarantees.

\begin{cor}\label{cor:marginal_beta_g_joint}
Recall  the definition of $\hat\bg_{\textrm{L}}$ from \eqref{def:hat_gl} and that $ \bg_L^f = \hat\bbeta_L^{f,d} -\bbeta$ from \eqref{def:hat_beta_f}. Under Assumption \ref{assumptions2}, for $1$-Lipschitz functions $\phi_\beta:(\R^p)^2\rightarrow \R$,
\begin{equation*}
    \supll |\phi_\beta(\hat\bbeta_{\textrm{L}},\hat\bg_{\textrm{L}}) - \E[\phi_\beta(\hat\bbeta_{\textrm{L}}^f,\bg_{\textrm{L}}^f)]| \convP 0.
\end{equation*}
\end{cor}

\begin{proof}
By triangle inequality,
\begin{align*}
    &\supll |\phi_\beta(\hat\bbeta_{\textrm{L}},\hat\bg_{\textrm{L}}) - \E[\phi_\beta(\hat\bbeta_{\textrm{L}}^f,\bg_{\textrm{L}}^f)]|\\
    \leq & \supll |\phi_\beta(\hat\bbeta_{\textrm{L}},\hat\bg_{\textrm{L}}) - \phi_\beta(\hat\bbeta_{\textrm{L}},\tilde\bbeta_{\textrm{L}}^d-\bbeta)| + \supll |\phi_\beta(\hat\bbeta_{\textrm{L}},\tilde\bbeta_{\textrm{L}}^d-\bbeta) - \E[\phi_\beta(\hat\bbeta_{\textrm{L}}^f,\bg_{\textrm{L}}^f)]|\\
    \leq & \supll\|\hat\bg_{\textrm{L}} - (\tilde\bbeta_{\textrm{L}}^d-\bbeta)\|_2 + \supll|\phi_\beta(\hat\bbeta_{\textrm{L}},\tilde\bbeta_{\textrm{L}}^d-\bbeta) - \E[\phi_\beta(\hat\bbeta_{\textrm{L}}^f,\hat\bbeta_L^{f,d} -\bbeta)]|\convP 0,
\end{align*}
since the first summand vanishes due to Lemma \ref{lemma:convergence_gl} and the second summand vanishes due to Lemma \ref{lemma:marginal_both}.
\end{proof}

To prove Lemma \ref{lemma:bridge_1}, that is, \eqref{eq:interim1}, we need to establish a convergence result that is uniform over $\lambda_L$ and $\lambda_R$. We divide this goal into a two-step strategy, where we first establish that \eqref{eq:interim1} holds with the supremum only over $\lambda_L$ (Lemma \ref{lemma:bridge_1_onlyL}) 
Then by appropriate Lipschitzness arguments (Lemma \ref{lemma:lipschitz_Psi_bridge1}), we extend the result to hold with supremum simultaneously over $\lambda_L$ and $\lambda_R$. 

\begin{lemma}\label{lemma:bridge_1_onlyL}
Denote the quantity $\E[\phi_\beta (\hat\bbeta_{\rmL}^{f,d},\hat\bbeta_{\textrm{R}}^{f,d},\bbeta)|\bg_{\textrm{L}}^f = \hat\bg_{\textrm{L}}]$ as $\phi_{\beta|L}(\hat\bg_{\textrm{L}})$. Then under Assumption \ref{assumptions2},
\begin{equation}\label{bridge_1_onlyL}  \supll|\phi_{\beta|L}(\hat\bg_{\textrm{L}}) - \E[\phi_\beta(\hat\bbeta_{\textrm{L}}^{f,d},\hat\bbeta_{\textrm{R}}^{f,d},\bbeta)]| \convP 0.
\end{equation}
\end{lemma}

\begin{proof}
From the definition, we know $\E[\phi_{\beta|L}(\bg_{\textrm{L}}^f)] = \E[\phi_\beta(\hat\bbeta_{\textrm{L}}^{f,d},\hat\bbeta_{\textrm{R}}^{f,d},\bbeta)]$, since, per the remark below \eqref{def:hat_gl}, $\tilde\bbeta_{\rmL}^d$ equals $\hat\bbeta_{\rmL}^{f,d}$ conditional on $\bg_{\textrm{L}}^f = \hat\bg_{\textrm{L}}$.

Therefore, we just need to show
$$\supll|\phi_{\beta|L}(\hat\bg_{\textrm{L}}) - \E[\phi_{\beta|L}(\bg_{\textrm{L}}^f)]| \convP 0,$$
which follows from Corollary \ref{cor:marginal_beta_g_joint}, if we can show that $\phi_{\beta|L}$ is a Lipschitz function.
The Lipschitzness follows by an argument similar to \cite[Section J.2]{celentano2021cad}, so we omit the details here. 
\end{proof}

Note that Corollary \ref{cor:marginal_beta_g_joint}, which in turn relies on Lemmas \ref{lemma:convergence_gl} and \ref{lemma:marginal_both}, forms a crucial ingredient for the preceding proof.

\begin{lemma}\label{lemma:lipschitz_Psi_bridge1}
Under Assumption \ref{assumptions2}, with probability $1-o(1)$, 
$\Psi(\lambda_{\textrm{R}})$ is an M-Lipschitz 
 function 
 of 
$\lambda_{\textrm{R}}$ 
for some positive constant $M$, where
\begin{equation*}
    \Psi(\lambda_{\textrm{R}}) := \supll |\phi_{\beta|L}(\hat\bg_{\textrm{L}}) - \E[\phi_\beta(\hat\bbeta_{\textrm{L}}^{f,d}(\lambda_{\rmL}),\hat\bbeta_{\textrm{R}}^{f,d}(\lambda_{\rmR}),\bbeta)]|.
\end{equation*}
\end{lemma}

\begin{proof}[Proof of Lemma \ref{lemma:lipschitz_Psi_bridge1}]

We define (with a slight overload of notations) an auxiliary function:
$$\psi(\lambda_{\textrm{L}},\lambda_{\textrm{R}}) := |\phi_{\beta|L}(\hat\bg_{\textrm{L}}) - \E[\phi_\beta(\hat\bbeta_{\textrm{L}}^{f,d}(\lambda_{\rmL}),\hat\bbeta_{\textrm{R}}^{f,d}(\lambda_{\rmR}),\bbeta)]|.$$

First, $\tau_{\rmL}$ and therefore $\hat\bbeta_{\textrm{L}}^{f,d}$ does not depend on $\lambda_{\textrm{R}}$. 
Also, $\hat\bbeta_{\textrm{R}}^{f,d} = \bbeta + \bg_{\textrm{R}}^f = \bbeta + \tau_{R}\tilde\bxi$ where $\tilde\bxi \sim \mcn(\bm{0},\bm{I}_p)$, $\tau_{R}$ is $C_1/\sqrt{n}$-Lipschitz in $\lambda_{\textrm{R}}$ from Lemma \ref{lemma:lipschitz_fixed_point}. Thus, we have
\begin{align*}
    &\left|\E[\phi_\beta(\hat\bbeta_{\textrm{L}}^{f,d}(\lambda_{\rmL}),\hat\bbeta_{\textrm{R}}^{f,d}(\lambda_1),\bbeta)] - \E[\phi_\beta(\hat\bbeta_{\textrm{L}}^{f,d}(\lambda_{\rmL}),\hat\bbeta_{\textrm{R}}^{f,d}(\lambda_2),\bbeta)]\right|\\
    \leq & \E [\|\hat\bbeta_{\textrm{R}}^{f,d}(\lambda_1) - \hat\bbeta_{\textrm{R}}^{f,d}(\lambda_2)\|_2]\\
    \leq & \sqrt{\E [\|(\tau_{\rmR}(\lambda_1) - \tau_{\rmR}(\lambda_2))\tilde\bxi\|_2^2]}\\
    \leq & C_1 |\lambda_1-\lambda_2|,
\end{align*}

which yields $\E[\phi_\beta(\hat\bbeta_{\textrm{L}}^{f,d}(\lambda_{\rmL}),\hat\bbeta_{\textrm{R}}^{f,d}(\lambda_{\rmR}),\bbeta)]|$ is $C$-Lipschitz in $\lambda_{\textrm{R}}$. 
Next note that conditioning on $\bg_{\textrm{L}}^f=\hat\bg_{\textrm{L}}$ necessitates  $\bg_{\textrm{R}}^f$ to take the following value so that their joint covariance structure matches \eqref{def:bS_g},
\begin{equation}\label{def:gr_f_conditional}
    \bg_{\textrm{R}}^f = \tau_{\textrm{R}}\rho/\tau_{\textrm{L}}\cdot \hat\bg_{\textrm{L}} + \tau_{\textrm{R}}\sqrt{1-\rho^2}\cdot \bxi, \,\, \text{where} \,\, \bxi\sim\mcn(0,\bI_p).
\end{equation}
This means that conditional on $\bg_{\textrm{L}}^f = \hat\bg_{\textrm{L}}$, $\hat\bbeta_{\textrm{R}}^{f,d} = \bbeta + \tau_{R}\rho/\tau_{L} \hat\bg_{\textrm{L}} + \tau_{R}\sqrt{1-\rho^2}\bxi$, whereas $\hat\bbeta_{\textrm{L}}^{d}$ never depends on $\lambda_{\textrm{R}}$.
Also, from \ref{lemma:existence_fixed_point} and Lemma \ref{lemma:lipschitz_fixed_point}, $\tau_{R}$ is bounded in $[\sqrt{\tau_{\min}^2/n},\sqrt{(\tau_{\max}^2+\sigma_{\max}^2)/n}]$ and is $C$-Lipschitz in $\lambda_{\textrm{R}}$ 
while both $\rho$ and $\sqrt{1-\rho^2}$ are bounded in $[0,1]$ and $C'$-Lipschitz in $\lambda_{\textrm{R}}$, and $1/\tau_{L} \leq \sqrt{n}/\tau_{\min}$. Combining Lemma \ref{lemma:convergence_gl} and \ref{lemma:beta_emp_diff} yields that $\|\hat\bg_{\textrm{L}}\|_2 \leq 2\|\hat\bbeta_{\textrm{L}}^d - \bbeta\|_2 $ with probability $1-o(1)$. In conjunction with  
Lemma \ref{lemma:marginal_both}, this yields that  $\|\hat\bg_{\textrm{L}}\|_2 \leq 
4 \E[\|\hat\bbeta_{\textrm{L}}^{f,d}-\bbeta\|_2] \leq 4 \sqrt {\E[\|\hat\bbeta_{\textrm{L}}^{f,d}-\bbeta\|_2^2]} = 4n\tau_{{\textrm{L}}}^2 \leq 4\sqrt{\tau_{\max}^2+\sigma_{\max}^2}$ with probability $1-o(1)$. Thus,
\begin{align*}
& \left| \E[\phi_\beta (\hat\bbeta_{\rmL}^{f,d},\hat\bbeta_{\textrm{R}}^{f,d}(\lambda_1),\bbeta)|\bg_{\textrm{L}}^f = \hat\bg_{\textrm{L}}] - \E[\phi_\beta (\hat\bbeta_{\rmL}^{f,d},\hat\bbeta_{\textrm{R}}^{f,d}(\lambda_2),\bbeta)|\bg_{\textrm{L}}^f = \hat\bg_{\textrm{L}}]\right|\\
\leq & \E\left[\left|\phi_\beta (\hat\bbeta_{\rmL}^{f,d},\hat\bbeta_{\textrm{R}}^{f,d}(\lambda_1),\bbeta)-\phi_\beta (\hat\bbeta_{\rmL}^{f,d},\hat\bbeta_{\textrm{R}}^{f,d}(\lambda_2),\bbeta)\right| |\bg_{\textrm{L}}^f = \hat\bg_{\textrm{L}}\right]\\
\leq & \E [\|\hat\bbeta_{\textrm{R}}^{f,d}(\lambda_1) - \hat\bbeta_{\textrm{R}}^{f,d}(\lambda_2)\|_2|\bg_{\textrm{L}}^f = \hat\bg_{\textrm{L}}]\\
    \leq & \sqrt{\E [\|\hat\bbeta_{\textrm{R}}^{f,d}(\lambda_1) - \hat\bbeta_{\textrm{R}}^{f,d}(\lambda_2)\|_2^2|\bg_{\textrm{L}}^f = \hat\bg_{\textrm{L}}]}\\
    \leq & \left(4(C'+\sqrt{\tau_{\max}^2+\sigma_{\max}^2}C)\sqrt{\tau_{\max}^2+\sigma_{\max}^2}/\tau_{\min} + (C'+\sqrt{\tau_{\max}^2+\sigma_{\max}^2}C)\right) |\lambda_1-\lambda_2|.
\end{align*}
Thus, we conclude that $\phi_{\beta|L}(\hat\bg_{\textrm{L}})$ is $(4(C'+\sqrt{\tau_{\max}^2+\sigma_{\max}^2}C)\sqrt{\tau_{\max}^2+\sigma_{\max}^2}/\tau_{\min} + (C'+\sqrt{\tau_{\max}^2+\sigma_{\max}^2}C))$-Lipschitz in $\lambda_{\textrm{R}}$ with probability $1-o(1)$. 

Therefore, with probability $1-o(1)$, $\psi(\lambda_{\textrm{L}},\lambda_{\textrm{R}})$ is an $M$-Lipschitz function of $\lambda_{\textrm{R}}$ for $M = 4(C'+\sqrt{\tau_{\max}^2+\sigma_{\max}^2}C)\sqrt{\tau_{\max}^2+\sigma_{\max}^2}/\tau_{\min} + (C'+\sqrt{\tau_{\max}^2+\sigma_{\max}^2}C) + C$. Notice that $M$ does not depend on $\lambda_{\textrm{L}}$. Therefore, by Lemma \ref{lemma:lipschitz_function_sup}, $\Psi(\lambda_{\textrm{R}}) := \sup_{\lambda_{\textrm{L}}}\psi(\lambda_{\textrm{L}},\lambda_{\textrm{R}})$ is also an $M$-Lipschitz function of $\lambda_{\textrm{R}}$, hence completing the proof.
\end{proof}

Now we are ready to prove Lemma \ref{lemma:bridge_1}.
\begin{proof}[Proof of Lemma \ref{lemma:bridge_1}]
Consider the high probability event in Lemma \ref{lemma:lipschitz_Psi_bridge1}. For any $\epsilon > 0$, define $\epsilon' = \epsilon/2M$. Let $k = \lceil(\lambda_{\max}-\lambda_{\min})/\epsilon'\rceil$. Define, for $i=0,...,k$: $\lambda_i = \lambda_{\min}+i\epsilon'$. Then by Lemma \ref{lemma:lipschitz_Psi_bridge1}, we know $\sup_{\lambda_{\textrm{R}} \in [\lambda_{i-1},\lambda_i]} \Psi(\lambda_{\textrm{R}}) \leq \Psi(\lambda_i) + M\epsilon'$.

By union bound, we have
\begin{equation*}
\begin{aligned}
    &\PP\left(\suplr \Psi(\lambda_{\textrm{R}}) \geq \epsilon\right)\\
    \leq &\sum_{i=1}^k \PP\left(\sup_{\lambda_{\textrm{R}} \in [\lambda_{i-1},\lambda_i]}\Psi(\lambda_{\textrm{R}}) \geq \epsilon\right)\\
    \leq&\sum_{i=1}^k \PP\left(\Psi(\lambda_i) \geq \epsilon-M\epsilon'\right)\\
    =&\sum_{i=1}^k \PP\left(\supll |\phi_{\beta|L}(\hat\bg_{\textrm{L}}(\lambda_{\rmL})) - \E[\phi_\beta(\hat\bbeta_{\textrm{L}}^{f,d}(\lambda_{\rmL}),\hat\bbeta_{\textrm{R}}^{f,d}(\lambda_i),\bbeta)]|\geq\epsilon/2\right)\\
    =&o(1),
\end{aligned}
\end{equation*}
on using Lemma \ref{lemma:bridge_1_onlyL}
\end{proof}

\subsubsection{Proof of Lemma \ref{lemma:bridge_2}}\label{subsec:proof:lemma:bridge_2}
Recall the definition of $\hat\bw_k$ from \eqref{uvw}. We define a loss function that yields $\hat\bw_\rmR$ as the minimizer:
\begin{equation*}
    \mcl{C}_{\lambda_\rmR}(\bw) := \frac{1}{2n}\|\bX\bw-\sigma\bz\|_2^2 + \frac{\lambda_\rmR}{2}\|\bw+\bbeta\|_2^2.
\end{equation*}
We next introduce some supporting lemmas.
\begin{lemma}\label{lemma:ridge_conditional_bw}
Recall the definition of $\hat\bg_{\textrm{L}}$ from \ref{def:hat_gl}. Under Assumption \ref{assumptions2}, for any $\epsilon > 0$, there exists a constant $C$ such that for a $1$-Lipschitz function $\phi_w:\R^p \rightarrow \R$,
$$\suplb \PP\left(\exists \bw \in \R^p, \left|\phi_w(\bw) - \E[\phi_w(\hat\bbeta_\rmR^f-\bbeta)|\bg_{\textrm{L}}^f = \hat\bg_{\textrm{L}}]\right| \geq \epsilon \ \text{and} \ \mcl{C}_{\lambda_{\textrm{R}}}(\bw) \leq \min \mcl{C}_{\lambda_{\textrm{R}}} + C\epsilon^2\right) = o(\epsilon^2).$$
\end{lemma}
\begin{proof}

From part of \cite[Lemma F.4]{celentano2021cad}, we know that 
$$\PP\left(\exists \bw \in \R^p, \left|\phi_w(\bw) - \E[\phi_w(\hat\bbeta_\rmR^f-\bbeta)|\bg_{\textrm{L}}^f = \hat\bg_{\textrm{L}}]\right| \geq \epsilon \ \text{and} \ \mcl{C}_{\lambda_{\textrm{R}}}(\bw) \leq \min \mcl{C}_{\lambda_{\textrm{R}}} + C\epsilon^2\right) = o(\epsilon^2).$$
The proof is then completed by the fact that as in \cite[Lemma F.4]{celentano2021cad}, $o(\epsilon^2)$ only hides constants that do not depend on $\lambda_{\rmL},\lambda_{\rmR}$.
\end{proof}

\begin{lemma}\label{lemma:ridge_loss_C_difference_lambda}
Under Assumption \ref{assumptions2}, there exists a positive constant $K$ such that 
\begin{equation*}
    \PP\left(\forall \lambda,\lambda' \in[\lambda_{\min},\lambda_{\max}], \mcl{C}_{\lambda'}(\hat{\bw}_{\rmR}(\lambda)) \leq \mcl{C}_{\lambda'}(\hat{\bw}_{\rmR}(\lambda')) + K|\lambda - \lambda'|\right) = 1-o(1).
\end{equation*}
\end{lemma}
\begin{proof}
We have
\begin{align*}
    &\mcl{C}_{\lambda'}(\hat{\bw}_{\rmR}(\lambda)) - \mcl{C}_{\lambda'}(\hat{\bw}_{\rmR}(\lambda')) \\
    =& \mcl{C}_{\lambda'}(\hat{\bw}_{\rmR}(\lambda)) - \mcl{C}_{\lambda}(\hat{\bw}_{\rmR}(\lambda)) + 
      \mcl{C}_{\lambda}(\hat{\bw}_{\rmR}(\lambda)) - \mcl{C}_{\lambda}(\hat{\bw}_{\rmR}(\lambda')) +  \mcl{C}_{\lambda}(\hat{\bw}_{\rmR}(\lambda')) - \mcl{C}_{\lambda'}(\hat{\bw}_{\rmR}(\lambda'))\\
      \leq & \mcl{C}_{\lambda'}(\hat{\bw}_{\rmR}(\lambda)) - \mcl{C}_{\lambda}(\hat{\bw}_{\rmR}(\lambda)) +  \mcl{C}_{\lambda}(\hat{\bw}_{\rmR}(\lambda')) - \mcl{C}_{\lambda'}(\hat{\bw}_{\rmR}(\lambda'))\\
      =& \frac{\lambda'-\lambda}{2}(\|\hat\bbeta_{\rmR}(\lambda)\|_2^2-\|\hat\bbeta_{\rmR}(\lambda')\|_2^2)\\
      \leq & \frac{|\lambda'-\lambda|}{2}(\|\hat\bbeta_{\rmR}(\lambda)\|_2^2+\|\hat\bbeta_{\rmR}(\lambda')\|_2^2).
\end{align*}
Further, with probability at least $1-e^{-n/2}$ we have $\|\bz\|_2 \leq 2 \sqrt{n}$, 
and therefore
$$\frac{\lambda}{2}\|\hat\bbeta_{\rmR}(\lambda)\|_2^2 \leq \mcl{C}_{\lambda}(\hat{\bw}_{\rmR}(\lambda)) = \argmin_{\bb} \mcl{C}_\lambda(\bb) \leq \mcl{C}_\lambda(\bm{0}) = \|\sigma\bz\|_2^2 + \frac{\lambda}{2}\|\bbeta\|_2^2 \leq 2\sigma^2 + \frac{\lambda}{2} \|\bbeta\|_2^2.$$
Hence $\|\hat\bbeta_{\rmR}(\lambda)\|_2^2$ is bounded with high probability and so is $\|\hat\bbeta_{\rmR}(\lambda')\|_2^2$, which completes the proof.
\end{proof}

\begin{lemma}\label{lemma:lipschitz_psi_conditional}
Under Assumption \ref{assumptions2}, for any $\epsilon > 0$, consider any $\lambda_1,\lambda_2 \in [\lambda_{\min},\lambda_{\max}]$ such that $|\lambda_1 - \lambda_2| \geq \epsilon$. Then with probability $1-o(1)$, we have
\begin{gather*}  |\psi(\lambda_1,\lambda_{\textrm{R}}) - \psi(\lambda_2,\lambda_{\textrm{R}})| \leq M|\lambda_1-\lambda_2|\\   |\psi(\lambda_{\textrm{L}},\lambda_1) - \psi(\lambda_{\textrm{L}},\lambda_2)| \leq M|\lambda_1-\lambda_2|
\end{gather*}
for some constant $M$ (that does not depend on $\epsilon$), where $\psi(\lambda_{\textrm{L}},\lambda_{\textrm{R}}) = \E[\phi_w(\hat\bbeta_{\textrm{R}}^f-\bbeta) |\bg_{\textrm{L}}^f = \hat\bg_{\textrm{L}}]$.
\end{lemma}
We defer the proof to Section \ref{subsec:proof:lemma:lipschitz_psi_conditional} 
We remark that this is a slightly weaker condition than the Lipschitzness of $\psi$ in $\lambda_k$, but it suffices for our purpose. For convenience, we call this "weak-Lipschitz" condition.

\begin{lemma}\label{lemma:debiased_lipschitz_hat}
Under Assumption \ref{assumptions2}, $\tilde\bbeta_\rmR^d$ is an $M$-Lipschitz function of $\hat\bbeta_\rmR$ for some constant $M$.
\end{lemma}

\begin{proof}
Recall from \eqref{def:tilde_beta_d} and \eqref{def:fixed_point} that $\tilde\bbeta_\rmR^d = \hat\bbeta_\rmR + \frac{\bX^\top(\by-\bX\hat\bbeta_\rmR)}{n\zeta_\rmR}$. Further, the KKT condition for ridge regression implies that 
$\frac{1}{n}\bX^\top(\by-\bX\hat\bbeta_\rmR) = \lambda_R\hat\bbeta_\rmR$. Since from \ref{lemma:existence_fixed_point} we know that $\zeta_\rmR$ is bounded below by a positive constant $\zeta_{\min}$, it follows that $\tilde\bbeta_\rmR^d$ is an $M=1+\lambda_{\max}/\zeta_{\min}$-Lipschitz function of $\hat\bbeta_\rmR$.
\end{proof}

\begin{proof}[Proof of Lemma \ref{lemma:bridge_2}]
Let $C>0$ as given by Lemma \ref{lemma:ridge_conditional_bw}, let $K>0$ as given by Lemma \ref{lemma:ridge_loss_C_difference_lambda}, and let $M>0$ as given by Lemma \ref{lemma:lipschitz_psi_conditional}. Consider any $\epsilon > 0$, define $\epsilon' = \min\left(\frac{C\epsilon^2}{K},\frac{\epsilon}{M}\right)$. Let $k = \lceil(\lambda_{\max}-\lambda_{\min})/\epsilon'\rceil$. Further define $\lambda_i = \lambda_{\min} + i\epsilon'$ for $i=0,...,k$. By Lemma \ref{lemma:ridge_conditional_bw}, the event
\begin{equation}\label{event:ridge_W2_conditional}
    \left\{\forall i_{\textrm{L}},i_{\textrm{R}}=1,...,k,\forall \bw \in \R^p,\mcl{C}_{\lambda_{i_{\textrm{R}}}}(\bw) \leq \min \mcl{C}_{\lambda_{i_{\textrm{R}}}} + C\epsilon^2 \Rightarrow \left|\phi_w(\bw) - \psi(\lambda_{i_{\textrm{L}}},\lambda_{i_{\textrm{R}}})\right|\leq\epsilon \right\}
\end{equation}
has probability $1-k^2o(\epsilon^2) = 1-o(1)$.
Now define $\lambda_{i_{\textrm{R}}}:=\argmax\{|\lambda_{\textrm{R}} - \lambda_i|,|\lambda_{\textrm{R}}-\lambda_{i+1}|\}$ with $\lambda_i \leq \lambda_{\textrm{R}} \leq \lambda_{i+1}$ (so we know $|\lambda_{i_{\textrm{R}}}-\lambda_{\textrm{R}}|\geq \epsilon'/2$) and similarly for $i_{\textrm{L}}$. 
On the intersection of event \ref{event:ridge_W2_conditional} and the event in Lemma \ref{lemma:ridge_loss_C_difference_lambda}, which has probability $1-o(1)$ 
we have that
$$\mcl{C}_{\lambda_{i_{\textrm{R}}}}(\hat\bw_{\rmR}(\lambda_{\textrm{R}})) \leq \min\mcl{C}_{\lambda_{i_{\textrm{R}}}} + K\epsilon' \leq \min\mcl{C}_{\lambda_{i_{\textrm{R}}}} + C\epsilon^2.$$

This implies (since we are on event \ref{event:ridge_W2_conditional}) that $|\phi_w(\hat\bw_{\rmR}(\lambda_{\textrm{R}})) - \psi(\lambda_{i_L},\lambda_{i_R})|\leq \epsilon$, where  $1\leq i_{\textrm{R}} \leq k$.
Thus we have
\begin{equation}
\begin{aligned}\label{eq:intermcalc}    &|\phi_w(\hat\bw_{\rmR}(\lambda_{\textrm{R}})) - \psi(\lambda_{\textrm{L}},\lambda_{\textrm{R}})|\\
    \leq & |\phi_w(\hat\bw_{\rmR}(\lambda_{\textrm{R}})) - \psi(\lambda_{i_L},\lambda_{i_R})| + |\psi(\lambda_{i_L},\lambda_{i_R}) - \psi(\lambda_{L},\lambda_{i_R})| + 
    |\psi(\lambda_{L},\lambda_{i_R}) - \psi(\lambda_{i_L},\lambda_{R})|\nonumber\\
    \leq & \epsilon + 2 M \epsilon'\nonumber\\
    \leq & 3\epsilon\nonumber\\
\end{aligned}
\end{equation}
with probability $1-o(1)$, where the second-to-last inequality follows from Lemma \ref{lemma:lipschitz_psi_conditional} and the fact that $|\lambda_{i_{\textrm{L}}}-\lambda_{\textrm{L}}|,|\lambda_{i_{\textrm{R}}}-\lambda_{\textrm{R}}|\leq \epsilon'$.
Finally we note that $\phi_\beta(\tilde\bbeta_{\rmL}^d,\tilde\bbeta_{\rmR}^d,\bbeta)$ is a Lipschitz function of $\tilde\bbeta_{\textrm{R}}^d$, in fact, $\phi_\beta(\tilde\bbeta_{\rmL}^d,\tilde\bbeta_{\rmR}^d,\bbeta)=\phi_\beta(\tilde\bbeta_{\rmL}^d,\hat\bbeta_{\rmR}(1+\lambda_R/\zeta_R),\bbeta)$, (by definition). This in turn equals $\phi_\beta(\tilde\bbeta_{\rmL}^d,(\hat{\bm{w}}_R + \bbeta)(1+\lambda_R/\zeta_R),\bbeta)$. If we define this to be $\phi_w(\hat{\bm{w}}_R)$, then $\psi(\lambda_L,\lambda_R)=\E [\phi_\beta(\tilde\bbeta_{\rmL}^d,\hat\bbeta_{\rmR}^f(1+\lambda_R/\zeta_R),\bbeta|\bm{g}^f_L=\hat{\bm{g}}_L)]=\E [\phi_\beta(\tilde\bbeta_{\rmL}^d,\hat\bbeta_{\rmR}^{f,d},\bbeta)|\bm{g}^f_L=\hat{\bm{g}}_L]$, once again by definition. Then \eqref{eq:intermcalc} yields the desired result.
\end{proof}

\section{Proof of supporting lemmas for Section \ref{sec:proof:lemma:bridge_1}}

\subsection{Proof of Lemma \ref{lemma:uvw_convergence}}\label{subsec:proof:lemma:uvw_convergence}

We introduce two supporting Lemmas:

\begin{lemma}\label{lemma:marginal_single}
    Under Assumption \ref{assumptions2}, for $k=L,R$ and any $1$-Lipschitz function $\phi_\beta:\R^p\rightarrow \R$,
\begin{equation*}
\suplk |\phi_\beta(\hat\bbeta_k) - \E[\phi_\beta(\hat\bbeta_k^{f})]| \convP 0.\\
\end{equation*}
\end{lemma}

\begin{lemma}\label{lemma:marginal_single_e}
Recall $\hat\bu_k = \by - \bX\hat\bbeta_k$. Further define $\hat\bu_k^f = \sqrt{n}\zeta_k\bh_k^f$, where $(\bh_k^f,\bh_{\textrm{R}}^f) \sim \mcn(\bm{0},\bS \otimes \bI_n)$ for $\bS$ defined in Eqn. \ref{def:fixed_point}. Under Assumption \ref{assumptions2}, for any $1$-Lipschitz functions $\phi_u:\R^n \rightarrow \R$,
$$\suplk \left|\phi_u\left(\frac{\hat\bu_k}{\sqrt{n}}\right) - \E\left[\phi_u\left(\frac{\hat\bu_k^f}{\sqrt{n}}\right)\right]\right|\convP 0.$$
\end{lemma}
Lemmas \ref{lemma:marginal_single} and  \ref{lemma:marginal_single_e} are proved in Section \ref{subsec:proof:lemma:marginal_single}. 

\begin{proof}[Proof of Lemma \ref{lemma:uvw_convergence}]
First consider $\hat\bw_k = \hat\bbeta_k - \bbeta$. From Lemma \ref{lemma:marginal_single}, we know that for any $1$-Lipschitz function $\phi$,
$$\suplk |\phi(\hat\bw_k)-\E[\phi(\hat\bbeta_k^f-\bbeta)]| \convP 0,$$
where recall from Eqn. \ref{def:hat_beta_f} that $\hat\bbeta_k^f-\bbeta = \eta_k(\bbeta+\bg_k^f,\zeta_k)-\bbeta$ is a $1/\sqrt{p}$-Lipschitz function of $\sqrt{p}\bg_k^f$ due to the Lipschitzness of the proximal mapping operator, where $\sqrt{p}\bg_k^f\sim\mcn(\bm{0},p\tau_k^2\bI_p)$ with $p\tau_k^2 \leq \delta\tau_{\max}^2 + \delta\sigma_{\max}^2$ and $\delta=p/n$. 
Also, $\E[\|\bg_k^f\|_2^2] = p\tau_k^2$. Moreover, from \eqref{def:hat_beta_f} and \eqref{def:fixed_point}, $\E[\|(\hat\bbeta_k^f-\bbeta)\|_2^2] = n\tau_k^2-\sigma^2$ 
$ \leq \tau_{\max}^2$ is bounded by \ref{lemma:existence_fixed_point}. Therefore, by Lemma \ref{lemma:concentration_second_moment}, we know
$$\supll \left| \|\hat\bw_k\|_2^2 - (n\tau_k^2-\sigma^2)\right| \convP 0.$$
We also know that $n\tau_k^2-\sigma^2 \geq \tau_{\min}^2$ by \ref{lemma:existence_fixed_point}. Thus, by Lemma \ref{sqrt_convergence}, 
$$\supll \left| \|\hat\bw_k\|_2 - \sqrt{n\tau_k^2-\sigma^2}\right| \convP 0.$$

As a direct corollary, with probability $1-o(1)$, $\|\hat\bw_k\|_2 \leq  2\sqrt{n\tau_k^2-\sigma^2} \leq 2\tau_{\max}$ and $\|\hat\bw_k\|_2 \geq  \sqrt{n\tau_k^2-\sigma^2}/2 \geq \tau_{\min}/2$ for all $\lambda_k$, so $\|\hat\bw_k\|_2$ is bounded both above and below for all $\lambda_k$ with probability $1-o(1)$.
Similarly, convergence of $\|\hat\bu_k\|_2/\sqrt{n}$ follows by starting from Lemma \ref{lemma:marginal_single_e} on combining with \eqref{lemma:existence_fixed_point} and Lemmas \eqref{lemma:concentration_second_moment}, \eqref{sqrt_convergence}.

Finally we consider $\hat\bv_k$. We know $\hat\bv_k = \bX^\top \hat\bu_k$ and that 
$\|\bX\|_\op/\sqrt{n}$ is bounded with probability $1-o(1)$ (Corollary \ref{largest_singular_value}). Thus,
$$\suplk\|\hat\bv_k\|_2/n \leq \|\bX\|_\op/\sqrt{n} \cdot \suplk \|\hat\bu_k\|_2/\sqrt{n},$$ which is bounded above with probability $1-o(1)$. This completes the proof.
\end{proof}

\subsection{Proof of Lemma \ref{lemma:lipschitz_psi_conditional}}\label{subsec:proof:lemma:lipschitz_psi_conditional}

We introduce another Lemma:

\begin{lemma}\label{lemma:wf_weak_lipschitz}
Under Assumption \ref{assumptions2}, for any $\epsilon > 0$, consider any $\lambda_1,\lambda_2 \in [\lambda_{\min},\lambda_{\max}]$ such that $|\lambda_1 - \lambda_2| \geq \epsilon$, then with probability $1-o(1)$, we have
\begin{equation*}
\begin{gathered}
\|\hat\bbeta_{\textrm{L}}(\lambda)\|_2 \leq M, \,\, \forall \lambda \in [\lambda_{\min},\lambda_{\max}],\\
\|\hat\bbeta_{L}(\lambda_1) - \hat\bbeta_{L}(\lambda_2)\|_2 \leq M |\lambda_1-\lambda_2|.
\end{gathered}
\end{equation*}
for some positive constant $M$ that does not depend on $\epsilon$. \end{lemma}

\begin{proof}
The first line follows directly from Lemma \ref{lemma:uvw_convergence} and the fact that $\hat\bbeta_{\textrm{L}} = \hat\bw_{\textrm{L}} + \bbeta$.
For the second line, consider any $\epsilon > 0$. It is a direct consequence of Lemma \ref{lemma:marginal_single} that 
\begin{equation}\label{eq:interimconv}
\PP\left(\supll \left|\|\hat\bbeta_\rmL(\lambda_{\textrm{L}})\|_2 - \E\|\hat\bbeta_\rmL^f(\lambda_{\textrm{L}})\|_2\right|\geq \epsilon\right) = o(1).
\end{equation}
For any $\lambda_1,\lambda_2$, by triangle inequality, we have    $$\|\hat\bbeta_{\rmL}(\lambda_1) - \hat\bbeta_{\rmL}(\lambda_2)\|_2\leq \|\hat\bbeta_{\rmL}(\lambda_1) - \hat\bbeta_{\rmL}^f(\lambda_1)\|_2 + \|\hat\bbeta_{\rmL}(\lambda_2) - \hat\bbeta_{\rmL}^f(\lambda_2)\|_2 + \|\hat\bbeta_{\rmL}^f(\lambda_1) - \hat\bbeta_{\rmL}^f(\lambda_2)\|_2.$$
Now, 
$\|\hat\bbeta_{\rmL}(\lambda_1) - \hat\bbeta_{\rmL}^f(\lambda_1)\|_2 \leq \epsilon$ and $\|\hat\bbeta_{\rmL}(\lambda_2) - \hat\bbeta_{\rmL}^f(\lambda_2)\|_2 \leq \epsilon$ with probability $1-o(1)$ by \eqref{eq:interimconv}.
Further, recalling the definition of $\hat\bbeta_{\textrm{L}}^f$ in \eqref{def:hat_beta_f} and notice that the minimization problem is separable, we know the $i$-th entry of $\hat\bbeta_{\textrm{L}}^f$ satisfies
$$\hat\bbeta_{\textrm{L}}^{f,(i)} = \eta(\beta_i+\tau_{L} Z_i,\frac{\lambda_{\textrm{L}}}{\sqrt{n}\zeta_{\textrm{L}}}),$$
where $Z_i \iid \mcn(0,1)$ and
$$\eta(x,b) = \begin{cases}x+b,\ x<-b\\ 0,\ -b\leq x\leq b\\x-b,\ x>b\end{cases},$$
the soft-thresholding operator,
is $1$-Lipshitz in both $x$ and $b$. Thus,
\begin{align*}
    &\|\hat\bbeta_{L}^f(\lambda_1) -\hat\bbeta_{L}^f(\lambda_2)\|_2^2\\
=& \sum_{i=1}^p \left( \eta(\beta_i+\tau_{L}(\lambda_1) Z_i,\frac{\lambda_1}{\sqrt{n}\zeta_{L}(\lambda_1)}) - \eta(\beta_i+\tau_{L}(\lambda_2) Z_i,\frac{\lambda_2}{\sqrt{n}\zeta_{L}(\lambda_2)})\right)^2\\
\leq & \sum_{i=1}^p 2\left( (\tau_{L}(\lambda_1)Z_i-\tau_{L}(\lambda_2)Z_i)^2 + (\frac{\lambda_1}{\sqrt{n}\zeta_{L}(\lambda_1)} - \frac{\lambda_2}{\sqrt{n}\zeta_{L}(\lambda_2)})^2\right)\\
=& 2\left((\tau_{L}(\lambda_1)-\tau_{L}(\lambda_2))^2\sum_{i=1}^p Z_i^2 + \delta(\frac{\lambda_1}{\zeta_{L}(\lambda_1)} - \frac{\lambda_2}{\zeta_{L}(\lambda_2)})^2\right)\\
\leq & M^2|\lambda_1-\lambda_2|^2
\end{align*}
for some constant $M$ with probability $1-o(1)$, where we used \ref{lemma:existence_fixed_point}, Lemma \ref{holder_function_operation}, and  the facts that $\sqrt{n}\tau_{L}(\lambda),\zeta_{L}(\lambda)$ are bounded Lipschitz functions of $\lambda_L$  (from Lemma \ref{lemma:lipschitz_fixed_point}) and $\sum_{i=1}^p Z_i^2/p$ is bounded with probability $1-o(1)$.
Combining the above, we know with probability $1-o(1)$,
$$ \|\hat\bbeta_{L}(\lambda_1) - \hat\bbeta_{L}(\lambda_2)\|_2\leq 2\epsilon + M|\lambda_1-\lambda_2| \leq (M+2) |\lambda_1-\lambda_2|,$$
which concludes the proof.
\end{proof}

As a corollary, we have the following lemma:

\begin{lemma}\label{lemma:gl_weak_lipschitz}
Under Assumption \ref{assumptions2}, for any $\epsilon > 0$, consider any $\lambda_1,\lambda_2 \in [\lambda_{\min},\lambda_{\max}]$ such that $|\lambda_1 - \lambda_2| \geq \epsilon$, then with probability $1-o(1)$, we have
$$\|\hat\bg_{L}(\lambda_1) - \hat\bg_{L}(\lambda_2)\|_2 \leq M|\lambda_1-\lambda_2|,$$
for some constant $M$ (that does not depend on $\epsilon$).
\end{lemma}

\begin{proof}
Recall Eqn. \ref{def:hat_gl}, and note the following:
\begin{itemize}
    \item $\zeta_{\textrm{L}},\sqrt{n}\tau_{L}$ are bounded Lipschitz functions of $\lambda_{\textrm{L}}$, and as a simple corollary, $\sqrt{n\tau_{\textrm{L}}^2-\sigma^2}$ is also a bounded Lipschitz function of $\lambda_{\textrm{L}}$.
    \item $\left|\|\hat\bw_{L}(\lambda_1)\|_2 - \|\hat\bw_{L}(\lambda_2)\|_2\right| \leq \|\hat\bw_{L}(\lambda_1)-\hat\bw_{L}(\lambda_2)\|_2 =
    \|\hat\bbeta_{L}(\lambda_1)-\hat\bbeta_{L}(\lambda_2)\|_2 \leq M|\lambda_1-\lambda_2|$ with probability $1-o(1)$ by Lemma \ref{lemma:wf_weak_lipschitz}. Also $\|\hat\bw\|_2$ is bounded with probability $o(1)$ by Lemma \ref{lemma:uvw_convergence}
    \item $\hat\bu_{\textrm{L}} = \sigma\bz-\bX\hat\bw_{\textrm{L}}$, so $\frac{1}{\sqrt{n}}\|\hat\bu_{L}(\lambda_1) -\hat\bu_{L}(\lambda_2)\|_2 \leq \frac{1}{\sqrt{n}}\sigma_{\max}(\bX)\|\hat\bw_{L}(\lambda_1)-\hat\bw_{L}(\lambda_2)\|_2 \leq (2+\sqrt{\delta})M|\lambda_1-\lambda_2|$ with probability $1-o(1)$, where we used Corollary \ref{largest_singular_value}. Also $\frac{1}{\sqrt{n}}\|\hat\bu_{\textrm{L}}\|_2 \leq \frac{1}{\sqrt{n}}(\sigma \|\bz\|_2 + \|\bX\hat\bw_{\textrm{L}}\|_2)$ is bounded with probability $1-o(1)$ for the same reason.
    \item By the same argument, $\frac{1}{n}\|\hat\bv_{L}(\lambda_1)-\hat\bv_{L}(\lambda_2)\|_2 \leq (2+\sqrt{\delta})^2M|\lambda_1-\lambda_2|$ with probability $1-o(1)$ and $\frac{1}{n}\|\hat\bv\|_2$ is bounded with probability $1-o(1)$.
\end{itemize}

Then the proof is complete on iteratively applying Lemma \ref{holder_function_operation} on the above displays.
\end{proof}

We are now ready to prove Lemma \ref{lemma:lipschitz_psi_conditional}.
\begin{proof}[Proof of Lemma \ref{lemma:lipschitz_psi_conditional}]
By Jensen's inequality and the fact that $\phi_w$ is 1-Lipschitz, we only need to show that conditional on $\bg_{\textrm{L}}^f = \hat\bg_{\textrm{L}}$, $\hat\bw_{\textrm{R}}^f=\hat\bbeta_{\textrm{R}}^f-\bbeta$
satisfies the ``weak-Lipschitz" condition in $\lambda_{\textrm{L}}$ with probability $1-o(1)$, where by Eqn. \ref{eqn:def:w_lambda},
\begin{align*}    
\hat\bbeta_{\textrm{R}}^f-\bbeta &= \frac{1}{\zeta_{\textrm{R}}+\lambda_{\textrm{R}}}(\zeta_{\textrm{R}}\bg_{\textrm{R}}^f - \lambda_R \bbeta)\\
    &= \frac{1}{\zeta_{\textrm{R}}+\lambda_{\textrm{R}}}(\zeta_{\textrm{R}}\tau_{R}\rho/\tau_{\textrm{L}}\cdot\hat{\bm{g}}_{\textrm{L}} + \zeta_{\textrm{R}}\tau_{R}\sqrt{1-\rho^2}\cdot\bxi - \lambda_{\textrm{R}} \bbeta).
\end{align*}
Now we note the following observations.

\begin{itemize}
    \item $\zeta_{\textrm{R}},\lambda_{\textrm{R}},\bbeta$ does not depend on $\lambda_{\textrm{L}}$
     and by Lemma \ref{lemma:lipschitz_fixed_point}, $\zeta_{\textrm{R}},\lambda_{\textrm{R}}$ are both bounded Lipschitz functions of $\lambda_{\textrm{R}}$. Further, $\|\bbeta\|_2$ is bounded. Thus, by Lemma \ref{holder_function_operation}, $\frac{\lambda_{\textrm{R}}}{\zeta_{\textrm{R}}+\lambda_{\textrm{R}}}\bbeta$ is Lipschitz in both $\lambda_{\textrm{L}},\lambda_{\textrm{R}}$ with some constant $M_1$.
    \item In addition, by Lemma \ref{lemma:lipschitz_fixed_point}, $\sqrt{n}\tau_{R}$ is bounded Lipschitz functions of $\lambda_{\textrm{R}}$, and $\sqrt{1-\rho^2}$ is bounded and Lipschitz in both $\lambda_{\textrm{L}}$ and $\lambda_{\textrm{R}}$. Further, $\bxi/\sqrt{n} \sim \mcn(\bm{0},\bm{I}_p/n)$  is bounded with probability $1-o(1)$. Therefore, by Lemma \ref{holder_function_operation}, $\frac{\zeta_{\textrm{R}}\tau_{R}\sqrt{1-\rho^2}}{\zeta_{\textrm{R}}+\lambda_{\textrm{R}}}\bxi$ is Lipschitz in both $\lambda_{\textrm{L}},\lambda_{\textrm{R}}$ with some constant $M_2$ with probability $o(1)$.
    
    \item By Lemma \ref{lemma:lipschitz_fixed_point}, $\rho$ is bounded and Lipschitz in both $\lambda_{\textrm{L}}$ and $\lambda_{\textrm{R}}$, and $\sqrt{n}\tau_{k}$ is bounded and Lipschitz in $\lambda_k$.
    
  \item  By Lemma \ref{lemma:convergence_gl} and Lemma \ref{lemma:marginal_both}, $\|\hat\bg_{\textrm{L}}\|_2 \leq 2\|\hat\bbeta_{\textrm{L}}^d - \bbeta\|_2 \leq 4 \E[\|\hat\bbeta_{\textrm{L}}^{f,d}-\bbeta\|_2] \leq 4 \sqrt {\E[\|\hat\bbeta_{\textrm{L}}^{f,d}-\bbeta\|_2^2]} = 4n\tau_{\textrm{L}}^2 \leq 4(\tau_{\max}^2+\sigma_{\max}^2)$ 
    with probability $1-o(1)$, where the equality follows from \eqref{def:hat_beta_f} and \eqref{def:fixed_point}.
    
  \item   Further, $\hat\bg_{\textrm{L}}$ does not depend on $\lambda_{\textrm{R}}$, and by Lemma \ref{lemma:gl_weak_lipschitz}, $\hat\bg_{\textrm{L}}$ satisfy the "weak-Lipschitz" condition w.r.t. $\lambda_{\textrm{L}}$ with probability $1-o(1)$.
    
  \item  Hence, by Lemma \ref{holder_function_operation}, $ \frac{\zeta_{\textrm{R}}\tau_{R}\rho}{\tau_{L}(\zeta_{\textrm{R}}+\lambda_{\textrm{R}})}\hat\bg_{\textrm{L}}$ satisfies the aforementioned condition w.r.t. both $\lambda_{\textrm{L}},\lambda_{\textrm{R}}$ with some constant $M_3$ with probability $1-o(1)$.
\end{itemize}
Combining the above steps completes the proof.
\end{proof}

\subsection{Proof of Lemma \ref{lemma:marginal_single}}
\label{subsec:proof:lemma:marginal_single}
For the case of the Lasso, \cite[Theorem 3.1]{miolane2021distribution} proved an analogous result for $W_2$ convergence. Although this does not directly yield our current lemma, the ideas therein sometimes prove to be useful. Below we present the proof of Lemma \ref{lemma:marginal_single} for the case of the ridge. Our approach towards handling $\ell_1$ convergence can be adapted to extend \cite[Theorem 3.1]{miolane2021distribution} for the lasso to an $\ell_1$ convergence result as well. For simplicity, we drop the subscript $R$ for this section. For convenience, we also make the $\lambda$ dependence explicit for certain expressions in this section (via writing $\hat\bbeta(\lambda)$, for instance).

\subsubsection{Converting the optimization problem}\label{sec:optconvert}

First, although the original optimization problem is
$$\hat\bbeta(\lambda) = \argmin_{\bb}\Big\{\frac{1}{2n}\|\by-\bX\bb\|_2^2 + \frac{\lambda}{2}\|\bb\|_2^2\Big\}:=\argmin_{\bb} \mcl{L}_\lambda(\bb),$$
it is more convenient to work with $\hat\bw(\lambda) = \hat\bbeta(\lambda) - \bbeta$, which satisfies
\begin{equation}\label{wu_optimization}
\begin{aligned}
    \hat\bw(\lambda) &= \argmin_{\bw} \frac{1}{2n} \|\bX \bw - \sigma\bz\|_2^2 + \frac{\lambda}{2} \|\bw + \bbeta\|_2^2:=\argmin_{\bw}\mathcal{C}_\lambda(\bw)\\
    &= \argmin_{\bw} \max_{\bu} \frac{1}{n}\bu^\top\bX\bw - \frac{\sigma}{n}\bu^\top \bz - \frac{1}{2n} \|\bu\|_2^2 + \frac{\lambda}{2} \|\bw + \bbeta\|_2^2 :=\argmin_{\bw} \max_{\bu}c_\lambda(\bw,\bu),
\end{aligned}
\end{equation}
where we used the fact that $\|\bx\|_2^2 = \max_{\bu} 2\bu^\top\bx - \|\bu\|_2^2$.
We call this the Primary Optimization (PO) problem, which involves a random matrix $\bX$. We then define the following Auxiliary Optimization (AO) that involves only independent random vectors $\bxi_g \sim \mcn(0,\bm{I}_p),\bxi_h \sim \mcn(0,\bm{I}_n)$. We call its solution $\bw^*(\lambda)$.
\begin{equation}\label{w_optimization}
\begin{aligned}
    \bw^*(\lambda) &= \argmin_{\bw} \max_{\bu} \frac{1}{n}\|\bu\|_2 \bxi_g^\top \bw + \sqrt{\|\bw\|_2^2+\sigma^2}\frac{\bxi_h^\top \bu}{n} - \frac{1}{2n} \|\bu\|_2^2 + \frac{\lambda}{2} \|\bw + \bbeta\|_2^2:=\argmin_{\bw} \max_{\bu} l_\lambda(\bw,\bu)\\
    &=\argmin_{\bw}\max_{\alpha\geq 0}\frac{\alpha}{\sqrt{n}}(\bxi_g^\top \bw + \sqrt{\|\bw\|_2^2+\sigma^2} \|\bxi_h\|_2) - \frac{\alpha^2}{2} + \frac{\lambda}{2} \|\bw + \bbeta\|_2^2:=\argmin_{\bw}\max_{\alpha\geq 0}\ell_\lambda(\bw,\alpha) \\
    &= \argmin_{\bw} \frac{1}{2n} (\bxi_g^\top \bw + \sqrt{\|\bw\|_2^2+\sigma^2}\|\bxi_h\|_2)_+^2 + \frac{\lambda}{2} \|\bw + \bbeta\|_2^2:= \argmin_{\bw} L_\lambda(\bw).
\end{aligned}
\end{equation}
Note to obtain  the third equality we set $\alpha := \|\bu\|_2/\sqrt{n}$.
Section \ref{subsubsec:connect_PO_AO} will establish a formal connection between PO and AO.

Next we perform some non-rigorous calculations using the AO to gain insights and show that it should asymtotically behave as a much simpler scalar optimization problem. To this end, we use the fact that $\sqrt{\|\bw\|_2^2 + \sigma^2} = \argmin_{\tau\geq 0} \frac{\|\bw\|_2^2+\sigma^2}{2\tau} + \frac{\tau}{2},$
and obtain
\begin{align*}
    &\min_{\bw}\max_{\alpha\geq 0}\ell_\lambda(\bw,\alpha)\\
    =&\min_{\bw}\max_{\alpha\geq 0}\min_{\tau\geq 0}\frac{\alpha}{\sqrt{n}}\left(\bxi_g^\top \bw + \left(\frac{\|\bw\|_2^2+\sigma^2}{2\tau}+\frac{\tau}{2}\right)\|\bxi_h\|_2\right) - \frac{\alpha^2}{2} + \frac{\lambda}{2} \|\bw + \bbeta\|_2^2.
\end{align*}
We know in the asymptotic limit $\|\bxi_h\|_2/\sqrt{n} \rightarrow 1$. Substituting in and optimizing $\bw$ first, we have
$$
\argmin_{\bw}\frac{\alpha}{\sqrt{n}}\bxi_g^\top \bw + \frac{\alpha(\|\bw\|_2^2+\sigma^2)}{2\tau}+\frac{\alpha\tau}{2} - \frac{\alpha^2}{2} + \frac{\lambda}{2} \|\bw + \bbeta\|_2^2  = -\frac{1}{\alpha/\tau + \lambda}\left( \frac{\alpha}{\sqrt{n}} \bxi_g + \lambda \bbeta\right).
$$
Plugging in and using the fact that 
$\|\bxi_g\|_2^2/p \rightarrow 1$, $\|\bbeta\|_2^2 \rightarrow \sigma_\beta^2$, and $\bxi_g^\top\bbeta/\sqrt{n} \rightarrow 0$, we arrive at a Scalar Optimization (SO) problem in the asymptotic limit:
\begin{equation}\label{eqn:def:SO}
\begin{aligned}    (\alpha_*,\tau_*) &= \argmax_{\alpha \geq 0} \min_{\tau \geq 0} \psi(\alpha,\tau)\\
     &:= \argmax_{\alpha \geq 0} \min_{\tau \geq 0} \frac{\alpha\sigma^2}{2\tau}+\frac{\alpha\tau}{2}-\frac{\alpha^2}{2}-\frac{\alpha^2\tau\delta}{2(\alpha+\tau\lambda)} + \frac{\alpha\lambda \sigma_\beta^2}{2(\alpha+\tau\lambda)}.
\end{aligned}
\end{equation}
By some algebra, it turns out that the solution to  \ref{eqn:def:SO} and the Ridge part of \eqref{def:fixed_point} can be related as $\tau_{\star} = \sqrt{n}\tau_{R}$ in and $\alpha_{\star} = \tau_{\star}\zeta_{\textrm{R}}$. 
So by the discussion before \ref{lemma:existence_fixed_point}, we know $\alpha_{\star},\tau_{\star}$  is unique and bounded. 
We in turn define
\begin{equation}\label{eqn:def:w_lambda}
\bw(\lambda)  = -\frac{1}{\alpha_*/\tau_* + \lambda}\left( \frac{\alpha_*}{\sqrt{n}} \bxi_g + \lambda \bbeta\right).
\end{equation}
By some algebra, this satisfies 
\begin{equation}\label{eq:wandbeta}
\bw(\lambda) = \hat\bbeta^f(\lambda) - \bbeta.
\end{equation}
We will rigorously prove the above conversion of AO to SO in Section \ref{subsubsec:connect_AO_SO}

\subsubsection{Connecting AO with SO}\label{subsubsec:connect_AO_SO}
First we show that AO has a minimizer.

\begin{prop}
$L_\lambda$ admits almost surely a unique minimizer $\bw^*(\lambda)$ on $\R^p$.
\end{prop}

\begin{proof}
First, $L_\lambda$ is a convex function that goes to $\infty$ at $\infty$, so it has a minimizer.

\textbf{Case 1}: there exists a minimizer $\bw$ such that $\bxi_g^\top \bw + \sqrt{\|\bw\|_2^2+\sigma^2} \|\bxi_h\|_2 > 0$.

In that case, there exists a neighborhood $O_{\bw}$ of $\bw$ such that for all $\bw' \in O_{\bw}$, $a(\bw'):= \bxi_g^\top \bw' + \sqrt{\|\bw'\|_2^2+\sigma^2} \|\bxi_h\|_2>0$. 
Thus for all $\bw' \in O_{\bw},L_\lambda(\bw') = \frac{1}{2}a(\bw')^2 + \frac{\lambda}{2} \|\bw' + \bbeta\|_2^2$ is strictly convex, because $a(\bw')$ is strictly convex (due to its first argument being linear and its second argument being positively quadratic) 
and remains positive on $O_{\bw}$ and $x>0 \mapsto x^2$ is strictly increasing. Hence $\bw$ is the only minimizer of $L_\lambda$.

\textbf{Case 2}: for all minimizer $\bw$ we have $\bxi_g^\top \bw + \sqrt{\|\bw\|^2+\sigma^2} \|\bxi_h\|_2 \leq 0$.

In this case, from optimality condition, any minimizer $\bw$ must satisfies $\lambda(\bw+\bbeta) = 0$, which implies $\bw = -\bbeta$ and $L_\lambda$ has a unique minimizer.
\end{proof}

Then we show that the minimizer $\bw^*(\lambda)$ is close to $\bw(\lambda)$:

\begin{lemma}\label{thm:ridge_AO_concentration}
There exists a constant $\gamma$ such that for all $\epsilon \in (0,1]$,
$$\sup_{\lambda\in[\lambda_{\min},\lambda_{\max}]}\PP\left(\exists \bw \in \R^p, \|\bw - \bw(\lambda)\|_2^2>\epsilon \ \ \text{and}\ \ L_\lambda(\bw) \leq \min_{\bv \in \R^p} L_\lambda(\bv) + \gamma \epsilon\right) = o(\epsilon),$$
where $\bm{w}(\lambda)$ is defined as in \eqref{eqn:def:w_lambda}.
\end{lemma}

\begin{proof}
The proof follows analogous to \cite[Theorem B.1]{miolane2021distribution}, which handles the case of the Lasso. For conciseness, we refer readers to their proof. Note that we can safely add the sup in front since $o(\epsilon)$  hides constants that do not depend on $\lambda$.
\end{proof}

\subsubsection{Connecting PO with AO}\label{subsubsec:connect_PO_AO}
We introduce a useful corollary that follows from the Convex Gaussian Min-Max Theorem (Theorem \ref{thm:CGMT}) and connects PO with AO.

\begin{cor}\label{cor:CGMT_single_estimator}
\ 
\begin{itemize}
    \item Let $D \subset \R^p$ be a closed set. We have for all $t \in \R$, 
    $$\PP(\min_{\bw \in D} \mcl{C}_\lambda(\bw) \leq t) \leq 2 \PP(\min_{\bw\in D} L_\lambda(\bw) \leq t).$$
    \item Let $D \subset \R^p$ be a convex closed set. We have for all $t \in \R$, 
    $$\PP(\min_{\bw \in D} \mcl{C}_\lambda(\bw) \geq t) \leq 2 \PP(\min_{\bw\in D} L_\lambda(\bw) \geq t).$$
\end{itemize}
\end{cor}
\begin{proof}
We will only prove the first point, since the second follows similarly.
Suppose that $\bX, \bz, \bg, \bh$ live on the same probability space and are independent. Let $\eps \in (0,1]$. Let $\sigma_{\max}(\bX)$ denote the largest singular value of $\bX$. By tightness we can find a constant $C>0$ such that the event
\begin{equation}\label{event:tightness}
    \{\sigma_{\max}(\bX)\leq C\sqrt{n}, \|\bz\|_2 \leq C\sqrt{n}, \|\bxi_g\|_2 \leq C\sqrt{n}, \|\bxi_h\|_2 \leq C\sqrt{n}\}
\end{equation}
has probability at least $1-\eps$. 
Let $D\subset \R^p$ be a non-empty closed set. Fix some $\bw_0 \in D$, on event \ref{event:tightness} both $\mcl{C}_\lambda(\bw_0)$ and $L_\lambda(\bw_0)$ are bounded by some constant $R$. Now for any $\bw$ such that $\mcl{C}_\lambda(\bw) \leq R$, we have $\|\bw + \bbeta\|_2^2 \leq \mcl{C}_\lambda(\bw) \leq R$. 

This means that there exists $R_1$ such that $\|\bw\|^2 \leq R_1$ on event \eqref{event:tightness}. Since this is true for all such $\bw$, then the minimum of $\mcl{C}_\lambda$ over $D$ is achieved on $D \cap B(0,R_1)$. Similarly, the minimum of $L_\lambda$ over $D$ is achieved on $D \cap B(0,R_2)$. WLOG, we can assume $R_1=R_2$. On event \ref{event:tightness} we have
$$\min_{\bw\in D}\mcl{C}_\lambda(\bw) = \min_{\bw \in D\cap B(0,R_1)}\mcl{C}_\lambda(\bw) = \min_{\bw \in D\cap B(0,R_1)} \max_{\bu \in B(0,R_3)}c_\lambda(\bw,\bu)$$
for some non-random $R_3>0$. Thus, for all $t \in R$, we have
$$\PP(\min_{\bw \in D}\mcl{C}_\lambda(\bw)\leq t) \leq \PP(\min_{\bw \in D\cap B(0,R_1)} \max_{\bu \in B(0,R_3)}c_\lambda(\bw,\bu)\leq t) + \eps,$$
and similarly
$$\PP(\min_{\bw \in D\cap B(0,R_1)} \max_{\bu \in B(0,R_3)}l_\lambda(\bw,\bu)\leq t) \leq \PP(\min_{\bw \in D}L_\lambda(\bw)\leq t) + \eps.$$
Now we know that $D\cap B(0,R_1)$ and $B(0,R_3)$ are compact, we can apply Theorem \ref{thm:CGMT} to get
$$\PP(\min_{\bw \in D} \mcl{C}_\lambda(\bw) \leq t) \leq 2 \PP(\min_{\bw\in D} L_\lambda(\bw) \leq t)+2\eps.$$

The corollary then follows from the fact that one can take $\eps$ arbitrarily small.
\end{proof}

Now for a fixed $\lambda$, consider the set $D_\lambda^\eps = \{\bw \in \R^p | \|\bw - \bw(\lambda)\|_2^2 \geq \eps\}$. Obviously this is a closed set. 
By first applying parts 1 and 2 from Corollary \ref{cor:CGMT_single_estimator} to the convex closed domain $\R^p$, then applying part 1 to the closed domain $D_\lambda^\eps$, we can show the following proposition:
\begin{lemma}\label{prop:AO_PO_set_min} For all $\epsilon \in (0,1]$,
\begin{equation*}
    \PP\left(\min_{\bw \in D_\lambda^\eps} \mcl{C}_\lambda(\bw) \leq \min_{\bw \in \R^p}\mcl{C}_\lambda(\bw) + \eps\right) \leq 2\PP\left(\min_{\bw \in D_\lambda^\eps} L_\lambda(\bw) \leq \min_{\bw \in \R^p}L_\lambda(\bw) + 3\eps\right) + o(\epsilon).
\end{equation*}
\end{lemma}

\begin{proof}
The proof of Lemma \ref{prop:AO_PO_set_min} is similar to Section C.1.1 in \cite{miolane2021distribution}, where the above was established for the Lasso.
\end{proof}

Combining Lemmas \ref{prop:AO_PO_set_min} and \ref{thm:ridge_AO_concentration}, we arrive at the following lemma, where we return back to the cost $\mcl{L}_\lambda(\bb)$:

\begin{lemma}\label{lemma:ridge_W2}
There exists a constant $\gamma>0$ such that for all $\epsilon \in (0,1]$, 
\begin{equation*}
    \sup_{\lambda \in [\lambda_{\min},\lambda_{\max}]} \PP\left(\exists \bb \in \R^p,\ \|\bb - \hat\bbeta^f(\lambda)\|_2^2 \geq \epsilon \ \ \text{and} \ \ \mcl{L}_\lambda(\bb) \leq \min\mcl{L}_\lambda+\gamma\epsilon \right) = o(\epsilon).
\end{equation*}
\end{lemma}
\begin{proof}
Let $\gamma > 0$ be from Lemma \ref{thm:ridge_AO_concentration}. We have
\begin{align*}
&\sup_{\lambda \in [\lambda_{\min},\lambda_{\max}]} \PP\left(\exists \bb \in \R^p,\ \|\bb - \hat\bbeta^f(\lambda)\|_2^2 \geq \epsilon \ \ \text{and} \ \ \mcl{L}_\lambda(\bb) \leq \min\mcl{L}_\lambda+\gamma\epsilon/3 \right) \\
=& \sup_{\lambda \in [\lambda_{\min},\lambda_{\max}]} \PP\left(\min_{\bw\in D_\lambda^\epsilon} \mcl{C}_\lambda(\bw) \leq \min\mcl{C}_\lambda+\gamma\epsilon/3 \right) \\
\leq & \sup_{\lambda \in [\lambda_{\min},\lambda_{\max}]} 2\PP\left(\min_{\bw \in D_\lambda^\eps} L_\lambda(\bw) \leq \min_{\bw \in \R^p}L_\lambda(\bw) + \gamma\eps\right) + o(\epsilon)\\
=& o(\epsilon),
\end{align*}
where the first equality comes from the definition of $\mathcal{C}_\lambda$ in \eqref{wu_optimization}, the inequality comes from Lemma \ref{prop:AO_PO_set_min} and the last equality comes from \ref{thm:ridge_AO_concentration}. Thus proved.
\end{proof}

\subsubsection{Uniform Control over $\lambda$}\label{subsubsec:uniform_control}

We now establish the result with uniform control over $\lambda$ (i.e. bring the $\sup$ in Lemma \ref{lemma:ridge_W2}) inside).

We first prove the following Lemma:

\begin{lemma}\label{lemma:ridge_loss_difference_lambda}
There exists constant $K$ such that 
\begin{equation*}
    \PP\left(\forall \lambda,\lambda' \in[\lambda_{\min},\lambda_{\max}], \mcl{L}_{\lambda'}(\hat{\bbeta}(\lambda)) \leq \mcl{L}_{\lambda'}(\hat{\bbeta}(\lambda')) + K|\lambda - \lambda'|\right) = 1-o(1).
\end{equation*}
\end{lemma}

\begin{proof}
The proof follows directly from Lemma \ref{lemma:ridge_loss_C_difference_lambda}.
\end{proof}

Then we prove the following proposition:
\begin{lemma}\label{lemma:ridge_W2_Lipschitz}
For $\lambda_1,\lambda_2 \in [\lambda_{\min},\lambda_{\max}]$, we have $\E[\|\hat\bbeta^f(\lambda_1) - \hat\bbeta^f(\lambda_2)\|_2] \leq M|\lambda_1-\lambda_2|$ for some constant $M$
\end{lemma}
\begin{proof}
First, by noting that $\tau_* = \sqrt{n}\tau_{\textrm{R}}$, $\alpha_* = \lambda\tau_*/\zeta_{\textrm{R}}$ as well as applying Lemma \ref{lemma:lipschitz_fixed_point} and Lemma \ref{holder_function_operation}, we know $\alpha_*(\lambda)$ and $\tau_*(\lambda)$ are bounded in some $[\alpha_{\min},\alpha_{\max}],[\tau_{\min},\tau_{\max}]$ respectively, and they are both continuous functions of $\lambda$.
Let $\lambda_1,\lambda_2 \in [\lambda_{\min},\lambda_{\max}]$. Next, recall from \eqref{eq:wandbeta} that $\bw(\lambda) = \hat\bbeta^f(\lambda) - \bbeta$. Thus, 
\begin{align*}
    &\E[\|\hat\bbeta^f(\lambda_1) - \hat\bbeta^f(\lambda_2)\|_2] ^ 2 
    \leq \E[\|\hat\bbeta^f(\lambda_1) - \hat\bbeta^f(\lambda_2)\|_2^2] 
    = \E[\|\bw(\lambda_1) - \bw(\lambda_2)\|_2^2]\\
    = & \sum_{i=1}^p \E\left[ \left( \frac{1}{\alpha_*(\lambda_1)/\tau_*(\lambda_1)+\lambda_1} \left(\frac{\alpha_*(\lambda_1)}{\sqrt{n}}\xi_{g,i} + \lambda_1 \beta_i\right) - 
    \frac{1}{\alpha_*(\lambda_2)/\tau_*(\lambda_2)+\lambda_2} \left(\frac{\alpha_*(\lambda_2)}{\sqrt{n}}\xi_{g,i} + \lambda_2 \beta_i\right) \right)^2\right]\\
    \leq & 2\E\left[\delta\left( \frac{\alpha_*(\lambda_1)}{\alpha_*(\lambda_1)/\tau_*(\lambda_1)+\lambda_1}\xi_{g}-\frac{\alpha_*(\lambda_2)}{\alpha_*(\lambda_2)/\tau_*(\lambda_2)+\lambda_2}\xi_{g}\right)^2\right]\\
    &\indent+ 2\|\bbeta\|_2^2\left[\left( \frac{\lambda_1}{\alpha_*(\lambda_1)/\tau_*(\lambda_1)+\lambda_1}-\frac{\lambda_2}{\alpha_*(\lambda_2)/\tau_*(\lambda_2)+\lambda_2}\right)^2
    \right],
\end{align*}
where $\xi_{g,i}$ is the $i$th entry of $\bxi_g \sim \mcn(0,\bm{I}_p)$ and $\xi_g\sim \mcn(0,1)$. For the first summand,
\begin{align*}
    &2\E\left[\delta\left( \frac{\alpha_*(\lambda_1)}{\alpha_*(\lambda_1)/\tau_*(\lambda_1)+\lambda_1}\xi_{g}-\frac{\alpha_*(\lambda_2)}{\alpha_*(\lambda_2)/\tau_*(\lambda_2)+\lambda_2}\xi_{g}\right)^2\right]\\
    =& 2\delta\left( \frac{\alpha_*(\lambda_1)}{\alpha_*(\lambda_1)/\tau_*(\lambda_1)+\lambda_1}-\frac{\alpha_*(\lambda_2)}{\alpha_*(\lambda_2)/\tau_*(\lambda_2)+\lambda_2}\right)^2\\
    =& 2\delta\left(\frac{\frac{\alpha_*(\lambda_1)\alpha_*(\lambda_2)}{\tau_*(\lambda_1)\tau_*(\lambda_2)}(\tau_*(\lambda_1)-\tau_*(\lambda_2)) + (\lambda_2\alpha_*(\lambda_1)-\lambda_1\alpha_*(\lambda_2))}{(\alpha_*(\lambda_1)/\tau_*(\lambda_1)+\lambda_1)(\alpha_*(\lambda_2)/\tau_*(\lambda_2)+\lambda_2)}\right)^2.
\end{align*}
We know both $\frac{\alpha_*(\lambda_1)\alpha_*(\lambda_2)}{\tau_*(\lambda_1)\tau_*(\lambda_2)}$ and $(\alpha_*(\lambda_1)/\tau_*(\lambda_1)+\lambda_1)(\alpha_*(\lambda_2)/\tau_*(\lambda_2)+\lambda_2)$ are bounded, $\tau_*$ is Lipschitz, and
\begin{align*}
    \lambda_2\alpha_*(\lambda_1)-\lambda_1\alpha_*(\lambda_2)
    =&\alpha_*(\lambda_1)(\lambda_2-\lambda_1) + \lambda_1(\alpha_*(\lambda_1) - \alpha_*(\lambda_2))\\
\leq&\alpha_{\max}|\lambda_2-\lambda_1| + \lambda_{\max}|\alpha_*(\lambda_1)-\alpha_*(\lambda_2)|.
\end{align*}
Thus, we know the first part is bounded by some $M_1(\lambda_1-\lambda_2)^2$. Similarly, the second part is bounded by some $M_2(\lambda_1-\lambda_2)^2$. Combining them completes the proof.
\end{proof}

We are now ready to prove Lemma \ref{lemma:marginal_single}.

\begin{proof}[proof of Lemma \ref{lemma:marginal_single}]
Let $\gamma > 0$ as given by Lemma \ref{lemma:ridge_W2}, let $K>0$ as given by Lemma \ref{lemma:ridge_loss_difference_lambda}, and let $M>0$ as given by Lemma \ref{lemma:ridge_W2_Lipschitz}.
Fix $\epsilon \in (0,1]$ and define $\epsilon' = \min\left( \frac{\gamma\epsilon}{K},\frac{\sqrt{\epsilon}}{M}\right)$. Let $k = \lceil(\lambda_{\max} - \lambda_{\min})/\epsilon'\rceil$. Further define $\lambda_i = \lambda_{\min} + i\epsilon'$ for $i=0,...,k$
By Lemma \ref{lemma:ridge_W2}, the event
\begin{equation}\label{event:ridge_W2}
    \left\{\forall i\in \{1,...,k\},\forall\bb \in \R^p,\mcl{L}_{\lambda_i}(\bb) \leq \min \mcl{L}_{\lambda_i} + \gamma\epsilon \Rightarrow \|\bb - \hat\bbeta^f(\lambda_i)\|_2^2\leq \epsilon \right\}
\end{equation}
has probability at least $1 - ko(\epsilon) = 1 - o(1)$. Therefore, on the intersection of event \ref{event:ridge_W2} and the event in Lemma \ref{lemma:ridge_loss_difference_lambda}, which has probability $1-o(1)$, we have for all $\lambda \in [\lambda_{\min},\lambda_{\max}]$,
$$\mcl{L}_{\lambda_i}(\hat\bbeta(\lambda)) \leq \min \mcl{L}_{\lambda_i} + K|\lambda-\lambda_i| \leq \min \mcl{L}_{\lambda_i} + \gamma\epsilon,$$
where $1 \leq i \leq k$ is such that $\lambda \in [\lambda_{i-1},\lambda_i]$. This implies (since we are on the event \ref{event:ridge_W2}) that $\|\hat\bbeta(\lambda) - \hat\bbeta^f(\lambda_i)\|_2^2\leq \epsilon$. 

Consider any Lipschitz function $\phi_\beta$. Since $\hat\bbeta^f(\lambda)$ is a $\alpha_*/(\sqrt{n}(\alpha_*/\tau_*+\lambda))\leq C/\sqrt{n}$-Lipschitz function of $\bxi_g \sim \mcn(\bm{0},\bm{I})$ for some constant $C$ (recall \eqref{eqn:def:w_lambda}), then by concentration of Lipschitz function of Gaussian random variables, we know the event
\begin{equation}\label{event:gaussian_lipschitz_concentration}
\left\{|\phi_\beta(\hat\bbeta^f(\lambda)) - \E[\phi_\beta(\hat\bbeta^f(\lambda))]| \leq \sqrt{\epsilon}\right\}
\end{equation}

has probability $1-o(\epsilon)$. Therefore, on the intersection of event \eqref{event:ridge_W2} and event \eqref{event:gaussian_lipschitz_concentration}, we have
\begin{align*}
&| \phi_\beta(\hat\bbeta(\lambda)) - \E[\phi_\beta(\hat\bbeta^f(\lambda))] | \\
\leq & |\phi_\beta(\hat\bbeta(\lambda)) - \phi_\beta(\hat\bbeta^f(\lambda_i))| + |\phi_\beta(\hat\bbeta^f(\lambda_i)) - \E[\phi_\beta(\hat\bbeta^f(\lambda_i))]| + | \E[\phi_\beta(\hat\bbeta^f(\lambda_i))] - \E[\phi_\beta(\hat\bbeta^f(\lambda))] |\\
\leq & \|\hat\bbeta(\lambda) - \hat\bbeta^f(\lambda_i)\|_2 + \sqrt{\epsilon} + \E[\|\hat\bbeta^f(\lambda) - \hat\bbeta^f(\lambda_i)\|_2]\\
\leq & 2\sqrt{\epsilon} + M|\lambda-\lambda_i|\\
\leq & 3 \sqrt{\epsilon}.
\end{align*}
The proof is then completed by noting that this holds for all $\lambda \in [\lambda_{\min},\lambda_{\max}]$. 
\end{proof}

\subsubsection{Proof of Lemma \ref{lemma:marginal_single_e}}\label{subsec:proof:lemma:marginal_single_e}

This is almost verbatim as the proof of Lemma \ref{lemma:marginal_single}, where we instead study the optimization for $\bu$ in \eqref{wu_optimization}. We omit the proof here for conciseness.

\section{Proof of main results}\label{sec:rest_proof} 
In this section we utilize Theorem \ref{main:thm:debiased_uniform_tilde} and some preceding supporting Lemmas to prove the remaining results in the main text, that is, Proposition \ref{main:cor:linear_combination} and Theorem \ref{method:thm:consistency}.

\subsection{Variance Parameter Estimation}\label{subsec:proof_tau_estimation}
We present the theorem that makes rigorous the consistency of \eqref{main:eqn:tau_estimation}.

\begin{thm}\label{main:thm:tau_estimation}
Consider $\hat\tau_{\rmL}^2,\hat\tau_{\rmR},\hat\tau_{\rmL\rmR}$ in \eqref{main:eqn:tau_estimation}. If Assumptions \ref{assumptions} and \ref{assumptionL} hold, then we have
\begin{equation*}
\begin{aligned}
    &\supll n \cdot |\hat\tau_{\textrm{L}}^2 - \tau_{\textrm{L}}^2| \convP 0,\\
    &\suplr n \cdot |\hat\tau_{\textrm{R}}^2 - \tau_{\textrm{R}}^2| \convP 0,\\
    &\sup_{\lambda_{\textrm{L}},\lambda_{\textrm{R}} \in [\lambda_{\min},\lambda_{\max}]} n \cdot |\hat\tau_{\rmL\rmR} - \tau_{\rmL\rmR}| \convP 0.\\
\end{aligned}
\end{equation*}
\end{thm}

The proof of Theorem \ref{main:thm:tau_estimation} crucially relies on the following Lemma, which is the counterpart to Theorem \ref{main:thm:debiased_uniform} in the context of residuals:

\begin{lemma}\label{lemma:debiased_uniform_e}
    \item Define $\hat\bu_k^d = \frac{\by - \bX\hat\bbeta_k}{1-\frac{\hat\df_k}{n}},\hat\bu_k^{f,d} = \sqrt{n}\bh_k^f$, where $(\bh_{\textrm{L}}^f,\bh_{\textrm{R}}^f) \sim \mcn(\bm{0},\bS \otimes \bI_n)$ for $\bS$ defined as in \ref{def:fixed_point}. Under Assumption \ref{assumptions2}, for any $1$-Lipschitz function $\phi_u:(\R^n)^2\rightarrow \R$,
    
    $$\suplb\left|\phi_u\left(\frac{\hat\bu_{\textrm{L}}^d}{\sqrt{n}},\frac{\hat\bu_{\textrm{R}}^d}{\sqrt{n}}\right) - \E\left[\phi_\beta\left(\frac{\hat\bu_{\textrm{L}}^{f,d}}{\sqrt{n}},\frac{\hat\bu_{\textrm{R}}^{f,d}}{\sqrt{n}}\right)\right]\right| \convP 0.$$
\end{lemma}
The proof of Lemma \ref{lemma:debiased_uniform_e} follows similar to the proof of Theorem \ref{main:thm:debiased_uniform}, thus we omit this for conciseness.

\begin{proof}[Proof of Theorem \ref{main:thm:tau_estimation}]
We know $\bh_k^f$ is a $1/\sqrt{n}$-Lipschitz function of $\sqrt{n}\bh_k^f$, where $(\sqrt{n}\bh_{\textrm{L}}^f,\sqrt{n}\bh_{\textrm{R}}^f) \sim \mcn(\bm{0},n\bS \otimes \bm{I}_n)$ and eigenvalues of $n\bS$ are all bounded (since entries of $n\bS$ are bounded from \ref{lemma:existence_fixed_point}). Also, $\E[\|\bh_k^f\|_2^2]=n\tau_{k}^2 \leq \tau_{\max}^2+\sigma_{\max}^2$.
Thus, by Lemma \ref{lemma:concentration_second_moment}, we know
$$\suplb \left\|\bT\left(\frac{\hat\bu_{\textrm{L}}^{d}}{\sqrt{n}},\frac{\hat\bu_{\textrm{R}}^{d}}{\sqrt{n}}\right) - \E\left[\bT\left(\frac{\hat\bu_{\textrm{L}}^{f,d}}{\sqrt{n}},\frac{\hat\bu_{\textrm{R}}^{f,d}}{\sqrt{n}}\right)\right]\right\|_F \convP 0,$$
where
\begin{equation}\label{def:T}
    \bT(\ba,\bb) :=\begin{pmatrix}\ba^\top\ba & \ba^\top\bb\\ \ba^\top\bb & \bb^\top\bb\end{pmatrix}.
\end{equation}
Extracting the entries of the above equation completes the proof. 
\end{proof}

\subsection{Ensembling: Proof of Proposition \ref{main:cor:linear_combination}} \label{subsec:proof:linear_combination}

\begin{proof}[Proof of Proposition \ref{main:cor:linear_combination}]
By arguments similar to the proof of Theorem \ref{main:thm:tau_estimation}, on using 
 Theorem \ref{main:thm:debiased_uniform} in conjunction with \ref{lemma:existence_fixed_point} and Lemma \ref{lemma:concentration_second_moment}, we have that 
$$\suplb \|\bT(\hat\bbeta_{\textrm{L}}^{d},\hat\bbeta_{\textrm{R}}^{d}) - \E[\bT(\hat\bbeta_{\textrm{L}}^{f,d},\hat\bbeta_{\textrm{R}}^{f,d})]\|_F \convP 0,$$
where $\bT$ is given by \ref{def:T}.
Therefore, we have
\begin{align*}
& \sup_{\lambda_{\textrm{L}},\lambda_{\textrm{R}} \in [\lambda_{\min},\lambda_{\max}]} \sup_{\alpha_{\rmL} \in [0,1]}\left|\|\tilde\bbeta_C^d\|_2^2 - \|\bbeta\|_2^2 - p\cdot \tau_C^2 \right|\\
= & \sup_{\lambda_{\textrm{L}},\lambda_{\textrm{R}} \in [\lambda_{\min},\lambda_{\max}]} \sup_{\alpha_{\rmL} \in [0,1]}\left|\|\alpha_{\textrm{L}}\hat\bbeta_{\textrm{L}}^{d} + (1-\alpha_{\textrm{L}})\hat\bbeta_{\textrm{R}}^{d}\|_2^2 - \|\bbeta\|_2^2 -  \alpha_{\textrm{L}}^2\cdot p\tau_{\rmL}^2 - (1-\alpha_{\textrm{L}})^2\cdot p\tau_{\rmR}^2 - 2\alpha_{\textrm{L}}(1-\alpha_{\textrm{L}})\cdot p\tau_{\rmL\rmR}\right|\\
\leq & \sup_{\lambda_{\textrm{L}},\lambda_{\textrm{R}} \in [\lambda_{\min},\lambda_{\max}]} \sup_{\alpha_{\rmL} \in [0,1]} \bigg[\alpha_\rmL^2 \left|\|\hat\bbeta_\rmL^d\|_2^2 - \|\bbeta\|_2^2 - p\tau_\rmL^2\right| + (1-\alpha_\rmL)^2 \left|\|\hat\bbeta_\rmR^d\|_2^2 - \|\bbeta\|_2^2 - p\tau_\rmR^2\right| \\
&+2\alpha_\rmL(1-\alpha_\rmL)\left|\langle\hat\bbeta_\rmL^d,\hat\bbeta_\rmR^d\rangle -\|\bbeta\|_2^2 - p\tau_{\rmL\rmR}\right| \bigg]\\
\leq & \sup_{\lambda_{\textrm{L}},\lambda_{\textrm{R}} \in [\lambda_{\min},\lambda_{\max}]} \sup_{\alpha_{\rmL} \in [0,1]} \left( \alpha_\rmL^2 + (1-\alpha_{\rmL})^2 + +2\alpha_\rmL(1-\alpha_\rmL) \right) \left\|\bT(\hat\bbeta_{\textrm{L}}^{d},\hat\bbeta_{\textrm{R}}^{d}) - \E[\bT(\hat\bbeta_{\textrm{L}}^{f,d},\hat\bbeta_{\textrm{R}}^{f,d})]\right\|_F\\
\leq & \sup_{\lambda_{\textrm{L}},\lambda_{\textrm{R}} \in [\lambda_{\min},\lambda_{\max}]}\left\|\bT(\hat\bbeta_{\textrm{L}}^{d},\hat\bbeta_{\textrm{R}}^{d}) -\E[\bT(\hat\bbeta_{\textrm{L}}^{f,d},\hat\bbeta_{\textrm{R}}^{f,d})]\right\|_F \convP 0,
\end{align*}
where we define
$$\tau_C^2 = \alpha_{\textrm{L}}^2\cdot \tau_{\rmL}^2 + (1-\alpha_{\textrm{L}})^2\cdot \tau_{\rmR}^2 + 2\alpha_{\textrm{L}}(1-\alpha_{\textrm{L}})\cdot \tau_{\rmL\rmR}.$$

Furthermore, from Theorem \ref{main:thm:tau_estimation},
\begin{align*}
&\sup_{\lambda_{\textrm{L}},\lambda_{\textrm{R}} \in [\lambda_{\min},\lambda_{\max}]} \sup_{\alpha_{\rmL} \in [0,1]} |p\cdot \tilde\tau_C^2 - p\cdot \tau_C^2| \\
\leq & \sup_{\lambda_{\textrm{L}},\lambda_{\textrm{R}} \in [\lambda_{\min},\lambda_{\max}]} \sup_{\alpha_{\rmL} \in [0,1]}\left( \alpha_\rmL^2 + (1-\alpha_{\rmL})^2 + +2\alpha_\rmL(1-\alpha_\rmL) \right) \left|p\cdot|\hat\tau_{\textrm{L}}^2 - \tau_{\textrm{L}}^2| + p\cdot|\hat\tau_{\textrm{R}}^2 - \tau_{\textrm{R}}^2| + p\cdot|\hat\tau_{LR} - \tau_{LR}|\right|\\
\leq & \sup_{\lambda_{\textrm{L}},\lambda_{\textrm{R}} \in [\lambda_{\min},\lambda_{\max}]} \left|p\cdot|\hat\tau_{\textrm{L}}^2 - \tau_{\textrm{L}}^2| + p\cdot|\hat\tau_{\textrm{R}}^2 - \tau_{\textrm{R}}^2| + p\cdot|\hat\tau_{LR} - \tau_{LR}|\right| \convP 0.
\end{align*}
Combining the above displays completes the proof.
\end{proof}

\subsection{Heritability Estimation: Proof of Theorem \ref{method:thm:consistency}}\label{subsec:proof:consistency}
\begin{proof}[Proof of Theorem \ref{method:thm:consistency}]

From Eqn. \ref{main:thm_eqn:alpha_optimal}, we know that $\hat\alpha_{\textrm{L}} \in [0,1]$, so Proposition \ref{main:cor:linear_combination} implies that 
\begin{equation}\label{eqn:heritability_proof_1}
    \sup_{\lambda_{\textrm{L}},\lambda_{\textrm{R}} \in [\lambda_{\min},\lambda_{\max}]} \left|\|\hat \bbeta_C^d\|_2^2  -  p\cdot \hat\tau_C^2- \sigma_\beta^2\right| \convP 0,
\end{equation}

which further implies that
\begin{equation}\label{eqn:heritability_proof_2}   \sup_{\lambda_{\textrm{L}},\lambda_{\textrm{R}} \in [\lambda_{\min},\lambda_{\max}]} \left|\frac{\|\hat \bbeta_C^d\|_2^2  -  p\cdot \hat\tau_C^2}{\sigma_\beta^2+\sigma^2}- \frac{\sigma_\beta^2}{\sigma_\beta^2+\sigma^2}\right| = \sup_{\lambda_{\textrm{L}},\lambda_{\textrm{R}} \in [\lambda_{\min},\lambda_{\max}]} \left|\frac{\|\hat \bbeta_C^d\|_2^2  -  p\cdot \hat\tau_C^2}{\sigma_\beta^2+\sigma^2}- h\right| \convP 0.
\end{equation}

Recalling the definitions of $\sigma_\beta^2,\sigma^2$ from Assumption \ref{assumptions2}, and the fact that $X_{ij}$'s are independent with mean $0$ and variance $1$, we know that
$$\widehat{\Var}(\by) \convP \sigma_\beta^2 + \sigma^2,$$
where $\widehat{\Var}(\by)$ denotes the sample variance of $\by$. 
Moreover, $\sigma_\beta^2 + \sigma^2$ is bounded by $2\sigma_{\max}^2$, and is positive (otherwise the problem becomes trivial). Thus we can safely replace the denominator in \eqref{eqn:heritability_proof_2} with $\widehat{\Var}(\by)$. The proof is then completed by the fact that clipping does not effect consistency.
\end{proof}

\section{Robustness under Accurate Covariance Estimation: Proof of Theorem \ref{method:thm:consistency_estimated_cov}}\label{sec:robustness_proof}

In this section we prove Theorem \ref{method:thm:consistency_estimated_cov}. To maintain notational consistency with the rest of the proofs, we refrain from using $\bZ$ and use $\bX$ instead. A simple derivation reveals that assuming $\bX$ with covariance $\bA^\top\bA$ is equivalent to assuming $\bX$ has independent columns but use $\bX\bA$ to denote the design matrix. For all quantities in this section. we include $\bA$ in a parenthesis or a subscript to denote the corresponding perturbed quantity if we replace $\bX$ by $\bX\bA$ for the design matrix. We first introduce a sufficient condition of Theorem \ref{method:thm:consistency_estimated_cov}.

\begin{lemma}\label{lemma:consistency_estimated_cov}
Let $\bA$ be a sequence of random matrices such that $\|\bA - \bI\|_{\op} \convP 0$. Under Assumption \ref{assumptions2}, we have
\begin{equation}\label{eqn:consistency_estimated_cov}
\begin{aligned}
\suplk \|\hat\bbeta_k(\bA) - \hat\bbeta_k\|_2 &\convP 0\\
\suplk \left|\frac{\hat\df_k(\bA)}{p} - \frac{\hat\df_k}{p}\right| &\convP 0.
\end{aligned}
\end{equation}
\end{lemma}

Lemma \ref{lemma:consistency_estimated_cov} is proved in Sections \ref{subsec:proof:lemma:consistency_estimated_cov_ridge} and \ref{subsec:proof:lemma:consistency_estimated_cov_lasso} for Ridge and Lasso, respectively. Compared to the assumptions in Theorem \ref{method:thm:consistency_estimated_cov}, besides reducing Assumption \ref{assumptions} to \ref{assumptions2} under Gaussian design and replacing assumptions about $\bSigma$, $\hat\bSigma$ with those about $\bA$, we also dropped Assumption \ref{assumptionA} as it can be proved under Gaussian design. Taking Lemma \ref{lemma:consistency_estimated_cov} momentarily, we prove Theorem \ref{method:thm:consistency_estimated_cov} below.

\begin{proof}[Proof of Theorem \ref{method:thm:consistency_estimated_cov}]
First we show that $\|\bA - \bI\|_{\op} \convP 0$ is a corollary of the assumptions in Theorem \ref{method:thm:consistency_estimated_cov}. Notice that Proposition \ref{examples_covariance_estimation} guarantees $\|\hat\bSigma-\bSigma\|_{\op} \rightarrow 0$. We have $\|\bA^{-1}-\bI\|_{\op} = \|\hat\bSigma^{1/2}\bSigma^{-1/2}-\bI\|_{\op} \leq \|\hat\bSigma^{1/2}-\bSigma^{1/2}\|_{\op} \|\bSigma^{-1/2}\|_{\op}\convP 0$ since $\|\hat\bSigma^{1/2}-\bSigma^{1/2}\|_{\op} \convP 0$ and eigenvalues of $\bSigma$ are bounded. Let $\bB = \bI - \bA^{-1}$, then the Neumann series $\sum_{k=0}^{\infty} \bB^k$ converges to $\bI$ in operator norm since $\|\bB\|_{\op} \convP 0$. Hence $\bA = (\bI - \bB)^{-1} = \sum_{k=0}^{\infty}\bB^k$ converges to $\bI$ in operator norm. In other words, $\|\bA - \bI\|_{\op} \convP 0$.

We are now left with arguing that under the assumptions of Lemma \ref{lemma:consistency_estimated_cov}, \eqref{eqn:consistency_estimated_cov} yields \eqref{method:thm:equation:consistency_estimated_cov}. Recall the definition of debiased estimators in \eqref{def:hat_beta_d}. We have the following:

\begin{itemize}
    \item $\suplk \|\hat\bbeta_k(\bA) - \hat\bbeta_k(\bA)\|_2 \convP 0$.
    \item $\frac{1}{\sqrt{n}}\|\bX-\bX\bA\|_{\op} \leq \frac{1}{\sqrt{n}}\|\bX\|_{\op} \|\bA-\bI\|_{\op} \convP 0$ from Corollary \ref{largest_singular_value}.
    \item $\frac{1}{\sqrt{n}}\|\by\|_2$ is bounded with probability $1-o(1)$.
    \item $\suplk \left|\frac{n}{n-\hat\df_k} - \frac{n}{n-\hat\df_k(\bA)}\right| \convP 0$ since $\suplk \left|\frac{\hat\df_k(\bA)}{n} - \frac{\hat\df_k}{n}\right| \convP 0$, together with the fact that $1-\frac{\hat\df_k}{n}$ is bounded below with probability $1-o(1)$ from Corollary \ref{cor:df_bounded}.
\end{itemize}
Then 
$$\suplk\|\hat\bbeta_k^d - \hat\bbeta_k^d(\bA)\|_2 \convP 0.$$

utilizing Lemma \ref{prop:sup_inequalities} and Lemma \ref{product_convergence}.

A similar proof yields
\begin{align*}
&\suplk |\hat\tau_k^2 - \hat\tau_k^2(\bA)| \convP 0\\
&\suplb |\hat\tau_{\rmL\rmR}^2 - \hat\tau_{\rmL\rmR}^2(\bA)| \convP 0.
\end{align*}
We are then left with articulating the expression of $\hat h^2$ and taking supremum over $\alpha_{\rmL} \in [0,1]$, which are straightforward using Lemma \ref{prop:sup_inequalities}, Lemma \ref{product_convergence} and proof techniques in Section \ref{subsec:proof:linear_combination} and \ref{subsec:proof:consistency}. We omit the details here.
\end{proof}

\subsection{Proof of Lemma \ref{lemma:consistency_estimated_cov}, Ridge case}\label{subsec:proof:lemma:consistency_estimated_cov_ridge}

In this section, we drop the subscript $\rmR$ for simplicity. First we show that
$$\suplk \|\hat\bbeta(\bA) - \hat\bbeta\|_2 \convP 0.$$
Recall the Ridge estimator and its perturbed version:
\begin{align*}
    \hat\bbeta &= (\frac{1}{n}\bX^\top\bX + \lambda\bI)^{-1}\frac{1}{n}\bX^\top\by\\
    \hat\bbeta(\bA) &= (\frac{1}{n}\bA^\top\bX^\top\bX\bA + \lambda\bI)^{-1}\frac{1}{n}\bA\bX^\top\by.
\end{align*}

We know both $\left\|(\frac{1}{n}\bX^\top\bX + \lambda\bI)^{-1}\right\|_{\op}$ and $\left\|(\frac{1}{n}\bA^\top\bX^\top\bX\bA + \lambda\bI)^{-1}\right\|_{\op}$ are bounded by $1/\lambda_{\min}$, uniformly for all $\lambda \in [\lambda_{\min},\lambda_{\max}]$. Further, their difference satisfies
$$
   (\frac{1}{n}\bX^\top\bX + \lambda\bI)^{-1} - (\frac{1}{n}\bA^\top\bX^\top\bX\bA+\lambda\bI)^{-1}   = (\frac{1}{n}\bX^\top\bX + \lambda\bI)^{-1} (\bA\bX^\top\bX\bA - \bX^\top\bX)  (\frac{1}{n}\bA^\top\bX^\top\bX\bA+\lambda\bI)^{-1}.
$$

Then we know $\supl\left\|(\frac{1}{n}\bX^\top\bX + \lambda\bI)^{-1} - (\frac{1}{n}\bA^\top\bX^\top\bX\bA+\lambda\bI)^{-1}\right\|_{\op} \convP 0$ by observing that $\frac{1}{\sqrt{n}}\|\bX\|_{\op}$ is bounded with probability $1-o(1)$, $\|\bI-\bA\|_{\op} \convP 0$, and applying Lemma \ref{product_convergence}. Another application of Lemma \ref{product_convergence} with additionally the fact that $\frac{1}{\sqrt{n}}\|\by\|_2$ is bounded concludes that $\suplk \|\hat\bbeta(\bA) - \hat\bbeta\|_2 \convP 0$.

Next, for $\hat\df/p$ and its perturbed counterpart, we have
\begin{align*}
    &\supl \left|\frac{\hat\df(\bA)}{p} - \frac{\hat\df}{p}\right| \\
    =& \supl\frac{1}{p}\left|\left(n-\text{Tr}\left((\frac{1}{n}\bX^\top\bX + \lambda\bI)^{-1}\frac{1}{n}\bX^\top\bX\right)\right) - \left(n-\text{Tr}\left((\frac{1}{n}\bA^\top\bX^\top\bX\bA + \lambda\bI)^{-1}\frac{1}{n}\bA^\top\bX^\top\bX\bA\right)\right)\right|\\
    =& \supl\frac{1}{p}\text{Tr}\left((\frac{1}{n}\bX^\top\bX + \lambda\bI)^{-1}\frac{1}{n}\bX^\top\bX - (\frac{1}{n}\bA^\top\bX^\top\bX\bA + \lambda\bI)^{-1}\frac{1}{n}\bA^\top\bX^\top\bX\bA\right)\\
    =& \supl\frac{\lambda}{p}\text{Tr}\left((\frac{1}{n}\bX^\top\bX + \lambda\bI)^{-1} - (\frac{1}{n}\bA^\top\bX^\top\bX\bA + \lambda\bI)^{-1}\right)\\
    \leq & \supl\frac{\lambda_{\max}}{p} \cdot p \left\|(\frac{1}{n}\bX^\top\bX + \lambda\bI)^{-1} - (\frac{1}{n}\bA^\top\bX^\top\bX\bA + \lambda\bI)^{-1}\right\|_{\op}\convP 0.
\end{align*}

where the third equality follows from $(\frac{1}{n}\bX^\top\bX + \lambda\bI)^{-1}\frac{1}{n}\bX^\top\bX = \bI - \lambda(\frac{1}{n}\bX^\top\bX + \lambda\bI)^{-1}$ and its perturbed counterpart. Hence the proof is complete.

\subsection{Proof of Lemma \ref{lemma:consistency_estimated_cov}, Lasso case}\label{subsec:proof:lemma:consistency_estimated_cov_lasso}

We drop the subscript $\rmL$ in this section for simplicity. The proof of the Lasso case is more involved. Denote 
\begin{align*}
\mcl{L}(\bb) &:= \frac{1}{2n} \|\by-\bX\bb\|_2^2 + \frac{\lambda}{\sqrt{n}}\|\bb\|_1\\
\mcl{L}_{\bA}(\bb) &:= \frac{1}{2n} \|\by-\bX\bA\bb\|_2^2 + \frac{\lambda}{\sqrt{n}}\|\bb\|_1
\end{align*}

the Lasso objective function and its perturbed counterpart. We first argue that $\mcl{L}(\hat\bbeta)$ is close to $\mcl{L}(\hat\bbeta(\bA))$ where $\hat\bbeta(\bA) := \argmin \mcl{L}_{\bA}$ is the perturbed Lasso solution (Lemma \ref{lemma:closeness_lasso_objective}), then we utilize techniques around the local stability of Lasso objective (c.f.~\cite{miolane2021distribution}, \cite{celentano2020lasso}) to prove the desired results in Sections \ref{subsubsec:proof:lemma:consistency_estimated_cov_lasso_p1} and \ref{subsubsec:proof:lemma:consistency_estimated_cov_lasso_p2}.

\begin{lemma}[Closeness of Lasso objectives]\label{lemma:closeness_lasso_objective}
Under the assumptions in Lemma \ref{lemma:consistency_estimated_cov}, we have
$$\supl|\mcl{L}(\hat\bbeta(\bA)) - \mcl{L}(\hat\bbeta)| \convP 0.$$

\end{lemma}

\begin{proof}
By optimality, $\mcl{L}(\hat\bbeta(\bA)) \geq \mcl{L}(\hat\bbeta)$ and $\mcl{L}_{\bA}(\hat\bbeta(\bA)) \leq \mcl{L}_{\bA}(\hat\bbeta)$. Further, we have
\begin{align*}
& \supl\left(\mcl{L}(\hat\bbeta(\bA)) - \mcl{L}(\hat\bbeta)\right)\\
=& \supl\left(\mcl{L}(\hat\bbeta(\bA)) - \mcl{L}_{\bA}(\hat\bbeta(\bA)) + \mcl{L}_{\bA}(\hat\bbeta(\bA)) - \mcl{L}_{\bA}(\hat\bbeta) + \mcl{L}_{\bA}(\hat\bbeta) - \mcl{L}(\hat\bbeta)\right)\\
\leq &  \supl\left(\mcl{L}(\hat\bbeta(\bA)) - \mcl{L}_{\bA}(\hat\bbeta(\bA)) + \mcl{L}_{\bA}(\hat\bbeta) - \mcl{L}(\hat\bbeta)\right)\\
\leq & \supl\left(\frac{1}{2n}\left|\|\by - \bX\hat\bbeta(\bA)\|_2^2-\|\by - \bX\bA\hat\bbeta(\bA)\|_2^2\right| + \frac{1}{2n}\left|\|\by - \bX\hat\bbeta\|_2^2-\|\by - \bX\bA\hat\bbeta\|_2^2\right|\right).
\end{align*}

We know from Lemma \ref{lemma:uvw_convergence} that $\hat\bbeta$ is uniformly bounded with probability $1-o(1)$. Further, replacing Lemma \ref{lemma:marginal_single} with \cite[Theorem 6]{celentano2020lasso}, we can use similar proof as Lemma \ref{lemma:uvw_convergence} to show that $\hat\bbeta(\bA)$ is also uniformly bounded with probability $1-o(1)$. Therefore, utilizing Lemma \ref{product_convergence},
\begin{align*}
    &\supl\frac{1}{2n}\left|\|\by - \bX\hat\bbeta\|_2^2-\|\by - \bX\bA\hat\bbeta\|_2^2\right|\\
    \leq & \supl\frac{1}{2}\cdot\frac{1}{\sqrt{n}}\|2\by + \bX\hat\bbeta+\bX\bA\hat\bbeta\|_2 \cdot \frac{1}{\sqrt{n}}\|\bX\hat\bbeta - \bX\bA\hat\bbeta\|_2
    \convP  0,
\end{align*}

and similarly $\supl\frac{1}{2n}\left|\|\by - \bX\hat\bbeta(\bA)\|_2^2-\|\by - \bX\bA\hat\bbeta(\bA)\|_2^2\right| \convP 0$. Hence, $$\supl|\mcl{L}(\hat\bbeta(\bA)) - \mcl{L}(\hat\bbeta)| \convP 0$$ as desired.
\end{proof}

\subsubsection{Closeness of Lasso solutions}\label{subsubsec:proof:lemma:consistency_estimated_cov_lasso_p1}

We introduce a sufficient Lemma for the first line of \eqref{eqn:consistency_estimated_cov}:

\begin{lemma}[Uniform local strong convexity of the Lasso objective]\label{lemma:uniform_local_convexity}
Under the assumptions in Lemma \ref{lemma:consistency_estimated_cov}, there exists a constant $C$ such that for any $\lambda$-dependent vector $\check\bbeta$,
$$\supl \mcl{L}(\check\bbeta)-\mcl{L}(\hat\bbeta) \geq \supl \min\{C\|\check\bbeta-\hat\bbeta\|_2^2,C\|\check\bbeta-\hat\bbeta\|_2\}.$$
\end{lemma}

The first line of \eqref{eqn:consistency_estimated_cov} directly follows by plugging in $\check\bbeta = \hat\bbeta(\bA)$ in Lemma \ref{lemma:uniform_local_convexity}.

\begin{proof}[Proof of Lemma \ref{lemma:uniform_local_convexity}]
Notice that this Lemma is the uniform extension of \cite[Lemma B.9]{celentano2020lasso} (their statement only involves $\|\check\bbeta-\hat\bbeta\|_2^2$ with an additional constraint, but our statement naturally follows from the proof procedure). Thus we only elaborate the necessary extensions of their proof here:

Consider the critical event in the proof:
$$\mcl{A}:= \left\{\frac{1}{\sqrt{n}}\kappa_-(\bX,n(1-\zeta/4))\geq\kappa_{\min}\right\} \cap \left\{\frac{1}{\sqrt{n}}\|\bX\|_{\op}\leq C\right\} \cap \left\{\frac{1}{n}\#\{j:|\hat t_j|\geq 1 - \Delta/2\}\leq 1 - \zeta/2\right\}$$

where $\zeta := 1 - \df/n$ as in \eqref{def:fixed_point} and $\bt := \frac{1}{\sqrt{n}\lambda}\bX^\top(\by - \bX\hat\bbeta)$ is the Lasso subgradient. We want to show that there exist constants $C, \kappa_{\min},\Delta$ such that event $\mcl{A}$ happens with probability $1-o(1)$, uniformly for $\lambda \in [\lambda_{\min},\lambda_{\max}]$.

For $\{\frac{1}{\sqrt{n}}\|\bX\|_{\op} \leq C\}$, it does not depend on $\lambda$, so Corollary \ref{largest_singular_value} guarantees the high probability.

For $\left\{\frac{1}{\sqrt{n}}\kappa_-(\bX,n(1-\zeta/4))\geq\kappa_{\min}\right\}$, we know $1-\zeta/4 \leq 1-\zeta_{\min}/4$ from \eqref{lemma:existence_fixed_point}, and we also know $\kappa_-(\bX,k) \geq \kappa_-(\bX,k')$ for any $k\leq k'$ by the definition of $\kappa_-$. Therefore the high-probability (which is proved in \cite[Section B.5.4]{celentano2020lasso}) naturally extends uniformly.

Finally, $\left\{\frac{1}{n}\#\{j:|\hat t_j|\geq 1 - \Delta/2\}\leq 1 - \zeta/2\right\}$ follows from \cite[Theorem E.5]{miolane2021distribution} by an $\epsilon$-net argument (c.f. \cite[Section E.3.4]{miolane2021distribution} and Section \ref{subsubsec:uniform_control}), which we omit here for conciseness.

Now that we have established $\mcl{A}$ happens uniformly with probability $1-o(1)$, the rest of the proof extends almost verbatim the proof in \cite[Section B.5.4]{celentano2020lasso} by taking the supremum in all inequalities. 
\end{proof}

\subsubsection{Closeness of Lasso sparsities}\label{subsubsec:proof:lemma:consistency_estimated_cov_lasso_p2}

Taking inspiration from \cite{miolane2021distribution}, we prove the convergence from above and below using different strategies. We state and prove two supporting Lemmas \ref{lemma:closeness_lasso_sparsity_perturbed} and \ref{lemma:closeness_lasso_subgradient_perturbed}. Taking those lemmas momentarily and noticing the fact that a nonzero entry in the Lasso solution guarantees a $\pm 1$ gradient, we immediately have 
$$\supl \left|\frac{\hat\df(\bA)}{p} - \frac{\df}{p} \right| \convP 0.$$
The second line of \eqref{eqn:consistency_estimated_cov} then follows by further applying Lemma \ref{lemma:df_diff}.

\begin{lemma}\label{lemma:closeness_lasso_sparsity_perturbed}
Under the assumptions in Lemma \ref{lemma:consistency_estimated_cov}, for any $\epsilon > 0$,
$$\PP\left(\forall \lambda \in [\lambda_{\min},\lambda_{\max}], \frac{1}{p} \|\hat\bbeta(\bA)\|_0 \geq \frac{\df}{p} - \epsilon\right) = 1-o(1).$$
\end{lemma}

\begin{proof}
As part of the proof in Lemma \ref{lemma:closeness_lasso_objective}, we know
\begin{equation}\label{lasso_sparsity_perturbed_event1}
    \PP\left(\forall \lambda \in [\lambda_{\min},\lambda_{\max}], \mcl{L}(\bb) \leq \mcl{L}_{\bA}(\bb) + \epsilon\right) = 1 - o(1)
\end{equation}
for any bounded $\bb$.

From \cite[Lemma F.4]{miolane2021distribution}, we know
\begin{equation}\label{lasso_sparsity_perturbed_event2}
    \supl\PP\left(\forall \bb, \mcl{L}(\bb) \leq \mcl{L}(\hat\bbeta)+\gamma\epsilon^3 \Rightarrow \frac{1}{p}\|\bb\|_0 \geq \frac{\df}{p} + \epsilon\right) = 1 - o(1)
\end{equation}
for some constant $\gamma$.

From \cite[Lemma B.12]{celentano2020lasso}, we know
\begin{equation}\label{lasso_sparsity_perturbed_event3}
    \PP\left(\forall \lambda,\lambda' \in [\lambda_{\min},\lambda_{\max}], \mcl{L}_{\lambda',\bA}(\hat\bbeta_{\lambda}(\bA)) \leq \mcl{L}_{\lambda',\bA}(\hat\bbeta_{\lambda'}(\bA)) + K|\lambda-\lambda'|\right) = 1 - o(1)
\end{equation}
for some constant $K$, where we write subscript to make the dependence on $\lambda$ explicit when necessary. 

Now let $K$ and $\gamma$ be the constants mentioned above, and let $C$ be the constant in Lemma \ref{lemma:lipschitz_fixed_point}(b) Consider any $\epsilon > 0$. Define $\epsilon' = \min\left\{\frac{\gamma\epsilon^3}{3K},\frac{\epsilon}{C+1}\right\}$. Let $k = \lceil \frac{\lambda_{\max}-\lambda_{\min}}{\epsilon'}\rceil$, and let $\lambda_i = \lambda_{\min} + i\epsilon'$. Applying a union bound on \eqref{lasso_sparsity_perturbed_event2} yields
\begin{equation}\label{lasso_sparsity_perturbed_event4}
    \PP\left(\forall i, \forall b, \mcl{L}_{\lambda_i}(\bb) \leq \mcl{L}_{\lambda_i}(\hat\bbeta_{\lambda_i})+\gamma\epsilon^3 \Rightarrow \frac{1}{p}\|\bb\|_0 \geq \frac{\df_{\lambda_i}}{p} + \epsilon\right) = 1 - o(1).
\end{equation}

Consider the intersection of events in \eqref{lasso_sparsity_perturbed_event1}, \eqref{lasso_sparsity_perturbed_event3} and \eqref{lasso_sparsity_perturbed_event4}, which has probability $1-o(1)$. For any $\lambda \in [\lambda_{\min},\lambda_{\max}]$, let $i$ be such that $\lambda \in [\lambda_i,\lambda_{i+1}]$. We have
\begin{align*}
\mcl{L}_{\lambda_i}(\hat\bbeta_{\lambda}(\bA)) &\leq \mcl{L}_{\lambda_i,\bA}(\hat\bbeta_{\lambda}(\bA)) + \frac{1}{3}\gamma\epsilon^3\\
&\leq \mcl{L}_{\lambda_i,\bA}(\hat\bbeta_{\lambda_i}(\bA)) + \frac{1}{3}\gamma\epsilon^3 + K|\lambda-\lambda'|\\
&= \mcl{L}_{\lambda_i,\bA}(\hat\bbeta_{\lambda_i}(\bA)) + \frac{2}{3}\gamma\epsilon^3\\
&\leq  \mcl{L}_{\lambda_i,\bA}(\hat\bbeta_{\lambda_i}) + \frac{2}{3}\gamma\epsilon^3\\
&\leq  \mcl{L}_{\lambda_i}(\hat\bbeta_{\lambda_i}) + \gamma\epsilon^3
\end{align*}

where the first and fourth inequalities comes from \eqref{lasso_sparsity_perturbed_event1}, the second inequality comes from \eqref{lasso_sparsity_perturbed_event3}, and the third inequality follows by optimality. Now since we are on event \eqref{lasso_sparsity_perturbed_event4}, we have
\begin{align*}
\frac{1}{p} \|\hat\bbeta_{\lambda}(\bA)\|_0 &\geq \frac{\df_{\lambda_i}}{p} - \epsilon \\
&= \frac{\df_{\lambda}}{p} - \epsilon + (\frac{\df_{\lambda_i}}{p} - \frac{\df_{\lambda}}{p})\\
&\geq \frac{\df_{\lambda}}{p} - \epsilon  - C|\lambda-\lambda'|\\
&\geq \frac{\df_{\lambda}}{p} - 2\epsilon.
\end{align*}
Hence the proof is complete upon observing that both $\epsilon$ and $\lambda$ are arbitrary.
\end{proof}

\begin{lemma}\label{lemma:closeness_lasso_subgradient_perturbed}
Denote 
\begin{equation*}
\begin{aligned}
    \hat\bt &:= \frac{1}{\sqrt{n}\lambda}\bX^\top(\by-\bX\hat\bbeta)\\
    \hat\bt(\bA) &:= \frac{1}{\sqrt{n}\lambda}\bA^\top\bX^\top(\by-\bX\bA\hat\bbeta(\bA))
\end{aligned}
\end{equation*} the Lasso subgradient and its perturbed counterpart. Under the assumptions in Lemma \ref{lemma:consistency_estimated_cov}, for any $\epsilon > 0$,
$$\PP\left(\forall \lambda \in [\lambda_{\min},\lambda_{\max}], \frac{1}{n} \# \{j:|\hat t_j(\bA)| = 1\} \leq \df + \epsilon\right) = 1-o(1)$$
\end{lemma}

\begin{proof}
We define an auxiliary loss function $\mcl{V}$ and its perturbed counterpart $\mcl{V}_{\bA}$:
\begin{equation*}
\begin{aligned}
    \mcl{V}(\bt) &= \min_{\|\bb\|_2\leq M} \left\{\frac{1}{2n}\|\bX\bb-\by\|_2^2 + \frac{\lambda}{\sqrt{n}}\bt^\top\bb\right\} := \min_{\|\bb\|_2\leq M} w(\bb,\bt)\\
    \mcl{V}_{\bA}(\bt) &= \min_{\|\bb\|_2\leq M} \left\{\frac{1}{2n}\|\bX\bA\bb-\by\|_2^2 + \frac{\lambda}{\sqrt{n}}\bt^\top\bb\right\} := \min_{\|\bb\|_2\leq M} w_{\bA}(\bb,\bt)
\end{aligned}
\end{equation*}

where $M$ is some large enough constant. For any $\epsilon > 0$, consider the following three high probability event statements:
\begin{equation}\label{lasso_gradient_perturbed_event1}
\PP\left(\forall\lambda\in[\lambda_{\min},\lambda_{\max}],\forall \|\bt\|_\infty \leq 1,\mcl{V}(\bt)\geq \mcl{V}_{\bA}(\bt) - \epsilon\right) = 1 - o(1);
\end{equation}
\begin{equation}\label{lasso_gradient_perturbed_event2}
\supl\PP\left(\forall \|\bt\|_\infty\leq 1, \mcl{V}(\bt) \geq \mcl{V}(\hat\bt)-3\gamma\epsilon^3 \Rightarrow \frac{1}{p}\#\{j:|t_j|\geq 1 - \epsilon\} \geq \frac{\df}{p} + K\epsilon\right) = 1 - o(1)
\end{equation}
for some constants $K_1$ and $\gamma$;
\begin{equation}\label{lasso_gradient_perturbed_event3}
\PP\left(\forall \lambda,\lambda' \in [\lambda_{\min},\lambda_{\max}], \mcl{V}_{\lambda',\bA}(\hat\bt_{\lambda}(\bA)) \geq \mcl{V}_{\lambda',\bA}(\hat\bt_{\lambda'}(\bA)) - K_2|\lambda-\lambda'|\right) = 1 - o(1)
\end{equation}
for some constant $K_2$.

Observe the exact same forms of statements \eqref{lasso_gradient_perturbed_event1}, \eqref{lasso_gradient_perturbed_event2}, \eqref{lasso_gradient_perturbed_event3} and statements \eqref{lasso_sparsity_perturbed_event1}, \eqref{lasso_sparsity_perturbed_event2}, \eqref{lasso_sparsity_perturbed_event3}. Hence, upon proving statements \eqref{lasso_gradient_perturbed_event1}, \eqref{lasso_gradient_perturbed_event2}, \eqref{lasso_gradient_perturbed_event3}, we can use a similar $\epsilon$-net argument as in the proof of Lemma \ref{lemma:closeness_lasso_sparsity_perturbed} to complete the proof of Lemma \ref{lemma:closeness_lasso_subgradient_perturbed}.

First we prove \eqref{lasso_gradient_perturbed_event1}. Fixing any $\|\bt\|_\infty \leq 1$ and denoting $\bb^* := \argmin_{\|\bb\|_2\leq M} w(\bb,\bt),\bb_{\bA}^* := \argmin_{\|\bb\|_2\leq M} w_{\bA}(\bb,\bt)$, we have
\begin{align*}
&\supl (\mcl{V}(\bt) - \mcl{V}_{\bA}(\bt))\\
:=& \supl (w(\bb^*,\bt)-w_{\bA}(\bb_{\bA}^*,\bt))\\
\leq & \supl (w(\bb_{\bA}^*,\bt)-w_{\bA}(\bb_{\bA}^*,\bt))\\
:=& \supl \frac{1}{2n}\left| \|\bX\bb_{\bA}^*-\by\|_2^2 - \|\bX\bA\bb_{\bA}^*-\by\|_2^2\right| \convP 0
\end{align*}

where the inequality follows from optimaility of $\bb^*$, and the convergence follows from Lemma \ref{product_convergence} and the fact that $\|\bb_{\bA}^*\|_2 \leq M$.

Writing out a symmetric argument yields the other direction (details omitted):
$$\supl (\mcl{V}_{\bA}(\bt) - \mcl{V}(\bt))\leq\supl \frac{1}{2n}\left| \|\bX\bb^*-\by\|_2^2 - \|\bX\bA\bb^*-\by\|_2^2\right| \convP 0.$$

Therefore, we know $\supl |\mcl{V}_{\bA}(\bt) - \mcl{V}(\bt)| \convP 0$, which leads to \eqref{lasso_gradient_perturbed_event1}.

Argument \eqref{lasso_gradient_perturbed_event2} is a direct corollary of \cite[Lemma E.9]{miolane2021distribution}.

We are now left with \eqref{lasso_gradient_perturbed_event3}. From KKT condition, on the event $\{\forall \lambda \in [\lambda_{\min},\lambda_{\max}],\|\hat\bbeta(\bA)\|_2\leq M\}$, which happens with probability $1-o(1)$ (see proof of Lemma \ref{lemma:closeness_lasso_objective}), we have
$$\mcl{L}_{\bA}(\hat\bbeta(\bA)) = \min_{\|\bb\|_2\leq M} \mcl{L}_{\bA}(\bb) := \min_{\|\bb\|_2\leq M} \max_{\|\bt\|_\infty \leq 1} w_{\bA}(\bb,\bt) = \max_{\|\bt\|_\infty \leq 1} \min_{\|\bb\|_2\leq M} w_{\bA}(\bb,\bt) := \max_{\|\bt\|_\infty \leq 1}\mcl{V}_{\bA}(\bt):= \mcl{\V}_{\bA}(\hat\bt(\bA))$$

where the permutation of min-max is authorized by \cite[Corollary 37.3.2]{rockafellar2015convex}.

As a result, with probability $1-o(1)$, $\forall \lambda,\lambda' \in [\lambda_{\min},\lambda_{\max}]$, 
\begin{align*}
& \mcl{V}_{\lambda',\bA}(\hat\bt_{\lambda'}(\bA))- \mcl{V}_{\lambda',\bA}(\hat\bt_{\lambda}(\bA))\\
= &\mcl{L}_{\lambda',\bA}(\hat\bbeta_{\lambda'}(\bA))-\mcl{V}_{\lambda',\bA}(\hat\bt_{\lambda}(\bA))\\
\leq &\mcl{L}_{\lambda',\bA}(\hat\bbeta_{\lambda}(\bA))-\mcl{V}_{\lambda',\bA}(\hat\bt_{\lambda}(\bA))\\
=& \mcl{L}_{\lambda',\bA}(\hat\bbeta_{\lambda}(\bA))-
\mcl{L}_{\lambda,\bA}(\hat\bbeta_{\lambda}(\bA)) + \mcl{V}_{\lambda,\bA}(\hat\bt_{\lambda}(\bA)) - \mcl{V}_{\lambda',\bA}(\hat\bt_{\lambda}(\bA))
\end{align*}
where the inequality follows from optimality of $\hat\bbeta_{\lambda'}(\bA)$. Statement \eqref{lasso_sparsity_perturbed_event3} already guarantees a high probability bound for $\mcl{L}_{\lambda',\bA}(\hat\bbeta_{\lambda}(\bA))-
\mcl{L}_{\lambda,\bA}(\hat\bbeta_{\lambda}(\bA))$. Finally, (re)denoting $\bb^*:=\argmin_{\|\bb\|_2 \leq M} w_{\lambda',\bA}(\bb,\hat\bt_{\lambda}(\bA))$, we have
\begin{align*}
& \mcl{V}_{\lambda,\bA}(\hat\bt_{\lambda}(\bA)) - \mcl{V}_{\lambda',\bA}(\hat\bt_{\lambda}(\bA))\\
:= & \min_{\|\bb\|_2 \leq M} w_{\lambda,\bA}(\bb,\hat\bt_{\lambda}(\bA)) - w_{\lambda',\bA}(\bb^*,\hat\bt_{\lambda}(\bA))\\
\leq & w_{\lambda,\bA}(\bb^*,\hat\bt_{\lambda}(\bA)) - w_{\lambda',\bA}(\bb^*,\hat\bt_{\lambda}(\bA))\\
=& \frac{|\lambda-\lambda'|}{\sqrt{n}}\bb^{*\top}\hat\bt_{\lambda}(\bA)\\
\leq & |\lambda-\lambda'| \cdot \|\bb^*\|_2 \cdot \frac{1}{\sqrt{n}}\|\hat\bt_{\lambda}(\bA)\|_2\\
\leq & M|\lambda-\lambda'|\cdot \frac{1}{n} \left\|\frac{1}{\lambda}\bA^\top\bX^\top(\by-\bX\bA\hat\bbeta_{\lambda}(\bA))\right\|_2\\
\leq & K_3 |\lambda-\lambda'|
\end{align*}
with probability $1-o(1)$ for some constant $K_3$, where the first inequality follows from optimality of $\bb^*$ and the last inequality follows from the fact that $\frac{1}{\lambda},\frac{1}{\sqrt{n}}\|\bX\|_{\op},\frac{1}{\sqrt{n}}\|\by\|_2,\|\bA\|_{\op}$ are all bounded with high probability. The proof is thus complete.
\end{proof}

\section{Auxiliary Lemmas}
This section introduces several auxiliary lemmas that we use repeatedly through our proofs.

\subsection{Convex Gaussian Minmax Theorem}\label{subsec:CGMT}

\begin{thm}\cite[Theorem 3]{thrampoulidis2015regularized}\label{thm:CGMT}
Let $S_w \subset \R^p$ and $S_u \subset \R^n$ be two compact sets and let $f:S_w \times S_u \rightarrow \R$ be a continuous function. Let $\bX = (\bX_{i,j}) \iid \mcn(0,1),\bxi_g \sim \mcn(0,\bm{I}_p), \bxi_h \sim \mcn(0,\bm{I}_n)$ be independent standard Gaussian vectors. Define
\begin{equation*}
    \begin{aligned}
    \mcl{C}^*(\bX) &= \min_{\bw \in S_w} \max_{\bu \in S_u} \bu^\top \bX \bw + f(\bw,\bu),\\
    L^*(\bxi_g,\bxi_h) &=  \min_{\bw \in S_w} \max_{\bu \in S_u} \|\bu\|_2\bxi_g^\top \bw + \|\bw\|_2\bxi_h^\top \bu + f(\bw,\bu).\\
    \end{aligned}
\end{equation*}
Then we have:
\begin{itemize}
    \item For all $t \in \R$,
    $$\PP(\mcl{C}^*(\bX) \leq t) \leq 2 \PP(L^*(\bxi_g,\bxi_h) \leq t).$$
    \item If $S_w$ and $S_u$ are convex and $f$ is convex-concave, then for all $t \in \R$, 
    $$\PP(\mcl{C}^*(\bX) \geq t) \leq 2 \PP(L^*(\bxi_g,\bxi_h) \geq t).$$
\end{itemize}
\end{thm}

\subsection{Properties of Equation System Solution}
\begin{lemma}[Lipschitzness of fixed point parameters]\label{lemma:lipschitz_fixed_point} 
Denote $(\tau_L,\tau_R,\zeta_L,\zeta_R, \rho)$ to be the unique solution to the fixed point equation systen \eqref{def:fixed_point}. There exists a constant $C$ such that
\begin{enumerate}[(a)]
    \item the mapping $\lambda_k \mapsto \tau_{k}$ is $C/\sqrt{n}$-Lipschitz.
    \item the mapping $\lambda_k \mapsto \zeta_k$ is $C$-Lipschitz.
    \item the mappings $(\lambda_{\textrm{L}},\lambda_{\textrm{R}}) \mapsto \rho, (\lambda_{\textrm{L}},\lambda_{\textrm{R}}) \mapsto \rho^\perp$ are both $C$-Lipschitz in both arguments, where $\rho^\perp = \sqrt{1-\rho^2}.$
\end{enumerate}
\end{lemma}
\begin{proof}
Note that (a) is a direct consequence of \cite{miolane2021distribution} Proposition A.3. 
(b) also follows from \cite{miolane2021distribution} on  noting that $\zeta_k = \beta_*/\tau_*$ in their notation, and on applying Lemma \ref{holder_function_operation} 
We remark that both results from \cite{miolane2021distribution} Proposition A.3 are for the case of the Lasso, but the Ridge case follows similarly.

We prove (c) here. We show that both $\rho$ and $\rho^\perp$ are Lipschitz in $\lambda_{\textrm{R}}$. The proof for $\lambda_{\textrm{L}}$ follows similarly and is therefore omitted.
From \eqref{def:bS_g} and \eqref{def:fixed_point}, we know that 
$$\rho = \frac{1}{n\tau_{L}\tau_{R}}\left(\sigma^2 + \E[\langle \hat\bbeta_{\textrm{L}}^{f}-\bbeta,\hat\bbeta_{\textrm{R}}^f - \bbeta\rangle] \right).$$
First, we know $\sigma^2$ and $\sqrt{n}\tau_{L}$ does not depend on $\lambda_{\textrm{R}}$ 
while $\sqrt{n}\tau_{R}$ is bounded and Lipschitz by \ref{lemma:existence_fixed_point} and the first statement of this Lemma.
Next, by Cauchy-Schwartz, we have
$$\E[\langle \hat\bbeta_{\textrm{L}}^{f}-\bbeta,\hat\bbeta_{\textrm{R}}^f - \bbeta\rangle]
\leq \E[\|\hat\bbeta_{\textrm{L}}^f - \bbeta\|_2] \cdot \E[\|\hat\bbeta_{\textrm{R}}^f - \bbeta\|_2],$$
where $\E[\|\hat\bbeta_{\textrm{L}}^f - \bbeta\|_2] \leq \sqrt{\E[\|\hat\bbeta_{\textrm{L}}^f - \bbeta\|_2^2]} = n\tau_{L}^2 -\sigma^2 \leq \tau_{\max}^2$ 
is bounded and does not depend on $\lambda_{\textrm{R}}$ (the equality follows from \eqref{def:bS_g} and \eqref{def:fixed_point}). Now,
\begin{align*}   \left|\E[\|\hat\bbeta_{R}^f(\lambda_1) - \bbeta\|^2]-\E[\|\hat\bbeta_{R}^f(\lambda_2) - \bbeta\|^2]\right| 
    \leq &  \E\left[\left|\|\hat\bbeta_{R}^f(\lambda_1) - \bbeta\|_2-\|\hat\bbeta_{R}^f(\lambda_2) - \bbeta\|_2\right|\right]\\
    \leq & \E[\|\hat\bbeta_{R}^f(\lambda_1) - \hat\bbeta_{R}^f(\lambda_2)\|_2]\\
    \leq & \sqrt{\E[\|\hat\bbeta_{R}^f(\lambda_1) - \hat\bbeta_{R}^f(\lambda_2)\|_2^2]}\\
    \leq & M(\lambda_1-\lambda_2)
\end{align*}
for some constant $M$, where the last line follows verbatim the proof of Lemma \ref{lemma:ridge_W2_Lipschitz}. Hence $\E[\langle \hat\bbeta_{\textrm{L}}^{f}-\bbeta_{\textrm{L}},\hat\bbeta_{\textrm{R}}^f - \bbeta_{\textrm{R}}\rangle]$ is Lipschitz in $\lambda_{\textrm{R}}$. Combining the two terms with Lemma \ref{holder_function_operation}, we conclude that $\rho$ is Lipschitz in $\lambda_{\textrm{R}}$.
For $\rho^\perp$, note the fact that $\rho^\perp = \sqrt{1-\rho^2}$, so
\begin{align*}
|\rho^\perp(\lambda_1) - \rho^\perp(\lambda_2)| &= \frac{|\rho(\lambda_1) +\rho(\lambda_2)|}{|\rho^\perp(\lambda_1) + \rho^\perp(\lambda_2)|} \cdot |\rho(\lambda_1) - \rho(\lambda_2)|\\
    &\leq \frac{2}{2\sqrt{1-\rho_{\max}^2}}|\rho(\lambda_1) - \rho(\lambda_2)|\\
    &\leq \frac{1}{\sqrt{1-\rho_{\max}^2}} \cdot M|\lambda_1-\lambda_2|.
\end{align*}
Thus $\rho^\perp$ is also Lipschitz in $\lambda_{\textrm{R}}$, concluding the proof.
\end{proof}

\subsection{Concentration of empirical second moments}

Most of our results are established for Lipschitz functions of random vectors. However we will frequently need to show the concentration of second order statistics of the form $\langle \ba^{(1)},\ba^{(2)}\rangle$. The following lemma provides the connection:

\begin{lemma}\label{lemma:concentration_second_moment}
Consider two random vectors $\ba_{\textrm{L}},\ba_{\textrm{R}} \in \R^p$, and a positive semi-definite matrix $\bS \in \R^{2 \times 2}$ 
(all possibly dependent on $\lambda_{\textrm{L}},\lambda_{\textrm{R}}$) that satisfy the following concentration guarantee:
\begin{itemize}
    \item There exist functions $\bphi_{\textrm{L}},\bphi_{\textrm{R}}$ (dependent on $\lambda_{\textrm{L}},
    \lambda_{\textrm{R}}$) that are $M/\sqrt{p}$-Lipschitz such that for Gaussian vectors $\bxi_{\textrm{L}},\bxi_{\textrm{R}} \sim \mcl{N}(0,\bS \otimes \bm{I}_p)$ and any $1$-Lipschitz functions $\phi:(\R^p)^2 \rightarrow \R$,
    \begin{equation*}
        \begin{aligned}
            \PP\left(\suplb |\phi(\ba_{\textrm{L}},\ba_{\textrm{R}}) - \E[\phi(\bphi_{\textrm{L}}(\bxi_{\textrm{L}}),\bphi_{\textrm{R}}(\bxi_{\textrm{R}}))]|>K\epsilon\right) = o(\epsilon)
        \end{aligned}
    \end{equation*}
    and
    \begin{equation*}
\E[\suplk\|\bphi_k(\bxi_k)\|_2^2]\leq C,\  k=L,R.
    \end{equation*}
    \item The parameters satisfy $M \leq CK$ and the singular values of $\bS$ are bounded by $C$ for all $\lambda \in [\lambda_{\min},\lambda_{\max}]$.
\end{itemize}
Then we have
$$\PP\left(\suplb |\langle \ba_{\textrm{L}},\ba_{\textrm{R}}\rangle - \E[\langle\bphi_{\textrm{L}}(\bxi_{\textrm{L}}),\bphi_{\textrm{R}}(\bxi_{\textrm{R}})\rangle]| > K\epsilon\right) =  o(\epsilon).$$
\end{lemma}

The proof for this Lemma follows almost verbatim from the proof of Lemma G.1 in \cite{celentano2021cad}, with Proposition \ref{prop:sup_inequalities} applied in appropriate places to make sure equalities/inequalities hold uniformly over the tuning parameter range.

\subsection{Largest singular value of random matrices}\label{sec:dependencies}
\begin{lemma}\cite[Corollary 5.35]{vershynin2010introduction}\label{lemma:largest_singular_value}
Let $\sigma_{\max}(\bX)$ be the largest singular value of the matrix $\bX \in \R^{n\times p}$ with i.i.d. $\mcn(0,1)$ entries, then
$$\PP(\sigma_{\max}(\bX) > \sqrt{n}+\sqrt{p}+t) \leq e^{-t^2/2}.$$
\end{lemma}

\begin{cor}\label{largest_singular_value}
In the setting of Lemma \ref{lemma:largest_singular_value}, 
$$\PP\left(\frac{1}{\sqrt{n}}\sigma_{\max}(\bX) > 2 + \sqrt{\delta}\right) \leq e^{-n/2} = o(1).$$
\end{cor}

\subsection{Inequalities for supremum}\label{subsec:sup_inequalities}
Our results often involve establishing that certain statements hold  uniformly over the range of tuning parameter values.
To this end, we frequently use the following results related to the properties of the supremum.

\begin{lemma}\label{prop:sup_inequalities}
\
\begin{itemize}
    \item For a random variable $X$ and functions $f_\lambda(X) = \sum_{k=1}^K f_\lambda^{(k)}(X)$ parametrized by $\lambda$, with probability $1$,
    $$\sup_\lambda f_\lambda(X) \leq \sum_{k=1}^K \sup_\lambda f_\lambda^{(k)}(X).$$
    \item For a random variable $X$ and non-negative functions $f_\lambda(X) = \prod_{k=1}^K f_\lambda^{(k)}(X)$ parametrized by $\lambda$, with probability $1$,
    $$\sup_\lambda f_\lambda(X) \leq \prod_{k=1}^K \sup_\lambda f_\lambda^{(k)}(X).$$
    \item For a random variable $X$ and a function $f_\lambda(X)$ parametrized by $\lambda$,
    $$\sup_\lambda \E[f_\lambda(X)] \leq \E [\sup_\lambda f_\lambda(X)].$$
\end{itemize}
\end{lemma}
The proof is straightforward and is therefore omitted.

\begin{lemma}\label{product_convergence}
For sequences of bounded positive random functions $\{f_n^{(1)}(\lambda)\},\{f_n^{(2)}(\lambda)\},\{g_n^{(1)}(\lambda)\},\{g_n^{(2)}(\lambda)\}$ such that $\sup_{\lambda}|f_n^{(1)}(\lambda) - f_n^{(2)}(\lambda)| \convP 0$ and $\sup_{\lambda}|g_n^{(1)}(\lambda) - g_n^{(2)}(\lambda)| \convP 0$, we have 
$$\sup_{\lambda}|f_n^{(1)}(\lambda)g_n^{(1)}(\lambda) - f_n^{(2)}(\lambda)g_n^{(2)}(\lambda)| \convP 0$$
\end{lemma}

\begin{proof}
The proof follows directly from Lemma \ref{subsec:sup_inequalities} together with the fact that
\begin{align*}
&|f_n^{(1)}g_n^{(1)} - f_n^{(1)}g_n^{(1)}| \\
= &|f_n^{(1)}g_n^{(1)}-f_n^{(1)}g_n^{(2)}+f_n^{(1)}g_n^{(2)}-f_n^{(2)}g_n^{(2)}|\\
\leq & |f_n^{(1)}|\cdot |g_n^{(1)}-g_n^{(2)}| + |g_n^{(2)}|\cdot |f_n^{(1)}-f_n^{(2)}|
\end{align*}
\end{proof}

\begin{lemma}\label{sqrt_convergence}
For sequences of positive random functions $\{f_n(\lambda)\}$,$\{g_n(\lambda)\}$ such that $\sup_\lambda|f_n(\lambda)|$ and $\sup_\lambda|g_n(\lambda)|$ are bounded below with probability $1-o(1)$ and $\sup_\lambda|f_n(\lambda)-g_n(\lambda)|\convP 0$, we have
$$\sup_\lambda|\sqrt{f_n(\lambda)}-\sqrt{g_n(\lambda)}|\convP 0.$$
\end{lemma}

\begin{proof}
\begin{align*}
    &\sup_\lambda|\sqrt{f_n(\lambda)}-\sqrt{g_n(\lambda)}|\\
    =&\sup_\lambda\left|\frac{f_n(\lambda)-g_n(\lambda)}{\sqrt{f_n(\lambda)}+\sqrt{g_n(\lambda)}}\right|\\
    \leq & \sup_\lambda |f_n(\lambda)-g_n(\lambda)| \cdot \sup_\lambda \left|\frac{1}{\sqrt{f_n(\lambda)}+\sqrt{g_n(\lambda)}}\right|\\
    \convP & \ 0.
\end{align*}
\end{proof}

\subsection{Properties on Lipschitz functions}
\begin{lemma}\label{holder_function_operation}
Suppose $f(x)$ has bounded norm in $[f_{\min},f_{\max}]$ and is Lipschitz with constant $C_f$, $g(x)$ has bounded norm in $[g_{\min},g_{\max}]$ and is Lipschitz with constant $C_g$, then
\begin{itemize}
    \item $h(x) = f(x)g(x)$ has bounded norm in $[f_{\min}g_{\min},f_{\max}g_{\max}]$, and is Lipschitz with constant $C_f g_{\max}+C_g f_{\max}$.
    \item $\frac{1}{f(x)}$ is positive and bounded in $[\frac{1}{f_{\max}},\frac{1}{f_{\min}}]$, and is Lipschitz with constant $C_f/f_{\min}^2$ in $x$
\end{itemize}
\begin{proof}
The boundedness part is trivial. For Lipschitzness part, we have
\begin{align*}
    \|h(x_1) - h(x_2)\| &= \|f(x_1)g(x_1)-f(x_2)g(x_2)\|\\
    &= \|(f(x_1)g(x_1)-f(x_2)g(x_1))+(f(x_2)g(x_1)-f(x_2)g(x_2))\|\\
    &\leq \|g(x_1)\|\|f(x_1)-f(x_2)\| + \|f(x_2)\|\|g(x_1)-g(x_2)\|\\
    &\leq g_{\max} C_f\|x_1-x_2\| + f_{\max} C_g\|x_1-x_2\|,
\end{align*}
and further,
\begin{align*}
    \|\frac{1}{f(x_1)} - \frac{1}{f(x_2)}\| &= \|\frac{f(x_2)-f(x_1)}{f(x_1)f(x_2)}\|\\
    &\leq \frac{\|f(x_2)-f(x_1)\|}{f_{\min}^2}\\
    &\leq \frac{C_f\|x_1-x_2\|}{f_{\min}^2}.
\end{align*}
\end{proof}
\end{lemma}

\begin{lemma}\label{lemma:lipschitz_function_sup}
Suppose $f(x,y)$ is $C$-Lipschitz in $y$ for all $x$, then $\sup_x f(x)$ is also $M$-Lipschitz in $y$.
\end{lemma}

\begin{proof}
Consider any $y_1,y_2$. Let $x_1 = \argmax_x f(x,y_1),x_2 = \argmax_x f(x,y_2)$, then
\begin{gather*}
    f(x_1,y_1)-f(x_2,y_2) = f(x_1,y_1)-f(x_1,y_2) + f(x_1,y_2)-f(x_2,y_2) \leq |f(x_1,y_1)-f(x_1,y_2)| \leq M|y_1-y_2|,\\
    f(x_2,y_2)-f(x_1,y_1) = f(x_2,y_2)-f(x_2,y_1) + f(x_2,y_1)-f(x_1,y_1) \leq |f(x_2,y_2)-f(x_2,y_1)| \leq M|y_1-y_2|,\\
\end{gather*}
where we have used the fact that $f(x_1,y_2)-f(x_2,y_2)\leq 0$ and $f(x_2,y_1)-f(x_1,y_1)\leq 0$ by optimality of $x_1$ and $x_2$.
Thus,
$$|\sup_x f(x,y_1) - \sup_x f(x,y_2)| = |f(x_1,y_1)-f(x_2,y_2)| \leq M|y_1-y_2|,$$
which shows $\sup_x f(x,y)$ is also $M$-Lipschitz in $y$.
\end{proof}

\section{Universality Proof}\label{sec:universality}
Here we extend our results from prior sections to the case of the general covariate and noise distributions mentioned under Assumptions \ref{assumptions} and \ref{assumptionL}.

To differentiate from the notations above, we use $\bG$ to denote a design matrix with i.i.d. $\mcn(0,1)$ entries, and $\bep^G$ to denote the noise vector with i.i.d. $\mcn(0,\sigma^2)$ entries. Unless otherwise noted, for all other quantities, we add a superscript $G$ to denote corresponding quantities that depend on $\bG$ and/or $\bep^G$. Also, our original notations from the preceding sections now refer to quantities with the general sub-Gaussian tailed design matrix and noise distribution specified by Assumptions \ref{assumptions} and \ref{assumptionL}.

Several results from prior sections assumed Gaussianity of the design or errors. A complete list of these results would be as follows: Lemmas \ref{lemma:df_diff}-\ref{lemma:lipschitz_psi_conditional}, Corollary \ref{cor:marginal_beta_g_joint}, Lemmas \ref{lemma:marginal_single}-\ref{lemma:gl_weak_lipschitz} (along with several supporting lemmas that enter the proof of Lemma \ref{lemma:marginal_single}), \ref{lemma:debiased_uniform_e}, Lemmas \ref{lemma:consistency_estimated_cov}-\ref{lemma:closeness_lasso_subgradient_perturbed}, \ref{lemma:largest_singular_value} and Corollary \ref{largest_singular_value}.
However, for many of these results,  the proof uses Gaussianity only through the use of certain other lemmas/results in the aforementioned list. Thus, to prove validity of such results beyond Gaussianity, it suffices to prove validity of the other lemmas used in their proofs. We thus identify a smaller set of results so that showing these hold under the non-Gaussian setting of Assumptions \ref{assumptions} and \ref{assumptionL} suffices for showing that all of the aforementioned results satisfy the same universality property. This smaller set is as follows: Lemmas \ref{lemma:df_diff}, \ref{lemma:marginal_both},
\ref{lemma:ridge_conditional_bw}, \ref{lemma:ridge_loss_C_difference_lambda}, Lemmas \ref{lemma:marginal_single} (along with some supporting results in this proof as outlined in Section \ref{proof:lemma:marginal_single_universality} below), \ref{lemma:marginal_single_e}, Lemmas \ref{lemma:closeness_lasso_objective}-\ref{lemma:closeness_lasso_subgradient_perturbed},  Lemma \ref{lemma:largest_singular_value} and Corollary \ref{largest_singular_value}. 
In this section, we show that these results continue to hold under Assumptions \ref{assumptions} and \ref{assumptionL} (and Assumption \ref{assumptionA} for Lemmas \ref{lemma:closeness_lasso_objective}-\ref{lemma:closeness_lasso_subgradient_perturbed}).

\subsection{Replacing Lemma \ref{lemma:largest_singular_value} and Corollary \ref{largest_singular_value}}

This follows from prior results in the literature that we quote below. Recall that, $\sup_n \max_{ij} \|X_{ij}\|_{\psi_2} < \infty$, where $\|\cdot\|_{\psi_2}$ is the Orcliz-2 norm or sub-Gaussian norm (see \cite[Section 2.1]{wellner2013weak} for the precise definition).

\begin{lemma}\cite[Theorem 4.4.5]{vershynin2018high}\label{lemma:vershynin}
Let $\sigma_{\max}(\bX)$ be the largest singular value of the matrix $\bX \in \R^{n\times p}$ with independent, mean $0$, variance $1$, uniformly sub-Gaussian entries, then
$$\PP(\sigma_{\max}(\bX) > CK(\sqrt{n}+\sqrt{p}+t)) \leq 2e^{-t^2},$$
where $C$ is an absolute constant and $K = \max_{ij}\|X_{ij}\|_{\psi_2}$
\end{lemma}

\begin{cor}\label{largest_singular_value_universality}
Under Assumption \ref{assumptions},
$$\PP\left(\frac{1}{\sqrt{n}}\sigma_{\max}(\bX) > CK(2+\sqrt{\delta})\right) \leq 2e^{-n} = o(1),$$
where $C,K$ are as in Lemma \ref{lemma:vershynin}.
\end{cor}

\subsection{Replacing Lemmas \ref{lemma:marginal_single} and \ref{lemma:marginal_single_e}}\label{proof:lemma:marginal_single_universality}
In this section, we will establish that the following more general version of Lemma \ref{lemma:marginal_single} holds. For Lemma \ref{lemma:marginal_single_e} a similar extension can be proved by using similar proof techniques as in this section, thus we omit writing the details for the latter.

\begin{lemma}[Replacing Lemma \ref{lemma:marginal_single}]\label{lemma:marginal_single_universality}
    Under Assumptions \ref{assumptions} and \ref{assumptionL}, for $k=L,R$ and any $1$-Lipschitz function $\phi_\beta:\R^p\rightarrow \R$,
\begin{equation*}
\suplk |\phi_\beta(\hat\bbeta_k) - \E[\phi_\beta(\hat\bbeta_k^{f})]| \convP 0.\\
\end{equation*}
\end{lemma}

The proof of the Lasso and Ridge cases are similar. We only present the case of the Ridge so it is easy  for readers to compare with Section \ref{subsec:proof:lemma:marginal_single}.
We introduce some useful supporting lemmas first. By recent results in \cite{han2022universality}, the below lemma follows.

\begin{lemma}\label{lemma:connect_universality_PO}
Suppose we are under Assumptions \ref{assumptions} and \ref{assumptionL}, and further assume $\bG\in \R^{n \times p}$ has i.i.d. $\mcn(0,1)$ entries. For any set $D_p \subset [-L_p/\sqrt{n},L_p/\sqrt{n}]^p$ with $L_p = K(\sqrt{\log p} + 2\sqrt{n}\|\bbeta\|_\infty)$ where $K$ is a constant that only depends on $\lambda_{\max}$, any $t \in \R, \epsilon > 0$, we have
\begin{gather*}
   \PP\left(\min_{\bw \in D_p} \mcl{C}_\lambda(\bw) > t + 3\epsilon \right) \leq \PP\left(\min_{\bw \in D_p} \mcl{C}_\lambda(\bw;\bG,\bep) > t + \epsilon \right) + o(1),\\
   \PP\left(\min_{\bw \in D_p} \mcl{C}_\lambda(\bw) < t - 3\epsilon \right) \leq \PP\left(\min_{\bw \in D_p} \mcl{C}_\lambda(\bw;\bG,\bep) < t - \epsilon \right) + o(1),
\end{gather*}
where $\mcl{C}_\lambda(\bw)$ is defined as in \eqref{wu_optimization} 
and $\mcl{C}_\lambda(\bw;\bG,\bep)$ represents $\mcl{C}_\lambda(\bw)$ where we replace $\bX$ by $\bG$ but keep the (possibly non-Gaussian) $\bep$.
\end{lemma}
\begin{proof}
The first argument is a direct consequence of \cite[Theorem 2.3]{han2022universality}, by specifically plugging in the formula for $\mcl{C}_\lambda(\bw)$ from \eqref{wu_optimization}.
    The second argument can be similarly proved by modifying the last lines of \cite[Section 4.2]{han2022universality}. We skip the proof here for conciseness.   
    Notice that the $\sqrt{n}$ in $D_p$ comes from the difference between their scaling and ours.
\end{proof}

Further, \cite[Proposition 3.3(2)]{han2022universality}, with the appropriate $\sqrt{n}$-rescaling in our case yields the following.

\begin{lemma}\label{lemma:bound_ridge_infty}
Under Assumptions \ref{assumptions} and \ref{assumptionL}, with probability $1-o(1)$,
$$\sqrt{n}\|\hat\bbeta_\rmR\|_\infty \leq K\left(\sqrt{\log p} + 2\sqrt{n}\|\bbeta\|_\infty\right),$$
where $K$ is a constant that only depends on $\lambda_{\max}$.
\end{lemma}
We note that Lemma \ref{lemma:connect_universality_PO} also holds for the Lasso case (with a different $\mcl{C}_\lambda(\bw)$), and the Lasso counterpart of Lemma \ref{lemma:bound_ridge_infty} is given by \cite[Proposition 3.7(2)]{han2022universality}.

Now we prove Lemma \ref{lemma:marginal_single_universality}. We follow the same structure as Section \ref{subsec:proof:lemma:marginal_single}.
\subsubsection{Converting the optimization problem}
This subsection requires no change.

\subsubsection{Connecting AO with SO}

\begin{lemma}[Replacing Proposition \ref{thm:ridge_AO_concentration}]\label{lemma:ridge_AO_concentration_universality}
Recall the definition of $\bw(\lambda)$ from \eqref{eqn:def:w_lambda}. There exists a constant $\gamma>0$ such that for all $\epsilon \in (0,1]$, 
$$\sup_{\lambda\in[\lambda_{\min},\lambda_{\max}]}\PP\left(\exists \bw \in \R^p, \|\bw - \bw(\lambda)\|_2^2>\epsilon \text{  and  } \|\bw\|_\infty \leq L_p/\sqrt{n} \  \ \text{and}\ \ L_\lambda(\bw) \leq \min_{\bv \in \R^p} L_\lambda(\bv) + \gamma \epsilon\right) = o(\epsilon).$$
\end{lemma}

\begin{proof}
The proof is a direct corollary of Lemma \ref{thm:ridge_AO_concentration}, since adding another condition in the event does not increase the probability.
\end{proof}

\subsubsection{Connecting PO with AO}
\begin{lemma}[Replacing Corollary \ref{cor:CGMT_single_estimator}]\label{lemma:CGMT_single_estimator_universality}
Consider the same setting as in Lemma \ref{lemma:connect_universality_PO}.
\begin{itemize}
    \item Further assuming $D_p$ is a closed set, we have for all $t \in \R, \epsilon > 0$, 
    $$\PP(\min_{\bw \in D_p} \mcl{C}_\lambda(\bw) \leq t) \leq 2 \PP(\min_{\bw\in D_p} L_\lambda(\bw) \leq t+2\epsilon) + o(1),$$
    where $L_{\lambda}$ is defined as in \eqref{w_optimization}.
    \item Further assuming $D_p$ is a closed convex set, we have for all $t \in \R,\epsilon > 0$, 
    $$\PP(\min_{\bw \in D_p} \mcl{C}_\lambda(\bw) \geq t) \leq 2 \PP(\min_{\bw\in D_p} L_\lambda(\bw) \geq t-2\epsilon) + o(1).$$
\end{itemize}
\end{lemma}

\begin{proof}
    We only prove the first argument as the second follows similarly.
    As a direct corollary of the second argument in \ref{lemma:connect_universality_PO} (by substituting $t\leftarrow t-3\epsilon$), we have    
    $$\PP(\min_{\bw \in D_p} \mcl{C}_\lambda(\bw) \leq t) \leq \PP(\min_{\bw\in D_p} \mcl{C}_\lambda(\bw;\bG,\bep) \leq t+2\epsilon) + o(1).$$
    
    Now consider the high probability event \eqref{event:tightness}. The fact that $\bz = \bep/\sigma$ is now sub-Gaussian instead of Gaussian does not change the fact that 
    we can find another $K$ that ensures this event occurs with high probability.
    Therefore, the proof of \ref{cor:CGMT_single_estimator} goes through, and we have the following inequality: for any $t \in \R$,
    $$\PP(\min_{\bw\in D_p} \mcl{C}_\lambda(\bw;\bG,\bep) \leq t) \leq 2 \PP(\min_{\bw\in D_p} L_\lambda(\bw) \leq t).$$
    Combining the two displays above finishes the proof.
\end{proof}

Now for a fixed $\lambda$, we instead consider the set $D = D_p \bigcap D_\lambda^\eps$, where $D_p \subset [-L_p/\sqrt{n},L_p/\sqrt{n}]^p$ is defined in Lemma \ref{lemma:connect_universality_PO} and $D_\lambda^\eps = \{\bw \in \R^p | \|\bw - \bw(\lambda)\|_2^2 > \eps\}$ is defined in Section \ref{subsubsec:connect_PO_AO}.

\begin{lemma}[Replacing Lemma \ref{prop:AO_PO_set_min}]\label{prop:AO_PO_set_min_universality}
For all $\epsilon \in (0,1]$,
    \begin{equation*}
        \PP\left(\min_{\bw \in D} \mcl{C}_\lambda(\bw) \leq \min_{\bw \in \R^p}\mcl{C}_\lambda(\bw) + \eps\right) \leq 2\PP\left(\min_{\bw \in D} L_\lambda(\bw) \leq \min_{\bw \in \R^p}L_\lambda(\bw) + 5\eps\right) + o(1).
    \end{equation*}
\end{lemma}
\begin{proof}
Recall from our remark in the proof of Lemma \ref{prop:AO_PO_set_min} that \cite[Section C.1.1]{miolane2021distribution} established this result for the case of the Lasso and Gaussian designs/error distributions. In their proof, their Corollaries 5.1 and B.1 played a crucial role. Now note that \cite[Corollary B.1]{miolane2021distribution} requires concentration of $\bep$, which also follows if $\bep$ is sub-Gaussian instead of Gaussian. Thus we are left to generalize their Corollary 5.1 (the analogue of this is Corollary \ref{cor:CGMT_single_estimator} in our paper) to the case of the ridge under our non-Gaussian design/error setting of Assumptions \ref{assumptions} and \ref{assumptionL}. We achieve this via Lemma \ref{lemma:CGMT_single_estimator_universality}. With these modifications in place, a proof analogous \cite[Section C.1.1]{miolane2021distribution} still works here. Note that the extra $2\epsilon$ on the RHS comes from the extra $2\epsilon$ in Lemma \ref{lemma:CGMT_single_estimator_universality}.
\end{proof}

\begin{lemma}[Replacing Lemma \ref{lemma:ridge_W2}]\label{lemma:ridge_W2_universality}
There exists a constant $\gamma>0$ such that for all $\epsilon \in (0,1]$, 
\begin{equation*}
    \sup_{\lambda \in [\lambda_{\min},\lambda_{\max}]} \PP\left(\exists \bb \in \R^p,\ \|\bbeta - \hat\bbeta^f(\lambda)\|_2^2 > \epsilon \text{  and  } \|\bb-\bbeta\|_\infty \leq L_p/\sqrt{n} \ \ \text{and} \ \ \mcl{L}_\lambda(\bb) \leq \min\mcl{L}_\lambda+\gamma\epsilon \right) = o(\epsilon).
\end{equation*}
\end{lemma}

\begin{proof}
This is a direct corollary of Lemma \ref{lemma:ridge_AO_concentration_universality} and Lemma \ref{prop:AO_PO_set_min_universality}, with $\bw(\lambda) = \hat\bbeta^f(\lambda) - \bbeta$ from \eqref{eqn:def:w_lambda}.
\end{proof}

\subsubsection{Uniform Control over $\lambda$}\label{subsubsec:uniform_control_universality}
\begin{lemma}[Re-stating Lemma \ref{lemma:ridge_loss_difference_lambda}]\label{lemma:ridge_loss_difference_lambda_universality}
There exists constant $K$ such that 
\begin{equation*}
    \PP\left(\forall \lambda,\lambda' \in[\lambda_{\min},\lambda_{\max}], \mcl{L}_{\lambda'}(\hat{\bbeta}(\lambda)) \leq \mcl{L}_{\lambda'}(\hat{\bbeta}(\lambda')) + K|\lambda - \lambda'|\right) = 1-o(1).
\end{equation*}
\end{lemma}

\begin{proof}
The only difference in the proofs of Lemma \ref{lemma:ridge_loss_difference_lambda_universality} and \ref{lemma:ridge_loss_difference_lambda} lies in the fact that $\bz$ is now independent sub-Gaussian instead of i.i.d. Gaussian. However, since the entries of $\bz$ still have mean $0$ and variance $1$, we still have $\|\bz\|_2 \leq 2\sqrt{n}$ with high probability. The rest of the proof follows directly.
\end{proof}

Note that Lemma \ref{lemma:ridge_loss_C_difference_lambda} is equivalent to Lemma \ref{lemma:ridge_loss_difference_lambda} so universality follows directly from the above.
Lemma \ref{lemma:ridge_W2_Lipschitz} remains as is since it does not involve $\bX$ or $\bep$.
We are now ready to prove Lemma \ref{lemma:marginal_single_universality}.

\begin{proof}[proof of Lemma \ref{lemma:marginal_single_universality}]
Let $\gamma > 0$ as given by Lemma \ref{lemma:ridge_W2_universality}, let $K>0$ as given by Lemma \ref{lemma:ridge_loss_difference_lambda_universality}, and let $M>0$ as given by Lemma \ref{lemma:ridge_W2_Lipschitz}.

Fix $\epsilon \in (0,1]$ and define $\epsilon' = \min\left( \frac{\gamma\epsilon}{K},\frac{\sqrt{\epsilon}}{M}\right)$. Let $k = \lceil(\lambda_{\max} - \lambda_{\min})/\epsilon'\rceil$. Further define $\lambda_i = \lambda_{\min} + i\epsilon'$ for $i=0,...,k$.
By Lemma \ref{lemma:ridge_W2_universality}, the event
\begin{equation}\label{event:ridge_W2_universality}
    \left\{\forall i\in \{1,...,k\},\forall\bb \in \R^p,\mcl{L}_{\lambda_i}(\bb) \leq \min \mcl{L}_{\lambda_i} + \gamma\epsilon \Rightarrow \|\bb - \hat\bbeta^f(\lambda_i)\|_2^2 \leq \epsilon \text{ or } \|\bb-\bbeta\|_\infty > L_p/\sqrt{n} \right\}
\end{equation}
has probability at least $1 - ko(\epsilon) = 1 - o(1)$. Therefore, on the intersection of event \ref{event:ridge_W2_universality} and the event in Lemma \ref{lemma:ridge_loss_difference_lambda_universality}, which has probability $1-o(1)$,
we have for all $\lambda \in [\lambda_{\min},\lambda_{\max}]$,
$$\mcl{L}_{\lambda_i}(\hat\bbeta(\lambda)) \leq \min \mcl{L}_{\lambda_i} + K|\lambda-\lambda_i| \leq \min \mcl{L}_{\lambda_i} + \gamma\epsilon,$$
where $1 \leq i \leq k$ is such that $\lambda \in [\lambda_{i-1},\lambda_i]$. This implies (since we are on the event \ref{event:ridge_W2_universality}) that either $\|\hat\bbeta(\lambda) - \hat\bbeta^f(\lambda_i)\|_2^2 \leq \epsilon$ or $\|\hat\bbeta(\lambda)-\bbeta\|_\infty > L_p/\sqrt{n}$. 

However, by Lemma \ref{lemma:bound_ridge_infty}, we know that 
\begin{align*}
&\|\hat\bbeta(\lambda)-\bbeta\|_\infty \\
\leq &  \|\hat\bbeta(\lambda)\|_\infty + \|\bbeta\|_\infty \\
\leq & K\left(\sqrt{(\log p)/ n} + 2\|\bbeta\|_\infty\right) + \|\bbeta\|_\infty\\
\leq &(K+1/2)\left(\sqrt{(\log p)/ n} + 2\|\bbeta\|_\infty\right)\\
= &L_p / \sqrt{n}
\end{align*}
with probability $1-o(1)$. 
Since the probability is over the randomness in $\bX$ and $\bep$, and $K$ only depends in $\lambda_{\max}$, then the argument holds for all $\lambda \in [\lambda_{\min},\lambda_{\max}]$.

Therefore, only with probability $o(1)$, there exists $\lambda$ such that $\|\hat\bbeta(\lambda)-\bbeta\|_\infty > L_p/\sqrt{n}$. Hence, still with probability $1-o(1)$, $\|\hat\bbeta(\lambda) - \hat\bbeta^f(\lambda_i)\|_2^2 \leq \epsilon$. The rest of the proof remains the same as in the end of Section \ref{subsubsec:uniform_control}.
\end{proof}

\subsection{Replacing Lemma \ref{lemma:df_diff}}
The ridge case is rather straightforward: \cite[Theorem 3.7]{knowles2017anisotropic} only requires mean $0$, variance $1$, independence, as well as uniformly bounded $p$th moments for the entries of $\bX$, which are already satisfied by Assumptions \ref{assumptions} and \ref{assumptionL} (specifically, uniform sub-Gaussianity implies uniformly bounded $p$th moments).
As for the Lasso case, we can modify the proof of \cite[Theorem F.1]{miolane2021distribution} in the same way as we modified Section \ref{subsec:proof:lemma:marginal_single} in Section \ref{proof:lemma:marginal_single_universality}. We omit the details here for conciseness. 

\subsection{Replacing Lemma \ref{lemma:marginal_both}}
The supporting Lemma \ref{lemma:marginal_single} has already been replaced by Lemma \ref{lemma:marginal_single_universality}. For the ridge case, the Lipschitz argument needs no change. For the Lasso case, the version of Lemma \ref{lemma:marginal_single_universality} for $\alpha$-smoothed Lasso can be extended similarly, and the following Lipschitz argument requires no change. Finally, \eqref{closeness_alpha_smoothed} relies critically on  \cite[Lemma B.9]{celentano2020lasso}. The corresponding proof in \cite[Lemma B.5.4]{celentano2020lasso} has been elaborated in our proof of Lemma \ref{lemma:uniform_local_convexity}, for which we articulate the replacement in \ref{subsec:robustness_universality}.

\subsection{Replacing Lemma \ref{lemma:ridge_conditional_bw}}
\begin{lemma}[Replacing Lemma \ref{lemma:ridge_conditional_bw}]\label{lemma:ridge_conditional_bw_universality}
Recall the definition of $\hat\bg_{\textrm{L}}$ from \ref{def:hat_gl}, and recall that $L_p = K(\sqrt{\log p} + 2\sqrt{n}\|\bbeta\|_\infty)$ where $K$ is a constant that depends only on $\lambda_{\max}$. Under Assumptions \ref{assumptions} and \ref{assumptionL}, for any $\epsilon > 0$, there exists a constant $C$ such that for a $1$-Lipschitz function $\phi_w:\R^p \rightarrow \R$,
\begin{align*}
\suplb &\PP\biggl(\exists \bw \in \R^p, \left|\phi_w(\bw) - \E[\phi_w(\hat\bbeta_\rmR^f - \bbeta)|\bg_{\textrm{L}}^f = \hat\bg_{\textrm{L}}]\right| \geq \epsilon \\
&\text{ and } \|\bw\|_\infty \leq L_p/\sqrt{n} \ \text{and} \ \mcl{C}_{\lambda_{\textrm{R}}}(\bw) \leq \min \mcl{C}_{\lambda_{\textrm{R}}} + C\epsilon^2 \biggr) = o(\epsilon^2).
\end{align*}
\end{lemma}
We introduce another supporting Lemma.

\begin{lemma}\label{lemma:connect_universality_PO_minmax}
Consider Assumptions \ref{assumptions} and \ref{assumptionL}, and further assume $\bG\in \R^{n \times p}$ has i.i.d.~$\mcn(0,1)$ entries. Denote $C(\bw,\bu) = \frac{1}{n}\bw^\top\bX\bu + f(\bw,\bu)$ where $f$ is convex-concave. For any set $D_p \subset [-L_p/\sqrt{n},L_p/\sqrt{n}]^p$ with $L_p = K(\sqrt{\log p} + 2\sqrt{n}\|\bbeta\|_\infty)$ where $K$ is a constant that only depends on $\lambda_{\max}$, any $t \in \R, \epsilon > 0$, we have
\begin{gather*}
   \PP\left(\max_{\bu\in \sqrt{n}D_p}\min_{\bw \in D_p} C(\bw,\bu) > t + 3\epsilon \right) \leq \PP\left(\max_{\bu\in \sqrt{n}D_p}\min_{\bw \in D_p} C(\bw,\bu;\bG) > t + \epsilon \right) + o(1),\\
   \PP\left(\max_{\bu\in \sqrt{n}D_p}\min_{\bw \in D_p} C(\bw,\bu) < t - 3\epsilon \right) \leq \PP\left(\max_{\bu\in \sqrt{n}D_p}\min_{\bw \in D_p} C(\bw,\bu;\bG) < t - \epsilon \right) + o(1),
\end{gather*}
where $C(\bw,\bu;\bG)$ represents $C(\bw,\bu)$ with $\bX$ replaced by $\bG$.
\end{lemma}

\begin{proof}
This is a consequence of \cite[Corollary 2.6]{han2022universality}, with necessary modifications similar to what we performed in the proof of Lemma \ref{lemma:connect_universality_PO}.
Again notice the adjusted scaling in our setting.
\end{proof}

Now we prove Lemma \ref{lemma:ridge_conditional_bw_universality}.

\begin{proof}[proof of Lemma \ref{lemma:ridge_conditional_bw_universality}]
Most of the proof will follow \cite[Lemma F.4]{celentano2021cad}. In their proof, they first showed their Lemma F.2: Conditional Gordon's Inequality, which states that $\min_{\bv \in E_w}\max_{\bu\in E_u}c(\bw,\bu)$ (as defined in \eqref{wu_optimization}, leaving the dependence on $\lambda$ implicit) concentrates around $\min_{\bv \in E_w}\max_{\bu\in E_u}l_{\rmR|\rmL}(\bw,\bu)$, the "conditional Auxiliary Optimization Problem" defined as $l(\bw,\bu)$ in \eqref{w_optimization} conditioned on $\bg_{\rmL}^f=\hat\bg_L$, which necessitates $\bxi_g$ and $\bxi_h$ taking the following values (see derivation in \cite[Section L]{celentano2021cad}):
\begin{align*}
    \hat\bxi_g &= \hat\bg_{\rmL} / \tau_{\rmL}\\
    \hat\bxi_h &= -\frac{\bX(\hat\bbeta_{\rmL} - \bbeta)}{\|\hat\bbeta_{\rmL}-\bbeta\|_2} + \frac{\langle \hat\bxi_g,\hat\bbeta_{\rmL}-\bbeta\rangle}{\|\by-\bX\hat\bbeta_{\rmL}\|_2 \|\hat\bbeta_{\rmL}-\bbeta\|_2}(\by - \bX\hat\bbeta_{\rmL}).
\end{align*}
Then, they showed the concentration of $\min_{\bv \in E_w}\max_{\bu\in E_u}l_{\rmR|\rmL}(\bw,\bu)$ around its asymptotic limit. Here $E_u,E_w$ represent any convex, closed sets.

We make modifications to their proof as follows:

\begin{enumerate}
    \item Similar to how we proved Lemma \ref{lemma:CGMT_single_estimator_universality} by invoking Lemma \ref{lemma:connect_universality_PO}, we can also modify their Lemma F.2 by invoking Lemma \ref{lemma:connect_universality_PO_minmax}, and thus establish that it holds under our general Assumptions \ref{assumptions} and \ref{assumptionL}.   
    \item For proving that $\min_{\bv \in E_w}\max_{\bu\in E_u}l_{\rmR|\rmL}(\bw,\bu)$ concentrates around its asymptotic limit, no modification is necessary. This is because $l_{\rmR|\rmL}(\bw,\bu)$ does not involve $\bX$ anymore, and the fact that replacing i.i.d. Gaussianity of $\bep$ by independent uniform sub-Gaussianity preserves its concentration properties (with possibly different constants).
\end{enumerate}
\end{proof}

Now with Lemma \ref{lemma:ridge_conditional_bw_universality} replacing Lemma \ref{lemma:ridge_conditional_bw}, the rest of the proof of Lemma \ref{lemma:bridge_2} in Section \ref{subsec:proof:lemma:bridge_2} naturally extends to the universality version under Assumptions \ref{assumptions} and \ref{assumptionL} (similar to how we modified Section \ref{subsubsec:uniform_control} to Section \ref{subsubsec:uniform_control_universality}), since the probability that there exists $\lambda$ such that $\|\hat\bbeta(\lambda)-\bbeta\|_\infty > L_p/\sqrt{n}$ is $o(1)$.

\subsection{Replacing Lemmas \ref{lemma:closeness_lasso_objective}-\ref{lemma:closeness_lasso_subgradient_perturbed}}\label{subsec:robustness_universality}

For Lemma \ref{lemma:closeness_lasso_objective}, we critically used the facts that $\|\hat\bbeta\|_2$, $\|\hat\bbeta(\bA)\|_2$ are bounded with probability $1-o(1)$. The former is guaranteed by applying Lemma \ref{lemma:concentration_second_moment} on Lemma \ref{lemma:marginal_single_universality}. The latter is stated in Assumption \ref{assumptionA}(1). 

For Lemma \ref{lemma:uniform_local_convexity}, consider the critical event
$$\mcl{A}:= \left\{\frac{1}{\sqrt{n}}\kappa_-(\bX,n(1-\zeta/4))\geq\kappa_{\min}\right\} \cap \left\{\frac{1}{\sqrt{n}}\|\bX\|_{\op}\leq C\right\} \cap \left\{\frac{1}{n}\#\{j:|\hat t_j|\geq 1 - \Delta/2\}\leq 1 - \zeta/2\right\}.$$

$\left\{\frac{1}{\sqrt{n}}\kappa_-(\bX,n(1-\zeta/4))\geq\kappa_{\min}\right\}$ happens with high probability according to Assumption \ref{assumptions}(2). $\left\{\frac{1}{\sqrt{n}}\|\bX\|_{\op}\leq C\right\}$ happens with high probability by Corollary \ref{largest_singular_value_universality}. Finally, \cite[Theorem E.5]{miolane2021distribution} can be extended to the universality version in the same way as we extended Section \ref{subsec:proof:lemma:marginal_single} in Section \ref{proof:lemma:marginal_single_universality}. As a direct corollary, $\left\{\frac{1}{n}\#\{j:|\hat t_j|\geq 1 - \Delta/2\}\leq 1 - \zeta/2\right\}$ happens with high probability.

For Lemma \ref{lemma:closeness_lasso_sparsity_perturbed}, we need to extend the critical events \eqref{lasso_sparsity_perturbed_event1}, \eqref{lasso_sparsity_perturbed_event2}, \eqref{lasso_sparsity_perturbed_event3}. \eqref{lasso_sparsity_perturbed_event1} follows naturally from the proof of the extended \ref{lemma:closeness_lasso_objective}. \eqref{lasso_sparsity_perturbed_event2} is a Corollary of \cite[Theorem F.4]{miolane2021distribution}, which again can be extended to the universality version in the same way as we extended Section \ref{subsec:proof:lemma:marginal_single} in Section \ref{proof:lemma:marginal_single_universality}. Lastly \eqref{lasso_sparsity_perturbed_event3} is stated in Assumption \ref{assumptionA}(2).

For Lemma \ref{lemma:closeness_lasso_subgradient_perturbed}, consider the critical events \eqref{lasso_gradient_perturbed_event1}, \eqref{lasso_gradient_perturbed_event2}, \eqref{lasso_gradient_perturbed_event3}. Proof of \eqref{lasso_gradient_perturbed_event1} remains unchanged (given Corollary \ref{largest_singular_value_universality}). \eqref{lasso_gradient_perturbed_event2} is a Corollary of \cite[Lemma E.9]{miolane2021distribution}, which can be similarly extended. Proof of \eqref{lasso_gradient_perturbed_event3} remains unchanged as well.

\end{document}